\documentclass{lmcs}
\pdfoutput=1
\usepackage[utf8]{inputenc}

\usepackage{lastpage}
\lmcsdoi{22}{1}{22}
\lmcsheading{}{\pageref{LastPage}}{}{}%
{Jun.~07,~2024}{Mar.~10,~2026}{}

\usepackage{amssymb}
\usepackage{amsmath}
\usepackage{mathtools}
\usepackage{color}
\usepackage{stmaryrd}
\usepackage{shuffle}
\usepackage{leftidx}
\usepackage{ifthen}
\usepackage{stackengine}
\usepackage{todonotes}
\usepackage{hyperref}


\begin{document}

\title[Definable Context-Free Sets of Graphs]{Characterizations of
  Monadic Second Order\texorpdfstring{\\}{} Definable Context-Free Sets of Graphs}

\author[R.~Iosif]{Radu Iosif\lmcsorcid{0000-0003-3204-3294}}[a]
\author[F.~Zuleger]{Florian Zuleger\lmcsorcid{0000-0003-1468-8398}}[b]

\address{Univ. Grenoble Alpes, CNRS, Grenoble INP, VERIMAG, 38000, France}
\address{Institute of Logic and Computation, Technische Universit\"{a}t Wien, Austria}

\theoremstyle{defC}
\newtheorem{conjC}[thm]{Conjecture}
\theoremstyle{thmC}
\newtheorem{corC}[thm]{Corollary}

\newenvironment{block}[1]{\begin{center}({#1}) \hspace*{2mm}
\begin{minipage}{0.85\textwidth}\em}
{\end{minipage}\end{center}}

\newcommand{\emptyseq}{\varepsilon}
\newcommand{\nat}{\mathbb{N}}
\newcommand{\btw}{\mathsf{\scriptscriptstyle{btw}}}
\newcommand{\arity}{\#}
\newcommand{\arityof}[1]{{\arity{#1}}}
\newcommand{\argsof}[1]{\mathrm{args}({#1})}
\newcommand{\rankof}[1]{\mathrm{rank}({#1})}
\newcommand{\frof}[1]{\mathrm{fr}({#1})}
\newcommand{\lenof}[1]{\mathrm{len}({#1})}
\newcommand{\cardof}[1]{\mathrm{card}({#1})}
\newcommand{\sizeof}[1]{\mathrm{size}({#1})}
\newcommand{\width}[1]{\mathrm{wd}({#1})}
\newcommand{\twd}[1]{\mathrm{twd}({#1})}
\newcommand{\at}[2]{\tuple{{#1}}_{#2}}
\newcommand{\N}{\mathbb{N}}
\newcommand{\univ}{\mathsf{U}}
\newcommand{\denotation}[1]{\llbracket{#1}\rrbracket}
\newcommand{\vars}{\mathbb{X}^{\scriptscriptstyle(1)}}
\newcommand{\Vars}{\mathbb{X}^{\scriptscriptstyle(2)}}
\newcommand{\nonterminals}{\mathbb{U}}
\newcommand{\globports}{\mathbb{P}}
\newcommand{\globstates}{\mathbb{Q}}
\newcommand{\preds}{\mathbb{P}}
\newcommand{\isdef}{\stackrel{\scalebox{0.5}{$\mathsf{def}$}}{=}}
\newcommand{\iffdef}{\stackrel{\scriptscriptstyle{\mathsf{def}}}{\iff}}
\newcommand{\interv}[2]{[{#1},{#2}]}
\newcommand{\tuple}[1]{\langle {#1} \rangle}
\newcommand{\Tuple}[1]{\left\langle {#1} \right\rangle}
\newcommand{\set}[1]{\{ {#1} \}}
\newcommand{\mset}[1]{[\![ {#1} ]\!]}
\newcommand{\Set}[1]{\left\{ {#1} \right\}}
\newcommand{\pow}[1]{\mathrm{pow}({#1})}
\newcommand{\mpow}[1]{\mathrm{mpow}({#1})}
\newcommand{\dom}[1]{\mathrm{dom}({#1})}
\newcommand{\img}[1]{\mathrm{img}({#1})}
\newcommand{\domone}[1]{\mathrm{dom}^1({#1})}
\newcommand{\uindex}{u}
\newcommand{\mleq}{\preceq}
\newcommand{\finsubseteq}{\subseteq_{\mathit{fin}}}
\newcommand{\notof}[1]{\mathrm{not}(#1)}

\newcommand{\bigO}{\mathcal{O}}
\newcommand{\twoexptime}{$2\mathsf{EXP}$}
\newcommand{\threeexptime}{$3\mathsf{EXP}$}
\newcommand{\fourexptime}{$4\mathsf{EXP}$}
\newcommand{\np}{$\mathsf{NP}$}
\newcommand{\ptime}{$\mathsf{PTIME}$}
\newcommand{\conp}{$\mathsf{co}$-$\mathsf{NP}$}
\newcommand{\exptime}{$\mathsf{EXP}$}
\newcommand{\pspace}{$\mathsf{PSPACE}$}

\newcommand{\klprob}[4]{\mathsf{P}^{({#3},{#4})}[{#1},{#2}]}
\newcommand{\sat}[2]{\mathsf{Sat}[{#1},{#2}]}
\newcommand{\klsat}[4]{\mathsf{Sat}^{({#3},{#4})}[{#1},{#2}]}
\newcommand{\tight}[2]{\mathsf{Tight}[{#1},{#2}]}
\newcommand{\kltight}[4]{\mathsf{Tight}^{({#3},{#4})}[{#1},{#2}]}
\newcommand{\loose}[2]{\mathsf{Loose}[{#1},{#2}]}
\newcommand{\klloose}[4]{\mathsf{Loose}^{({#3},{#4})}[{#1},{#2}]}
\newcommand{\bound}[2]{\mathsf{Bnd}[{#1},{#2}]}
\newcommand{\klbound}[4]{\mathsf{Bnd}^{({#3},{#4})}[{#1},{#2}]}
\newcommand{\entl}[3]{\mathsf{Entl}[{#1},{#2},{#3}]}
\newcommand{\klentl}[5]{\mathsf{Entl}^{{#4},{#5}}[{#1},{#2},{#3}]}

\newcommand{\signature}{\Sigma}
\newcommand{\fsignature}{\mathcal{F}}
\newcommand{\dsignature}{\mathcal{D}}
\newcommand{\tsignature}{\mathcal{T}}
\newcommand{\hrsignature}{\fsignature_{\mathsf{HR}}}
\newcommand{\treesignature}{\fsignature_{\mathsf{tree}}}
\newcommand{\parsesignature}{\fsignature_{\mathsf{parse}}}
\newcommand{\parsesignatureof}[1]{\fsignature^{#1}_{\mathsf{parse}}}
\newcommand{\fsignatureTwo}{\mathcal{R}}
\newcommand{\relations}{\mathbb{R}}
\newcommand{\constants}{\mathbbl{C}}
\newcommand{\alphabet}{\mathbb{A}}
\newcommand{\alphabetTwo}{\mathbb{B}}
\newcommand{\alphabetParse}{\alphabetTwo_{\mathsf{parse}}}
\newcommand{\alphabetParseOf}[1]{\alphabetTwo^{#1}_{\mathsf{parse}}}
\newcommand{\alphabetDecomp}{\alphabetTwo_{\mathsf{decomp}}}
\newcommand{\relsiggraph}[1]{\relations^{#1}_\mathsf{graph}}
\newcommand{\relsigtree}[1]{\relations^{#1}_\mathsf{tree}}
\newcommand{\relsigdecomp}[1]{\relations^{#1}_\mathsf{decomp}}
\newcommand{\relof}[2]{\relations^{\!\!\scriptscriptstyle{#2}}_{#1}}
\newcommand{\monoid}{\mathcal{M}}
\newcommand{\asucc}{\mathfrak{s}}
\newcommand{\arel}{\mathsf{r}}
\newcommand{\brel}{\mathsf{b}}
\newcommand{\erel}{\mathsf{e}}
\newcommand{\srel}{\mathsf{s}}
\newcommand{\qrel}{\mathsf{q}}
\newcommand{\prel}{\mathsf{p}}
\newcommand{\acst}{\mathsf{c}}
\newcommand{\bcst}{\mathsf{b}}
\newcommand{\ecst}{\mathsf{e}}
\newcommand{\uterm}{\mathsf{u}}
\newcommand{\vterm}{\mathsf{v}}
\newcommand{\nrel}{N}
\newcommand{\ncst}{M}
\newcommand{\nrule}{R}
\newcommand{\npred}{P}

\newcommand{\posfunc}{P}
\newcommand{\profile}[1]{\mathfrak{P}_{\scriptscriptstyle{#1}}}
\newcommand{\beh}{\mathbb{B}}
\newcommand{\behof}[1]{\beh({#1})}
\newcommand{\states}{\mathcal{Q}}
\newcommand{\state}{q}
\newcommand{\initstates}{\mathcal{I}}
\newcommand{\initstate}{\iota}
\newcommand{\finstates}{\mathcal{F}}
\newcommand{\finstate}{q_{\scriptscriptstyle{\mathsf{fin}}}}
\newcommand{\statesof}[1]{\mathit{states}({#1})}
\newcommand{\ports}{\Pi}
\newcommand{\smallstate}{q}
\newcommand{\arrow}[2]{\xrightarrow{{\scriptscriptstyle #1}}_{{\scriptstyle #2}}}
\newcommand{\Arrow}[2]{\xRightarrow{{\scriptscriptstyle #1}}_{\raisebox{6pt}{\!$\scriptstyle{#2}$}}}
\newcommand{\transitionsystem}{T}
\newcommand{\statemachines}{\mathbb{B}}
\newcommand{\arun}{\theta}
\newcommand{\pre}[1]{\leftidx{^\bullet}{\!{#1}}{}}
\newcommand{\post}[1]{\leftidx{}{{#1}}{^\bullet}}
\newcommand{\prepost}[1]{\leftidx{^\bullet}{\!{#1}}{^\bullet}}
\newcommand{\reach}{\leadsto}
\newcommand{\occurs}[1]{\mathsf{Occur}({#1})}
\newcommand{\runeq}[1]{\approx_{#1}}
\newcommand{\runclass}[2]{[{#2}]_{#1}}

\newcommand{\auto}[2]{\mathcal{A}_{
    {#1}
    \ifthenelse{\equal{#2}{}}{}{,{#2}}
}}

\newcommand{\satof}[1]{{#1}^{\scriptscriptstyle\mathsf{sat}}}

\newcommand{\autsat}[2]{\mathcal{A}^{\scriptscriptstyle\mathsf{sat}}_{
    {#1}
    \ifthenelse{\equal{#2}{}}{}{,{#2}}
}}

\newcommand{\satstates}{\states^{\scriptscriptstyle\mathsf{sat}}}
\newcommand{\satinitstates}{\initstates^{\scriptscriptstyle\mathsf{sat}}}
\newcommand{\satfinstates}{\finstates^{\scriptscriptstyle\mathsf{sat}}}
\newcommand{\sattrans}{\trans^{\scriptscriptstyle\mathsf{sat}}}

\newcommand{\cutof}[1]{{#1}^{\scriptscriptstyle\mathsf{cut}}}

\newcommand{\autcut}[2]{\mathcal{A}^{\scriptscriptstyle\mathsf{cut}}_{
    {#1}
    \ifthenelse{\equal{#2}{}}{}{,{#2}}
}}

\newcommand{\cstof}[1]{{#1}^{\scriptscriptstyle\mathsf{cst}}}

\newcommand{\autcst}[2]{\mathcal{A}^{\scriptscriptstyle\mathsf{cst}}_{
    {#1}
    \ifthenelse{\equal{#2}{}}{}{,{#2}}
  }
}

\newcommand{\conststates}{\states^{\scriptscriptstyle\mathsf{cst}}}
\newcommand{\constinitstate}{\initstate^{\scriptscriptstyle\mathsf{cst}}}
\newcommand{\consttrans}{\trans^{\scriptscriptstyle\mathsf{cst}}}
\newcommand{\cutnode}[2]{{#1}^{\scriptscriptstyle{\Uparrow}}_{#2}}

\newcommand{\authpers}[2]{\mathcal{A}^{\scriptscriptstyle\mathsf{pers}}_{{#1},{#2}}}
\newcommand{\persstates}{\states^{\scriptscriptstyle\mathsf{pers}}}
\newcommand{\persinitstates}{\initstates^{\scriptscriptstyle\mathsf{pers}}}
\newcommand{\perstrans}{\trans^{\scriptscriptstyle\mathsf{pers}}}
\newcommand{\authcolor}[2]{\mathcal{A}^{\scriptscriptstyle\mathsf{color}}_{{#1},{#2}}}
\newcommand{\asidcolor}{\asid^{\scriptscriptstyle\mathsf{color}}}
\newcommand{\colorstates}{\states^{\scriptscriptstyle\mathsf{color}}}
\newcommand{\colorinitstates}{\initstates^{\scriptscriptstyle\mathsf{color}}}
\newcommand{\colortrans}{\trans^{\scriptscriptstyle\mathsf{color}}}
\newcommand{\weightof}[1]{\omega({#1})}
\newcommand{\parikh}[1]{\mathrm{Pk}({#1})}
\newcommand{\shuffleof}[1]{{#1}^\shuffle}
\newcommand{\preordof}[1]{\mathrm{pre}({#1})}

\makeatletter
\newcommand*{\da@rightarrow}{\mathchar"0\hexnumber@\symAMSa 4B }
\newcommand*{\da@leftarrow}{\mathchar"0\hexnumber@\symAMSa 4C }
\newcommand*{\xdashrightarrow}[2][]{%
  \mathrel{%
    \mathpalette{\da@xarrow{#1}{#2}{}\da@rightarrow{\,}{}}{}%
  }%
}
\newcommand{\xdashleftarrow}[2][]{%
  \mathrel{%
    \mathpalette{\da@xarrow{#1}{#2}\da@leftarrow{}{}{\,}}{}%
  }%
}
\newcommand*{\da@xarrow}[7]{%
  \sbox0{$\ifx#7\scriptstyle\scriptscriptstyle\else\scriptstyle\fi#5#1#6\m@th$}%
  \sbox2{$\ifx#7\scriptstyle\scriptscriptstyle\else\scriptstyle\fi#5#2#6\m@th$}%
  \sbox4{$#7\dabar@\m@th$}%
  \dimen@=\wd0 %
  \ifdim\wd2 >\dimen@
    \dimen@=\wd2 %
  \fi
  \count@=2 %
  \def\da@bars{\dabar@\dabar@}%
  \@whiledim\count@\wd4<\dimen@\do{%
    \advance\count@\@ne
    \expandafter\def\expandafter\da@bars\expandafter{%
      \da@bars
      \dabar@
    }%
  }%
  \mathrel{#3}%
  \mathrel{%
    \mathop{\da@bars}\limits
    \ifx\\#1\\%
    \else
      _{\copy0}%
    \fi
    \ifx\\#2\\%
    \else
      ^{\copy2}%
    \fi
  }%
  \mathrel{#4}%
}
\makeatother

\newcommand{\Open}[2]{\xdashrightarrow{{\scriptstyle #1}}_{\raisebox{0pt}{\!$\scriptstyle{#2}$}}}


\newcommand{\store}{\mathfrak{s}}
\newcommand{\struc}{\sigma}
\newcommand{\stutterof}[1]{
        {#1}^{\scriptscriptstyle{\parallel}}
}
\newcommand{\substruc}{\sqsubseteq}
\newcommand{\gstruc}[1]{|\!|{#1}|\!|}
\newcommand{\gstrucall}[1]{\gstruc{#1}^+}
\newcommand{\strucof}[1]{\mathcal{S}({#1})}
\newcommand{\cardconstr}[3]{\mathsf{card}_{{#2},{#3}}({#1})}
\newcommand{\astruc}{\mathsf{S}}
\newcommand{\adomof}[1]{\universeGeneral_{#1}}
\newcommand{\interpof}[1]{\struc_{#1}}
\newcommand{\glue}{\mathit{glue}}
\newcommand{\absglue}{\glue^\sharp}
\newcommand{\glueof}[2]{\glue({#1},{#2})}
\newcommand{\absglueof}[2]{\absglue({#1},{#2})}
\newcommand{\fgcst}[1]{\mathit{fgcst}_{#1}}
\newcommand{\absfgcst}[1]{\mathit{fgcst}^\sharp_{#1}}
\newcommand{\qrof}[1]{\mathrm{qr}({#1})}
\newcommand{\twsid}[1]{\asid({#1})}
\newcommand{\twformsid}[2]{\asid({#1},{#2})}
\newcommand{\typesof}[1]{\mathbb{F}^{#1}_{\scriptscriptstyle{\mathsf{MSO}}}}
\newcommand{\atype}{\tau}
\newcommand{\typeof}[1]{\mathrm{sort}({#1})}
\newcommand{\absof}[1]{{#1}^\sharp}
\newcommand{\tweq}[1]{\mathit{Tw}({#1})}
\newcommand{\abstweq}[1]{\mathit{Tw}^\sharp({#1})}
\newcommand{\Dom}[1]{\mathrm{Dom}({#1})}
\newcommand{\Rel}[1]{\mathrm{Rel}({#1})}
\newcommand{\ballof}[2]{\beta_{#2}({#1})}
\newcommand{\strucset}{\mathcal{S}}
\newcommand{\plusdom}{\mathit{add}_\domsymb}
\newcommand{\minusdom}{\mathit{rem}_\domsymb}
\newcommand{\plustd}[2]{\mathit{add}^{#1}_\domsymb({#2})}

\newcommand{\empstruc}{\varepsilon}
\newcommand{\disj}{\bot}
\newcommand{\ndisj}{\not\!\!\bot}
\newcommand{\statemap}{\varrho}
\newcommand{\config}{(\comps,\interacs,\statemap,\store)}
\newcommand{\aconfig}{\gamma}
\newcommand{\nodesof}[1]{\mathrm{nodes}({#1})}
\newcommand{\configg}[1]{(\struc{#1}, \store{#1}, \statemap{#1})}
\newcommand{\configset}{\mathbf{\Gamma}}
\newcommand{\configgstate}[1]{(\struc, \store, \statemap{#1})}
\newcommand{\errconfigs}{\top}
\newcommand{\ahavoc}{\mathfrak{h}}
\newcommand{\havocrule}{$\mathsf{Havoc}$}
\newcommand{\schange}{\mathfrak{c}}
\newcommand{\achange}[2]{\schange{[{#1},{#2}]}}
\newcommand{\wchange}[1]{\schange[{#1}]}
\newcommand{\degreenode}[2]{\delta_{#1}({#2})}
\newcommand{\degreeof}[1]{\degreenode{}{#1}}

\newcommand{\cia}{\textsf{CIA}}
\newcommand{\sepalg}{\Sigma}
\newcommand{\scomp}{\uplus}
\newcommand{\gcomp}{\bullet}
\newcommand{\comp}{\bullet}
\newcommand{\abscomp}{\comp^\sharp}
\newcommand{\Comp}{\raisebox{-3pt}{\scalebox{1.8}{$\bullet$}}}
\newcommand{\subconfig}{\sqsubseteq}
\newcommand{\bigcomp}{\scalebox{2}{$\comp$}}
\newcommand{\substreq}{\sqsubseteq}
\newcommand{\substr}{\sqsubset}
\newcommand{\isubstreq}{\substreq}
\newcommand{\isubstr}{\substr}
\newcommand{\iequiv}{=}
\newcommand{\unitsepalg}{E}
\newcommand{\algelement}{e}
\newcommand{\lift}[2]{{#1}\!\!\uparrow^{\scriptscriptstyle{#2}}}

\renewcommand{\smallsetminus}{\!\setminus\!}

\newcommand{\cl}{\textsf{CL}}
\newcommand{\predname}[1]{\mathsf{#1}}
\newcommand{\true}{\mathrm{true}}
\newcommand{\apred}{\predname{A}}
\newcommand{\bpred}{\predname{B}}
\newcommand{\ainit}{\apred^{\scriptscriptstyle\mathsf{init}}}
\newcommand{\binit}{\bpred^{\scriptscriptstyle\mathsf{init}}}
\newcommand{\spred}{\predname{S}}
\newcommand{\appred}{{\apred'}}
\newcommand{\bppred}{{\bpred'}}
\newcommand{\psid}{{\widetilde{\asid}}}
\newcommand{\emp}{\predname{emp}}
\newcommand{\predstate}{\predname{state}}
\newcommand{\symconfs}{\mathbb{S}}
\let\Asterisk\undefined
\newcommand{\Asterisk}{\mathop{\scalebox{1.9}{\raisebox{-0.2ex}{$\ast$}}}\hspace*{1pt}}%
\renewcommand{\vec}[1]{\mathbf #1}
\newcommand{\fv}[1]{\mathrm{fv}({#1})}
\newcommand{\spaceform}{\xi}
\newcommand{\compin}[2]{{#1}@{#2}}
\newcommand{\compact}[1]{[{#1}]}
\newcommand{\compactin}[2]{\compin{\compact{#1}}{#2}}
\newcommand{\interacn}[4]{\tuple{{#1}.\mathit{#2}, \ldots, {#3}.\mathit{#4}}}
\newcommand{\interactwo}[4]{\tuple{{#1}.\mathit{#2}, {#3}.\mathit{#4}}}
\newcommand{\formeq}[1]{\approx_{#1}}
\newcommand{\formneq}[1]{{\not\approx}_{#1}}
\newcommand{\closeq}[2]{[{#1}]^{\formeq{#2}}}

\newcommand{\cmso}{$\mathsf{CMSO}$}
\newcommand{\sol}{\textsf{SO}}
\newcommand{\fol}{\textsf{FO}}
\newcommand{\sigsol}[1]{$\mathsf{SO}({#1})$}
\newcommand{\mso}{\textsf{MSO}}
\newcommand{\mscmso}{\textsf{(C)MSO}}
\newcommand{\sigmso}[1]{$\mathsf{MSO}({#1})$}
\newcommand{\wss}{\textsf{WS1S}}
\newcommand{\solmso}{\textsf{(M)SO}}
\newcommand{\Models}{\models}

\newcommand{\seplog}{\textsf{SL}}
\newcommand{\csl}{\textsf{CSL}}
\newcommand{\slr}{\textsf{SLR}}
\newcommand{\tb}{\textsf{TWB}}
\newcommand{\tbfo}{$\mathsf{TWB^{\scriptscriptstyle{\mathsf{FOL}}}}$}
\newcommand{\tbsl}{$\mathsf{TWB^{\scriptscriptstyle{\mathsf{SLR}}}}$}
\newcommand{\gl}{\textsf{GL}}
\newcommand{\sigslr}[1]{$\mathsf{SLR}({#1})$}
\newcommand{\finmap}{\rightharpoonup_{\scriptscriptstyle\mathit{fin}}}
\newcommand{\rank}{\mathfrak{K}}
\newcommand{\heap}{\mathsf{h}}
\newcommand{\heapcomp}{\uplus}
\newcommand{\Heapcomp}{\biguplus}
\newcommand{\seq}{\dot{=}}
\newcommand{\sneq}{\dot{\neq}}
\newcommand{\slstore}{\overline{\store}}
\newcommand{\slsid}{{\overline{\asid}}}
\newcommand{\slmodels}{\Vdash}
\newcommand{\slmodelsid}{\slmodels_{\scriptscriptstyle\slsid}}
\newcommand{\dep}[1]{\preceq_{\scriptscriptstyle #1}}
\newcommand{\deps}[2]{\mathrm{dep}_{\scriptscriptstyle #1}({#2})}
\newcommand{\depson}{\leadsto}
\newcommand{\xdepson}{\dot{\leadsto}}
\newcommand{\pureform}{\pi}
\newcommand{\ppureform}{\overline{\pi}}
\newcommand{\abstruc}{\struc^\sharp}

\newcommand{\pos}[3]{\mathrm{pos}({#1},{#2},{#3})}
\newcommand{\ipos}[2]{{\mathrm{inter}}({#1},{#2})}
\newcommand{\xipos}[2]{\mathcal{Z}_{#1}({#2})}
\newcommand{\spos}[1]{{\mathrm{state}}({#1})}
\newcommand{\gaifman}[1]{\mathbb{G}({#1})}
\newcommand{\gaifimg}[1]{\eta({#1})}
\newcommand{\intermap}{\iota}
\newcommand{\annotate}[2]{\overline{#2}^{#1}_{\intermap^{#1}}}
\newcommand{\xannot}[3]{\overline{#3}^{#1}_{#2}}
\newcommand{\xiatoms}[3]{\mathcal{I}_{#1}^{#2}({#3})}
\newcommand{\xituples}[3]{\mathsf{Tuples}_{#1}^{#2}({#3})}
\newcommand{\compstate}[2]{\mathsf{CompStates}_{#1}({#2})}
\newcommand{\interparam}[2]{\mathsf{InterAtoms}_{#1}({#2})}
\newcommand{\poly}[1]{\mathsf{poly}\left({#1}\right)}

\newcommand{\plunfold}[1]{\unfold{#1}^{\circ}}
\newcommand{\unfoldrule}{\leftarrow}
\newcommand{\asid}{\Delta}
\newcommand{\csid}{\Gamma}
\newcommand{\empsid}{\asid^{\scriptscriptstyle{\mathsf{emp}}}}
\newcommand{\annotsid}{\asid^{\scriptscriptstyle{\mathsf{lab}}}}
\newcommand{\imodels}[3]{~\models^{#1\Downarrow{#3}}_{#2}}
\newcommand{\modelsid}{\models_{\scriptscriptstyle\asid}}
\newcommand{\aasid}{\asid^\sharp}
\newcommand{\arule}{\rho}
\newcommand{\ring}[2]{\predname{ring}_{{#1},{#2}}}
\newcommand{\chain}[2]{\predname{chain}_{{#1},{#2}}}
\newcommand{\chainsymconfs}{\symconfs_{\predname{chain}}}
\newcommand{\size}[1]{\mathrm{size}({#1})}
\newcommand{\instports}[2]{\mathit{prov}({#1},{#2})}
\newcommand{\usedports}[2]{\mathit{req}({#1},{#2})}
\newcommand{\symconfeq}[1]{\simeq_{#1}}
\newcommand{\symconfneq}[1]{\not\simeq_{#1}}
\newcommand{\defn}[2]{\mathrm{def}_{#1}({#2})}
\newcommand{\defs}[2]{\mathrm{def}^*_{#1}({#2})}
\newcommand{\defnof}[1]{\mathrm{def}({#1})}
\newcommand{\stageno}{\mathfrak{n}}
\newcommand{\stagemset}{\mathfrak{m}}

\newcommand{\acomm}{\predname{R}}
\newcommand{\primcomms}{\mathfrak{P}}
\newcommand{\localcomms}{\mathfrak{L}}
\newcommand{\new}{\predname{new}}
\newcommand{\delete}{\predname{delete}}
\newcommand{\connect}{\predname{connect}}
\newcommand{\disconnect}{\predname{disconnect}}
\newcommand{\predwhen}[2]{\predname{when}\;#1\;\predname{do}\;#2}
\newcommand{\predwith}[2]{\predname{with}\;#1\;\predname{do}\;#2}
\newcommand{\whencomm}[2]{\predname{when}\;#1\;\predname{do}\;#2}
\newcommand{\withcomm}[3]{\predname{with}\;{#1}:{#2}\;\predname{do}\;{#3}\;\predname{od}}
\newcommand{\withcommv}[3]{\begin{array}{l}
\predname{with}\;{#1}:{#2}\;\predname{do} \\ \hspace*{4mm} {#3}\; \predname{od}
\end{array}}
\newcommand{\assume}{\predname{assume}}
\newcommand{\when}{\predname{when}}
\newcommand{\preddo}{\predname{do}}
\newcommand{\havoc}{\predname{havoc}}
\newcommand{\predtrue}{\predname{true}}
\newcommand{\predfalse}{\predname{false}}
\newcommand{\error}{\predname{error}}
\newcommand{\predskip}{\predname{skip}}
\newcommand{\skipcomm}{\predname{skip}}
\renewcommand{\mod}{~\mathrm{mod}~}

\newcommand{\closed}[1]{\arrow{#1}{c}}
\newcommand{\interaction}{\mathcal{A}}

\newcommand{\succj}[3]{{#1}: {#2} \leadsto {#3}}
\newcommand{\abort}[2]{{#1}: {#2} ~\text{\rotatebox[origin=c]{90}{$\leadsto$}}}
\newcommand{\gfp}{\mathrm{gfp}}
\newcommand{\lfp}{\mathrm{lfp}}
\newcommand{\rbr}{{\bf ]\!]}}
\newcommand{\lbr}{{\bf [\![}}
\newcommand{\sem}[1]{{\lbr #1 \rbr}}
\newcommand{\gsem}[2]{\sem{#1}_\domsymb^{{#2}}}
\newcommand{\sidsem}[2]{\sem{{#1}}_{#2}}
\newcommand{\gsidsem}[3]{{\sidsem{#1}{#2}}_\domsymb^{{#3}}}
\newcommand{\csem}[2]{\sem{{#1}}^{\scriptscriptstyle\mathsf{c}}_{#2}}
\newcommand{\rcsem}[2]{\sem{{#1}}^{\scriptscriptstyle\mathsf{r}}_{#2}}
\newcommand{\intfusion}[1]{\mathtt{IF}({#1})}

\newcommand{\sintfusion}[2]{\widetilde{\mathtt{IF}}({#1}\ifthenelse{\equal{#2}{}}{}{,{#2}})}

\newcommand{\extfusion}[2]{\mathtt{EF}({#1},{#2})}
\newcommand{\reachfusion}[1]{\mathtt{EF}^*({#1})}
\newcommand{\ireachfusion}[1]{\mathtt{IEF}^*({#1})}
\newcommand{\minor}[1]{\mathtt{MIN}({#1})}

\newcommand{\lng}{{\langle \! \langle}}
\newcommand{\rng}{{\rangle \! \rangle}}
\newcommand{\semcomm}[1]{\lng #1 \rng}
\newcommand{\step}[1]{\Rightarrow_{\scriptscriptstyle{#1}}}
\newcommand{\diseq}{\mathfrak{d}}

\newcommand{\hoare}[3]{\{ {#1} \} ~\mathsf{#2}~ \{ {#3} \}}
\newcommand{\cond}[1]{\{ {#1} \}}
\newcommand{\modif}[1]{\predname{modif}({#1})}
\newcommand{\hinv}{(\sharp)}

\newcommand{\open}{\mathfrak{o}}
\newcommand{\looserule}{$\mathsf{Loose}$}
\newcommand{\tightrule}{$\mathsf{Tight}$}
\newcommand{\aopen}[2]{\open{[{#1},{#2}]}}
\newcommand{\interof}[1]{\Sigma[{#1}]}
\newcommand{\are}{\predname{L}}
\newcommand{\lang}[1]{\mathcal{L}({#1})}
\newcommand{\langof}[2]{\mathcal{L}_{#1}({#2})}
\newcommand{\treelangof}[2]{\mathcal{T}_{#1}({#2})}
\newcommand{\runsof}[3]{\mathcal{R}^{#1}_{~{#2}}({#3})}
\newcommand{\parcomp}{\bowtie}
\newcommand{\lpar}{\{\!\!\{}
\newcommand{\rpar}{\}\!\!\}}
\newcommand{\aenv}{\eta}
\newcommand{\havoctriple}[4]{{#1}\triangleright\lpar{#2}\rpar~{#3}~\lpar{#4}\rpar}
\newcommand{\semlang}[2]{\lng {#1} \rng({#2})}
\newcommand{\requiv}[1]{\cong_{\scriptstyle{#1}}}
\newcommand{\congrec}[2]{\cong^{\scriptscriptstyle{#1}}_{\scriptscriptstyle{#2}}}
\newcommand{\wopen}[1]{\open{[{#1}]}}
\newcommand{\proj}[2]{{#1}\!\!\downharpoonleft_{\scriptscriptstyle{#2}}}
\newcommand{\projrel}[2]{\pi_{#2}({#1})}
\newcommand{\supp}[1]{\mathrm{supp}({#1})}
\newcommand{\iatoms}[1]{\mathrm{inter}({#1})}
\newcommand{\patoms}[1]{\mathrm{preds}({#1})}
\newcommand{\atoms}[1]{\mathrm{atoms}({#1})}
\newcommand{\disabled}[2]{{#1} \dagger {#2}}
\newcommand{\excluded}[2]{{#1} \ddagger {#2}}
\newcommand{\struceq}{\simeq}
\newcommand{\strucgeq}{\succeq}

\newcommand{\inter}{$\Sigma$}
\newcommand{\ui}{$\cup$}
\newcommand{\ue}{$\subset$}
\newcommand{\congr}{$\requiv{}$}
\newcommand{\lu}{$\mathsf{LU}$}
\newcommand{\conseq}{$\mathsf{C}$}
\newcommand{\disr}{$\dagger$}
\newcommand{\supr}{\ii}
\newcommand{\ii}{$\ainterac-$}
\newcommand{\iidisr}{$\ainterac$\disr}
\newcommand{\ie}{$\ainterac+$}
\newcommand{\parr}{$\parcomp$}
\newcommand{\botr}{$\bot$}
\newcommand{\epsilonr}{$\epsilon$}
\newcommand{\front}[2]{\mathcal{F}({#1},{#2})}

\newcommand{\havocline}[4]{
  \UnaryInfC{$\havoctriple{#1}{#2}{#3}{#4}$} }
\newcommand{\binhavocline}[4]{
  \BinaryInfC{$\havoctriple{#1}{#2}{#3}{#4}$}
}
\newcommand{\trihavocline}[4]{
  \TrinaryInfC{$\havoctriple{#1}{#2}{#3}{#4}$}
}
\newcommand{\quahavocline}[4]{
  \QuaternaryInfC{$\havoctriple{#1}{#2}{#3}{#4}$}
}
\newcommand{\havocblock}[4]{
  \UnaryInfC{${#1}$}
  \noLine
  \UnaryInfC{$\triangleright~\lpar{#2}\rpar~{#3}~\lpar{#4}\rpar$}
}
\newcommand{\binhavocblock}[4]{
  \BinaryInfC{${#1}$}
  \noLine
  \UnaryInfC{$\triangleright~\lpar{#2}\rpar~{#3}~\lpar{#4}\rpar$}
}
\newcommand{\trihavocblock}[4]{
  \TrinaryInfC{${#1}$}
  \noLine
  \UnaryInfC{$\triangleright~\lpar{#2}\rpar~{#3}~\lpar{#4}\rpar$}
}
\newcommand{\havoctwolines}[4]{
  \UnaryInfC{${#1}\triangleright~\lpar{#2}\rpar$}
  \noLine
  \UnaryInfC{${#3}~\lpar{#4}\rpar$}
}
\newcommand{\binhavoctwolines}[4]{
  \BinaryInfC{${#1}\triangleright~\lpar{#2}\rpar$}
  \noLine
  \UnaryInfC{${#3}~\lpar{#4}\rpar$}
}
\newcommand{\trihavoctwolines}[4]{
  \TrinaryInfC{${#1}\triangleright~\lpar{#2}\rpar$}
  \noLine
  \UnaryInfC{${#3}~\lpar{#4}\rpar$}
}
\newcommand{\quahavoctwolines}[4]{
  \QuaternaryInfC{${#1}\triangleright~\lpar{#2}\rpar$}
  \noLine
  \UnaryInfC{${#3}~\lpar{#4}\rpar$}
}
\newcommand{\havocfourlines}[4]{
  \UnaryInfC{${#1}$}
  \noLine
  \UnaryInfC{$\triangleright~\lpar{#2}\rpar$}
  \noLine
  \UnaryInfC{${#3}$}
  \noLine
  \UnaryInfC{$\lpar{#4}\rpar$}
}
\newcommand{\binhavocfourlines}[4]{
  \BinaryInfC{${#1}$}
  \noLine
  \UnaryInfC{$\triangleright~\lpar{#2}\rpar$}
  \noLine
  \UnaryInfC{${#3}$}
  \noLine
  \UnaryInfC{$\lpar{#4}\rpar$}
}
\newcommand{\trihavocfourlines}[4]{
  \TrinaryInfC{${#1}$}
  \noLine
  \UnaryInfC{$\triangleright~\lpar{#2}\rpar$}
  \noLine
  \UnaryInfC{${#3}$}
  \noLine
  \UnaryInfC{$\lpar{#4}\rpar$}
}
\newcommand{\havocthreelines}[4]{
  \UnaryInfC{${#1}~\triangleright~\lpar{#2}\rpar$}
  \noLine
  \UnaryInfC{${#3}$}
  \noLine
  \UnaryInfC{$\lpar{#4}\rpar$}
}

\newcommand{\interofyxch}{{\Sigma_{y,x}^1}}
\newcommand{\interofyxct}{{\Sigma_{y,x}^2}}
\newcommand{\interofzy}{{\Sigma_{z,y}}}
\newcommand{\interofyv}{{\Sigma_{y,v}}}
\newcommand{\interofvxch}{{\Sigma_{v,x}^1}}
\newcommand{\interofvxct}{{\Sigma_{v,x}^2}}

\newcommand{\interoftree}[1]{\interof{ \tree({#1}) }}

\newcommand{\mycomment}[2]{\textcolor{blue}{[#1]} \textcolor{red}{#2}}

\newcommand{\tokcomp}{\mathsf{S}}
\newcommand{\tokinterac}{\mathsf{T}}
\newcommand{\toktoken}{\mathsf{T}}
\newcommand{\toknotok}{\mathsf{H}}

\newcommand{\tokin}{\textit{in}}
\newcommand{\tokout}{\textit{out}}

\newcommand{\treenode}{\mathit{Node}}
\newcommand{\treeleaf}{\mathit{Leaf}}
\newcommand{\treeroot}{\mathit{Root}}
\newcommand{\linterac}{I_\ell}
\newcommand{\rinterac}{I_r}
\newcommand{\stateidle}{\mathit{idle}}
\newcommand{\stateleft}{\mathit{left}}
\newcommand{\stateright}{\mathit{right}}
\newcommand{\statestart}{\mathit{start}}
\newcommand{\stateleafidle}{\mathit{leaf\!\_idle}}
\newcommand{\stateleafbusy}{\mathit{leaf\!\_busy}}
\newcommand{\lrecv}{\mathit{\ell\!\_rcv}}
\newcommand{\rrecv}{\mathit{r\!\_rcv}}
\newcommand{\send}{\mathit{snd}}
\newcommand{\treeseg}{\predname{tseg}}
\newcommand{\treesegrec}{\predname{tsegrec}}
\newcommand{\treeidle}{\tree_{\stateidle}}
\newcommand{\treemaybeidle}{\tree_{(\stateidle)}}
\newcommand{\treenotidle}{\tree_{\neg\stateidle}}
\newcommand{\treestar}{\tree_{\star}}



\newcommand{\inmathinf}[3]{\{\;#1\;\}\;\,#2\;\,\{\;#3\;\}}
\newcommand{\infformula}[3]{$\{\;#1\;\}\;\,#2\;\,\{\;#3\;\}$}
\newcommand{\codesize}{\footnotesize}
\newcommand{\boxaround}[1]{\ovalbox{${#1}$}}
\newcommand*\circled[1]{\tikz[baseline=(char.base)]{
            \node[shape=circle,draw,inner sep=2pt] (char) {#1};}}
\newcommand{\havoccond}[1]{\ahavoc(\sem{#1}{\asid}) \subseteq \sem{#1}{\asid}}


\newcommand{\closureof}[1]{\mathrm{cl}({#1})}
\newcommand{\constrof}[1]{\mathrm{dist}({#1})}
\newcommand{\basecomps}{\comps^\sharp}
\newcommand{\pbasecomps}{\overline{\comps}^\sharp}
\newcommand{\baseinteracs}{\interacs^\sharp}
\newcommand{\pbaseinteracs}{\overline{\interacs}^\sharp}
\newcommand{\basetuple}{(\basecomps, \baseinteracs, \pureform)}
\newcommand{\basetuplen}[1]{(\basecomps_{#1}, \baseinteracs_{#1}, \pureform_{#1})}
\newcommand{\pbasetuplen}[1]{(\pbasecomps_{#1}, \pbaseinteracs_{#1}, \ppureform_{#1})}
\newcommand{\satbasetuples}{\mathsf{SatBase}}
\newcommand{\satbasetuplesof}[1]{\satbasetuples_{\scriptscriptstyle\set{#1}}}
\newcommand{\abasetuple}{\mathfrak{t}}
\newcommand{\anbasetuple}{\mathfrak{u}}
\newcommand{\loosetuples}{\mathsf{Loose}}
\newcommand{\basetupleof}[2]{\mathsf{Base}({#1},{#2})}
\newcommand{\basepairof}[1]{\mathsf{Base}({#1})}
\newcommand{\loosetupleof}[2]{\mathsf{Loose}({#1},{#2})}
\newcommand{\absfunc}{\mathcal{F}}
\newcommand{\Absfunc}{\overrightarrow{\absfunc}}
\newcommand{\basecomp}{\otimes}
\newcommand{\loosecomp}{\odot}
\newcommand{\Loosecomp}{\Odot}
\newcommand{\Basecomp}{\bigotimes}
\newcommand{\basesid}{\asid^\sharp}
\newcommand{\loosesid}{\asid^\flat}
\newcommand{\basevartuple}{\overrightarrow{\mathcal{X}}}
\newcommand{\basevarof}[1]{\mathcal{X}({#1})}
\newcommand{\leastbase}{\mu\basevartuple.\basesid}
\newcommand{\loosebase}{\mu\basevartuple.\loosesid}
\newcommand{\leastbaseof}[1]{\leastbase({#1})}
\newcommand{\loosebaseof}[1]{\loosebase({#1})}
\newcommand{\reprof}[3]{{\{\!\!\{{#1}\}\!\!\}}^{\scriptscriptstyle{#2}}_{#3}}
\newcommand{\nordsubsets}[2]{S_{#1,#2}}
\newcommand{\polynomial}{\mathrm{poly}}
\newcommand{\looseof}[1]{\widetilde{#1}}

\newcommand{\Squig}{$\scriptsize$\approx$\normalsize$\!}
\newcommand{\lsquigend}{$\scriptsize$\lhd\!$\normalsize$}
\newcommand{\rsquigend}{$\scriptsize\rule{.1ex}{0ex}$\rhd$\normalsize$}

\newcounter{index}

\newcommand\squigs[1]{%
  \setcounter{index}{0}%
  \whiledo {\value{index}< #1}
  {\addtocounter{index}{1}\squig}
}

\newcommand\Squigs[1]{%
  \setcounter{index}{0}%
  \whiledo {\value{index}< #1}
  {\addtocounter{index}{1}\Squig}
}

\newcommand\rsquigarrow[2]{$
  \setbox0\hbox{$\squigs{#2}\rsquigend$}%
  \tiny$%
  \!\!\!\!\begin{array}{c}%
  {#1}\\%
  \usebox0%
  \end{array}%
  $\normalsize$\!\!%
}

\newcommand\Rsquigarrow[2]{$
  \setbox0\hbox{$\Squigs{#2}\rsquigend$}%
  \tiny$%
  \!\!\!\!\begin{array}{c}%
  {#1}\\%
  \usebox0%
  \end{array}%
  $\normalsize$\!\!%
}

\newcommand\lsquigarrow[2]{$
  \setbox0\hbox{$\lsquigend\squigs{#2}$}%
  \tiny$%
  \!\!\!\!\begin{array}{c}%
  {#1}\\%
  \usebox0%
  \end{array}%
  $\normalsize$\!\!%
}

\newcommand\Lsquigarrow[2]{$
  \setbox0\hbox{$\lsquigend\Squigs{#2}$}%
  \tiny$%
  \!\!\!\!\begin{array}{c}%
  {#1}\\%
  \usebox0%
  \end{array}%
  $\normalsize$\!\!%
}

\newcommand{\depof}[2]{\raisebox{4pt}{$~\rsquigarrow{({#1},{#2})}{4}~$}}
\newcommand{\Depof}[2]{\raisebox{4pt}{$~\rsquigarrow{({#1},{#2})}{8}~$}}
\newcommand{\ddepof}[2]{\raisebox{4pt}{$~\Rsquigarrow{({#1},{#2})}{4}~$}}
\newcommand{\Ddepof}[2]{\raisebox{4pt}{$~\Rsquigarrow{({#1},{#2})}{8}~$}}
\newcommand{\deppof}[2]{\ddepof{#1}{#2}^{\raisebox{-4pt}{$\bullet$}}}
\newcommand{\Deppof}[2]{\Ddepof{#1}{#2}^{\raisebox{-4pt}{$\bullet$}}}
\newcommand{\graphof}[1]{\mathfrak{G}_{\scriptscriptstyle{#1}}}
\newcommand{\eqgraphof}[1]{\mathcal{G}^\approx({#1})}
\newcommand{\eqof}[1]{\approx_{\scriptscriptstyle{#1}}}
\newcommand{\primegraphof}[1]{\mathcal{G}^\prime({#1})}
\newcommand{\atpos}[2]{{#1}^{\scriptscriptstyle\![{#2}]}}
\newcommand{\tgraphof}[1]{\mathcal{G}^\bullet({#1})}
\newcommand{\unfold}[2]{\Arrow{#1}{#2}}

\newcommand{\degreebound}{\mathfrak{B}}
\newcommand{\interactionvars}[1]{\mathrm{iv}({#1})}

\newcommand{\graph}{G}
\newcommand{\graphsof}[1]{{\mathcal{G}^{\scriptscriptstyle{#1}}}}
\newcommand{\graphs}{\graphsof{}}
\newcommand{\btwgraphsof}[1]{{\mathcal{G}^{\scriptscriptstyle{#1}}_\red}}
\newcommand{\vertices}{V}
\newcommand{\vertof}[1]{\vertices_{\scriptscriptstyle{#1}}}
\newcommand{\edgeof}[1]{\edges_{\scriptscriptstyle{#1}}}
\newcommand{\labels}{\lambda}
\newcommand{\labof}[1]{\labels_{\scriptscriptstyle{#1}}}
\newcommand{\edgerel}{\upsilon}
\newcommand{\edgerelof}[1]{\edgerel_{\scriptscriptstyle{#1}}}
\newcommand{\sourcelabels}{\mathbb{S}}
\newcommand{\treelabels}{\mathbb{T}}
\newcommand{\slabs}{\tau}
\newcommand{\colorof}[1]{\mathtt{colors}({#1})}
\newcommand{\sources}{\xi}
\newcommand{\sourceof}[1]{\sources_{\scriptscriptstyle{#1}}}
\newcommand{\agop}[1]{\gamma_{#1}}
\newcommand{\algof}[1]{\mathbf{#1}}
\newcommand{\initalg}[1]{{\mathbf{M}({#1})}}
\newcommand{\inituniv}[1]{{\mathcal{M}({#1})}}
\newcommand{\red}{\mathsf{\scriptscriptstyle{gen}}}
\newcommand{\redalgof}[1]{{\algof{#1}_\red}}
\newcommand{\universeOf}[1]{\mathcal{#1}}
\newcommand{\reduniverseOf}[2]{{\universeOf{#1}^{#2}_\red}}
\newcommand{\algebra}{\algof{G}}
\newcommand{\btwalgebraof}[1]{\algebra^{#1}_\red}
\newcommand{\btwalgebra}{\btwalgebraof{}}
\newcommand{\treealgebra}{\algof{T}}
\newcommand{\rootalgebra}{\algebra_\aroot}
\newcommand{\repalgebra}{\algof{P}}
\newcommand{\emptygraph}{\mathbf{0}}
\newcommand{\emptygraphSymb}{\underline{\emptygraph}}
\newcommand{\singlegraph}{\overline{\mathbf{1}}}
\newcommand{\sgraph}[1]{\mathbf{#1}}
\newcommand{\sgraphSymb}[1]{\underline{\sgraph{#1}}}
\newcommand{\tgraph}[1]{\mathfrak{#1}}
\newcommand{\graphsign}{\mathcal{H}}
\newcommand{\aroot}{\mathfrak{r}}
\newcommand{\temp}{\mathit{aux}}
\newcommand{\grammar}{\Gamma}
\newcommand{\acomp}{\circ}
\newcommand{\accomp}{\!\parallel\!}
\newcommand{\ptrees}[2]{\mathcal{T}_{#1}({#2})}
\newcommand{\val}[1]{\overline{\mathrm{val}}({#1})}
\newcommand{\hval}{\mathbf{val}}
\newcommand{\bagof}[2]{B_{#2}(#1)}
\newcommand{\tdsof}[1]{\mathcal{D}({#1})}
\newcommand{\dec}{\mathrm{dec}}
\newcommand{\nonterminal}[1]{\mathsf{nonterm}_{#1}}
\newcommand{\node}{\mathsf{node}}
\newcommand{\nodefunc}{\mathtt{node}}
\newcommand{\graphfunc}{\mathtt{graph}}
\newcommand{\parent}{\mathsf{parent}}
\newcommand{\leftedge}{\mathsf{left}}
\newcommand{\rightedge}{\mathsf{right}}
\newcommand{\nextedge}{\mathsf{next}}
\newcommand{\lex}{\prec_{\mathit{lex}}}
\newcommand{\leaf}{\mathit{leaf}}
\newcommand{\bag}{\mathsf{bag}}
\newcommand{\trans}{\delta}
\newcommand{\scheme}{\Theta}
\newcommand{\defdof}[2]{\mathrm{def}^{{#2}}_{{#1}}}
\newcommand{\defd}[1]{\mathrm{def}_{#1}}
\newcommand{\hr}{$\mathsf{HR}$}

\newcommand{\asgn}{\alpha}
\newcommand{\eqpathof}[2]{\beta_{\scriptscriptstyle{#1}}^{\!{#2}}}
\newcommand{\asgnset}{\mathsf{Asgn}}
\newcommand{\maxarity}{\mathsf{MaxArity}}
\newcommand{\params}{\Pi}
\newcommand{\maxpreds}{\mathsf{MaxPreds}}
\newcommand{\persistof}[1]{\mathsf{Persist}({#1})}
\newcommand{\acolor}{\mathsf{c}}
\newcommand{\iscolored}[2]{\mathsf{C}_{#2}({#1})}
\newcommand{\isasgncolored}[3]{\mathsf{C}^{#2}_{#3}({#1})}
\newcommand{\colors}{\Gamma}
\newcommand{\maxexists}{\mathcal{M}}
\newcommand{\thecolors}{\gamma}
\newcommand{\remcolors}{\eta}

\newcommand{\strucgraph}{\struc_\graph}
\newcommand{\signagraph}{\Gamma}
\newcommand{\encof}[1]{\mathrm{enc}\ifthenelse{\equal{#1}{}}{}{({#1})}}
\newcommand{\posof}[1]{\mathrm{pos}({#1})}
\newcommand{\vertset}{\mathsf{V}}
\newcommand{\tnodes}{\mathcal{N}}
\newcommand{\tedges}{\mathcal{F}}
\newcommand{\nodes}{{N}}
\newcommand{\edges}{{E}}
\newcommand{\alabel}{\lambda}
\newcommand{\tree}{T}
\newcommand{\trees}{\universeOf{T}}
\newcommand{\subtree}[2]{{#1}|_{{#2}}}
\newcommand{\precof}[1]{\prec_{#1}}
\newcommand{\grid}[1]{\mathcal{H}_{#1}}
\newcommand{\clique}[1]{\mathcal{K}_{#1}}
\newcommand{\iso}{\simeq}
\newcommand{\isoof}[1]{[#1]_\iso}
\newcommand{\twof}[1]{\mathrm{tw}({#1})}
\newcommand{\charform}[1]{\Theta({#1})}
\newcommand{\exclof}[1]{{#1}^\exists}
\newcommand{\occin}[1]{\leadsto_{#1}}
\newcommand{\domsymb}{\mathfrak{D}}
\newcommand{\vertex}{u}
\newcommand{\vertexSet}{U}
\newcommand{\Rank}{r}
\newcommand{\thof}{\mathit{Th}}
\newcommand{\types}{\mathit{Types}}
\newcommand{\operation}{\mathit{op}}
\newcommand{\fuse}{\mathit{fuse}}
\newcommand{\method}{m}
\newcommand{\domain}{\mathit{dom}}
\newcommand{\encode}{\mathit{encode}}
\newcommand{\rlandf}{\mathit{forget}}
\newcommand{\pop}{\parallel}
\newcommand{\sort}{\sigma}
\newcommand{\sorts}{\Sigma}
\newcommand{\hrsorts}{\sorts_{\mathsf{HR}}}

\newcommand{\restrict}[1]{\mathsf{restrict}_{#1}}
\newcommand{\restrictSymb}[1]{\underline{\mathsf{restrict}}_{#1}}
\newcommand{\rename}[1]{\mathsf{rename}_{#1}}
\newcommand{\renameSymb}[1]{\underline{\mathsf{rename}}_{#1}}

\newcommand{\binlabel}[1]{\mathrm{bin}_{#1}}
\newcommand{\unlabel}[1]{\mathrm{un}_{#1}}
\newcommand{\extend}[1]{\mathsf{append}_{#1}}
\newcommand{\vertexlabel}{\mathsf{vertex}}
\newcommand{\vertexsetlabel}{\mathsf{source}}
\newcommand{\edgelabel}{\mathsf{edge}}
\newcommand{\parsefunc}{\pi}

\begin{abstract}
  We give a characterization of the sets of graphs that are both
  \emph{definable} in Counting Monadic Second Order Logic (CMSO) and
  \emph{context-free}, i.e., least solutions of Hyperedge-Replacement
  (HR) grammars introduced by Courcelle and Engelfriet
  \cite{courcelle_engelfriet_2012}. We prove the equivalence of these
  sets with:
  (a) \emph{recognizable} sets (in the algebra of graphs with
  HR-operations) of bounded tree-width; we refine this condition
  further and show equivalence with recognizability in a finitely
  generated subalgebra of the HR-algebra of graphs;
  (b) \emph{parsable} sets, for which there is a definable transduction from
  graphs to a set of derivation trees labelled by HR operations, such
  that the set of graphs is the image of the set of derivation trees
  under the canonical evaluation of the HR operations;
  (c) images of recognizable unranked sets of trees under a definable
  transduction, whose inverse is also definable.
  We rely on a novel connection between two seminal results, a logical
  characterization of context-free graph languages in terms of
  tree-to-graph definable transductions, by Courcelle and
  Engelfriet~\cite{CourcelleEngelfriet95} and a proof that an
  optimal-width tree decomposition of a graph can be built by an
  definable transduction, by Boja\'{n}czyk and
  Pilipczuk~\cite{10.1145/2933575.2934508,journals/lmcs/BojanczykP22}.
\end{abstract}

\maketitle

\section{Introduction}

Formal language theory studies finite representations of infinite sets
of objects (e.g., words, trees, graphs). These representations can be
\emph{descriptive}, specifying logical properties of their members
(e.g. planar or Hamiltonian graphs), or \emph{constructive},
describing how the members of the set are built. In particular,
constructive representations come with algebras that define sets of
operations. \emph{Context-free} sets arise from the least solutions of
recursive equation systems, which use operations from the considered
algebra, with unknowns ranging over sets. \emph{Recognizable} sets are
defined in terms of congruence relations over the algebra, with a
finite number of equivalence classes; these equivalence classes can be
used to define equivalent notions of recognizability in terms of
automata or homomorphisms to finite algebras (such as monoids for
words).

Monadic Second Order Logic (\mso) is the most prominent descriptive
representation of graphs, and has seen decades of study, see for
example~\cite{courcelle_engelfriet_2012}.  Hyperedge Replacement (\hr)
algebras provide standard constructive representations, using (sorted)
substitutions of a hyperedge in a graph by graph with a tuples of
designated vertices, that matches the sort of the
hyperedge~\cite{courcelle_engelfriet_2012}.  The notion of
context-free \hr\ graph grammar then follows immediately from the
definition of a \hr\ algebra of graphs.  In contrast, the right notion
of recognizability is somewhat less obvious.  This is because words
and trees have a clear beginning (root) and traversal direction
(left-right, top-down or reverse), while graphs do not have either.
Instead of having a congruence relation with finite index, the
proposal put forward in~\cite{courcelle_engelfriet_2012} consists of a
\emph{locally finite} congruence relation, i.e., a congruence relation
with a finite number of equivalence classes for every sort.

The comparison of the expressive powers of different representations
is central to formal language theory.  For words, definability in
\mso{} coincides with recognizability~\cite{Buechi90}, being subsumed
by context-freeness, whereas for ground terms over a finite set of
function symbols definability in \mso, recognizability and
context-freeness coincide\footnote{The context-free tree grammars
mentioned in \cite{DBLP:journals/corr/Engelfriet15} use first-order
parameters in rules, being thus strictly more expressive than tree
automata~\cite{comon:hal-03367725}. Here, by a grammar, we understand
a finite set of recursive equations whose left-hand sides consist of a
single nonterminal of arity zero. The components of the least solution
of such a system are also known as \emph{equational
sets}~\cite{courcelle_engelfriet_2012}.}
\cite{Doner70,MezeiWright67}. For unranked and unordered trees (i.e.,
trees with arbitrarily many children per node, whose order is,
moreover, not important), definability in \cmso{} and recognizability
coincide, where \cmso\ is the extension of \mso\ with modulo
constraints on the cardinality of sets \cite{CourcelleI}. For graphs,
definability in \cmso{} implies recognizability but not vice versa,
whereas context-freeness is incomparable to the two other
notions~\cite{CourcelleI}. The equivalence between recognizability and
definability in \cmso{} can be recovered for graphs of bounded
\emph{tree-width}~\cite{10.1145/2933575.2934508}. Moreover, recognizability of
bounded tree-width sets of graphs (by locally finite congruences) is
equivalent to recognizability by congruences having finitely many
classes~\cite{courcelle_lagergren_1996}.

Finite representations are used in system design and automated
verification tools. Descriptive representations (logics) specify
correctness properties, e.g., sets of safe states or behaviors (traces
of states), whereas constructive representations describe the
implementations of a system, with respect to the low-level details of
state changes. Verification problems, such as conformance with certain
safety criteria, or equivalence of two implementations, amount to
checking inclusion between sets of words, trees or graphs, represented
in different ways. Hence, the interest for classes of representations
having a \emph{decidable inclusion problem}.

We study the intersection between the classes of graphs that are both
context-free and definable (in \cmso). The main motivation is that
inclusion is decidable for the members of this intersection. Let
$\mathcal{L}_1$ and $\mathcal{L}_2$ be sets defined by
\cmso\ formul{\ae} $\phi_1$ and $\phi_2$, respectively. Then,
$\mathcal{L}_1 \subseteq \mathcal{L}_2$ if and only if the formula
$\phi_1 \wedge \neg \phi_2$ is not satisfiable.  If, moreover,
$\mathcal{L}_1$ is context-free, there is an effectively computable
bound on the tree-width of the models of $\phi_1 \wedge \neg \phi_2$,
if any. Since the satisfiability problem for \cmso\ is decidable for
graphs of bounded tree-width, by a seminal result of
Courcelle~\cite[Corollary 4.8 (2)]{CourcelleI}, the problem
$\mathcal{L}_1 \subseteq \mathcal{L}_2$ is decidable.

Our characterization of context-free and definable graph languages
starts from the notion of \emph{strongly context-free} sets,
introduced by Courcelle~\cite{CourcelleV}. These are sets
$\mathcal{L}$ generated by an \hr\ grammar, having an additional
\emph{parsability} property: there exists a binary relation $F$
between graphs and derivation trees, such that
(i) each output tree is defined by a finite tuple of
  \cmso\ formul{\ae} interpreted over the input graph, and
(ii) for each graph $\graph\in\mathcal{L}$, the set $F(\graph)$
  contains a derivation tree that evaluates to $\graph$.
  In this context, Courcelle stated the following conjectures:

\begin{conjC}[{\cite[Conjecture 3]{CourcelleV}}]\label{conj:courcelle3}
  If a set of graphs is context-free and definable, then it is
  strongly context-free.
\end{conjC}

This conjecture leads to the following insight: given a graph grammar,
assume that we want to prove that its language is definable. By
\autoref{conj:courcelle3}, now proved as a consequence of
\autoref{thm:cf-def-v2}, definability is equivalent to the existence
of a definable parsing transduction for the language generated by the
grammar. Hence, for constructing the desired \cmso{} formula it is
always a viable proof strategy to either explicitly or implicitly
build such a parsing function as part of the overall construction.
The contrapositive of this conjecture is proved as \cite[Theorem
  4.8]{CourcelleV}. Moreover, the equivalence with the following
conjecture is also proved in \cite{CourcelleV}:

\begin{conjC}[{\cite[Conjecture 2]{CourcelleV}}]\label{conj:courcelle2}
  For each $k \in \nat$, the set of all graphs of tree-width at most
  $k$ is strongly context-free.
\end{conjC}

\paragraph{Our Contributions}
The main contribution of the paper is a detailed and self-contained
proof of the above conjectures, yielding two characterizations of the
intersection between the context-free and the definable classes of
graphs.

The fine-grained version (Theorem \ref{thm:cf-def}), takes into
account the finite set of sorts occurring in the grammar and proves
the equivalence between:
(1) $\mathcal{L}$ is definable in \cmso{} and generated by a
\hr\ grammar that uses a finite set of sorts $\slabs$,
(2) $\mathcal{L}$ is recognizable in the infinitely-sorted \hr\
algebra and is represented by a set of $\slabs$-sorted terms,
(3) $\mathcal{L}$ is recognizable in the $\slabs$-sorted
\hr\ subalgebra and is represented by a set of $\slabs$-sorted terms,
which
(4) can be extracted from the graphs in $\mathcal{L}$ by a relation
definable in \cmso.

The equivalence of (2) and (3) provides also a simpler proof of the
equivalence between locally-finite and finite recognizability for
graphs of bounded tree-width, initially proved by Courcelle and
Lagergren~\cite{courcelle_lagergren_1996}. Since we prove later that
recognizability in the infinitely-sorted \hr\ algebra is equivalent to
recognizability in an infinite sequence of finitely-sorted
\hr\ subalgebras (Theorem~\ref{thm:rec-limit}), the equivalence of
points (2) and (3) in Theorem~\ref{thm:cf-def} gives a cut-off result:
a set of tree-width bounded graphs is recognizable if and only if it
is recognizable in a finitely-sorted algebra.

The coarse-grained version (Theorem~\ref{thm:cf-def-v2}) quantifies
existentially over the set of sorts and states the equivalence
between:
(1) $\mathcal{L}$ is definable in \cmso{} and \hr\ context-free,
(2) $\mathcal{L}$ is recognizable and of bounded tree-width,
(3) $\mathcal{L}$ is represented by a set of terms that can be
extracted from $\mathcal{L}$ by a relation definable in \cmso, and
%
(4) the existence of two definable relations from graphs to trees and
back, whose composition is the identity on $\mathcal{L}$.
It is known that context-free sets have bounded tree-width, but not
the other way around. A consequence of our result is that every
definable set that has bounded tree-width is context-free.

Our results rely on two seminal ingredients. The first is a
characterization of the context-free sets of graphs, as \emph{images
of recognizable ranked sets of trees under definable relations}, by
Courcelle and Engelfriet \cite{CourcelleEngelfriet95}. The second is a
construction of \emph{tree decompositions of optimal width}, by means
of definable relations, by Boja\'{n}czyk and
Pilipczuk~\cite{10.1145/2933575.2934508,journals/lmcs/BojanczykP22}.
We connect the two results using
(i) a generalization of ~\cite{CourcelleEngelfriet95} by considering
\emph{unranked} instead of ranked recognizable sets of trees
(Corollary \ref{cor:cf-unranked}), and
(ii) a definable translation of a tree decomposition into a parse tree
of a \hr\ grammar (Lemma
\ref{lem:mapping-tree-decompositions-back-to-dervation-trees}).

\paragraph{Related Work}
Following its initial
development~\cite{CourcelleV,CourcelleEngelfriet95}, the study of
definable context-free sets of graphs has seen recent interest. Our
work is closely related to~\cite{journals/corr/abs-2305-18039}, where
Boja\'{n}czyk proposes the notions of \emph{logical recognizability}
and \emph{definable tree decompositions}.  The latter notion is used
to formalize the condition (\ref{it4:thm:cf-def-v2}) of
Theorem~\ref{thm:cf-def-v2}.  For his logical recognizability notion,
it becomes immediate from the definitions that the recognizable
subsets of a class that has definable tree decompositions are
definable in \cmso.  Then,~\cite{journals/corr/abs-2305-18039}
concentrates on establishing the equivalence between congruence-based
and logical recognizability on words, trees and graphs of bounded
tree-width.  The equivalences for words and ranked trees are based on
the classical results of B\"uchi~\cite{Buechi90} and
Doner~\cite{Doner70}, respectively, whereas the equivalence for
bounded tree-width graphs uses the same ingredients as our
work~\cite{CourcelleEngelfriet95,10.1145/2933575.2934508,journals/lmcs/BojanczykP22},
considered in more generality and on a higher level of abstraction, as
the goal of~\cite{journals/corr/abs-2305-18039} is to avoid the
introduction of sorts and graph operations as needed for the
definition of context-free sets of graphs in terms of graph grammars.
In summary, Boja\'{n}czyk establishes the equivalences (1)
$\Leftrightarrow$ (2) $\Leftrightarrow$ (4) of our coarse-grained
Theorem~\ref{thm:cf-def-v2}, but the fine-grained characterization of
Theorem~\ref{thm:cf-def} cannot be immediately derived from the
development in~\cite{journals/corr/abs-2305-18039}.

We note that the problem of whether a given context-free grammar
defines a recognizable (and hence definable) language is undecidable
(even for words), according to a result by Greibach~\cite{Greibach68}.
This has motivated the search for regular grammars and regular
expressions over graphs, whose languages are guaranteed to be
recognizable resp.\  definable.  In particular, already Courcelle
in~\cite[Section~5]{CourcelleV} proposed so-called \emph{regular graph
grammars}, built over hyperedge-replacement operations that need to
satisfy some local connectivity requirements.  In recent work, we have
proposed \emph{tree-verifiable graph
grammars}~\cite{LPAR2024:Tree_Verifiable_Graph_Grammars}, which
strictly generalize the regular graph grammars of Courcelle. We note
that these grammars do not capture all recognizable sets of graphs but
only those recognizable sets of graphs of bounded \emph{embeddable}
tree-width, i.e., those sets of graphs for which there is a tree
decomposition whose backbone is a spanning tree of the considered
graph. There has also been recent progress on defining the
recognizable sets of tree-width at most $2$ (a class of graphs that is
orthogonal to the class of graphs of bounded embeddable
treewidth). These sets can be defined equivalently by regular
expressions~\cite{DBLP:conf/icalp/Doumane22} and regular graph
grammars~\cite{LICS2025:Reg_Grammars_TW_2}.

We finally mention a formal comparison between the expressivity of
\mso{} with that of Separation Logic over graphs of bounded
tree-width~\cite{DBLP:conf/concur/IosifZ23}.

\section{Preliminaries}
\label{sec:definitions}

This section introduces the basic notions of the descriptive and
constructive representations of infinite sets. We introduce (Counting)
Monadic Second Order Logic as the main descriptive language for sets
of structures (\autoref{sec:cmso}) and relations between structures
(\autoref{sec:transductions}). The constructive representations we
consider are the standard notions of recognizable
(\autoref{sec:recognizable}) and context-free sets
(\autoref{sec:context-free}) in multi-sorted algebras.

We denote by $\nat$ the set of natural numbers and $\nat_+ \isdef \nat
\setminus \set{0}$. Given $i, j \in \nat$, we write $\interv{i}{j}
\isdef \set{i, i+1, \ldots, j}$, assumed to be empty if $i>j$. The
cardinality of a finite set $A$ is denoted by $\cardof{A}$. By writing
$A \finsubseteq B$ we mean that $A$ is a finite subset of $B$. For a
set $A$, we denote by $\pow{A}$ its powerset, $A^0 \isdef
\set{\epsilon}$, $A^{i+1} \isdef A^i \times A$, for all $i \geq 0$,
$A^* \isdef \bigcup_{i\geq0} A^i$ and $A^+ \isdef \bigcup_{i\geq1}
A^i$, where $\times$ is the Cartesian product and $\emptyseq$ denotes
the empty sequence. Intuitively, $A^*$ (resp.\  $A^+$) denotes the set
of possibly empty (resp.\  nonempty) sequences of elements from $A$. The
length of a sequence $\vec{a} \in A^*$ is denoted as $\lenof{\vec{a}}$
and $\vec{a}_i$ denotes its $i$-th element, for $i \in
\interv{1}{\lenof{\vec{a}}}$.

For a relation $R \subseteq A \times B$, we denote by $\dom{R}$ and
$\img{R}$ the sets consisting of the first and second components of
the pairs in $R$, respectively. We write $R^{-1}$ for the inverse
relation and $R(S)$ for the image of a set $S$ via $R$. Sometimes we
write $R(a)$ instead of $R(\set{a})$, for an element $a \in A$. The
\emph{domain-restriction} $\proj{R}{C}$ restricts the relation $R$ to
the pairs with first element in $C$. A bijective function $f$ is an
\emph{$A$-permutation} if $\set{a \in \dom{f} \mid f(a) \neq a}
\subseteq A \subseteq \dom{f}$. It is a \emph{finite permutation} if
it is an $A$-permutation, for some finite set $A$.

\subsection{Counting Monadic Second Order Logic (\cmso)}
\label{sec:cmso}

A \emph{relational signature} $\relations$ is a finite set of
\emph{relation symbols}, ranged over by $\arel$, of arities
$\arityof{\arel}\geq0$. A relation symbol $\arel$ is a
\emph{constant}, \emph{unary} or \emph{binary} if $\arityof{\arel}=0$,
$1$ or $2$, respectively.

A \emph{$\relations$-structure} is a pair $\astruc = (\univ,\struc)$,
where $\univ$ is a \emph{universe} and $\struc: \relations \rightarrow
\pow{\univ^*}$ is an \emph{interpretation}, that maps each relation
symbol $\arel$ into a subset of $\univ^{\arityof{\arel}}$ of
corresponding arity.  Here we consider only structures with finite
universe, also called \emph{finite structures}.  The set of
$\relations$-structures is denoted by $\strucof{\relations}$.

The \emph{Counting Monadic Second Order Logic} (\cmso) is the set of
formul{\ae} written using a set $\vars = \set{x,y,\ldots}$ of
\emph{first-order variables}, a set $\Vars = \set{X,Y,\ldots}$ of
\emph{second-order variables} and the relation symbols from
$relations$, according to the following syntax:
\[\psi := x=y \mid \arel(x_1, \ldots, x_{\arityof{\arel}}) \mid X(x) \mid
\cardconstr{X}{q}{p} \mid \neg\psi \mid \psi \wedge \psi \mid \exists
x ~.~ \psi \mid \exists X ~.~ \psi\] where $p,q \in \nat$ are
constants, such that $p \in \interv{0}{q-1}$. By \mso\ we denote the
subset of \cmso\ consisting of formul{\ae} that do not contain atomic
propositions of the form $\cardconstr{X}{q}{p}$, also called
\emph{cardinality constraints}. A variable is \emph{free} in a formula
$\phi$ if it does not occur in the scope of a quantifier. A
\emph{sentence} is a formula with no free variables.

The semantics of \cmso\ is given by a satisfaction relation
$(\univ,\struc) \Models^\store \psi$, where the store $\store : \vars
\cup \Vars \rightarrow \univ \cup \pow{\univ}$ maps each variable $x
\in \vars$ to an element of the universe and each variable $X \in
\Vars$ to a subset of $\univ$.  This relation is defined
inductively on the syntactic structure of formul{\ae}:
\[\begin{array}{rclcl}
(\univ,\struc) & \Models^\store & x=y & \iff & \store(x)=\store(y) \\
(\univ,\struc) & \Models^\store & \arel(x_1, \ldots, x_{k}) & \iff &
\tuple{\store(x_1), \ldots, \store(x_{k})} \in \struc(\arel) \\
(\univ,\struc) & \Models^\store & X(x) & \iff & \store(x) \in \store(X) \\
(\univ,\struc) & \Models^\store & \cardconstr{X}{q}{p} & \iff & \cardof{\store(X)} = kq+p \text{, for some } k\in\nat \\
(\univ,\struc) & \Models^\store & \phi \wedge \psi & \iff & (\univ,\struc) \Models^\store \phi \text{ and } (\univ,\struc) \Models^\store \psi \\
(\univ,\struc) & \Models^\store & \neg\phi & \iff & (\univ,\struc) \not\Models^\store \phi \\
(\univ,\struc) & \Models^\store & \exists x ~.~ \psi & \iff & (\univ,\struc) \Models^{\store[x\leftarrow u]}
\text{, for some element } u \in \univ \\
(\univ,\struc) & \Models^\store & \exists X ~.~ \psi & \iff & (\univ,\struc) \Models^{\store[X\leftarrow V]} \psi
 \text{, for some set } V \subseteq \univ
\end{array}\]
If $\phi$ is a sentence, the satisfaction relation does not depend on
the store and we write $(\univ,\struc) \Models \phi$ instead of
$(\univ,\struc) \Models^\store \phi$.  A set $S$ of structures is
\emph{definable} iff $S = \set{(\univ,\struc) \mid (\univ,\struc)
  \Models \phi}$, for some \cmso{} sentence $\phi$, and
\mso-\emph{definable} in case $\phi$ belongs to the \mso{} fragment of
\cmso.
Two structures are \emph{isomorphic} iff they differ only by a
renaming of their elements (a formal definition is given in
\cite[Section A3]{DBLP:books/daglib/0082516}). It is known that the
satisfaction relation of \cmso{} does not distinguish between
isomorphic structures. We note that we only consider finite structures
in this paper, and hence quantification is over finite sets
only\footnote{For infinite structures there is a difference in
expressivity when only quantification over finite set is allowed and
the resulting logic is known as \emph{weak} \mso. However, quantifiers
in finite structures can only capture finite sets and hence we not
need distinguish between \mso\ and \emph{weak} \mso.}.
  
\subsection{Definable Transductions}
\label{sec:transductions}

Let $\relations$ and $\relations'$ be relational signatures. A
relation $\trans$ between $\relations$- and $\relations'$-structures
is a $k$-\emph{copying} $(\relations,\relations')$-\emph{transduction}
iff each output structure $\astruc' \in \trans(\astruc)$ is produced
from $k$ disjoint copies of the input structure $\astruc$, called
\emph{layers}. The transduction is said to be \emph{copyless} if
$k=1$.  The outcome of the transduction also depends on the valuation
of zero or more \emph{set parameters} $X_1, \ldots, X_n \in \Vars$,
that range over the subsets of the input universe. The transduction is
said to be \emph{parameterless} if $n=0$. Formally, we define
$(\relations,\relations')$-transductions using \emph{transduction
schemes}, i.e., finite tuples of \cmso\ formul{\ae}:
\begin{align*}
  \scheme=\tuple{\varphi,\set{\psi_i}_{i \in \interv{1}{k}},
    \set{\theta_{(\qrel,i_1,\ldots,i_{\arityof{\qrel}})}}_{
      \qrel\in\relations',~
      i_1,\ldots,i_{\arityof{\qrel}}\in\interv{1}{k}
  }}
\end{align*}
where:
\begin{itemize}[label=$\triangleright$]
\item $\varphi(X_1,\ldots,X_n)$ selects the input structures
  $(\univ,\struc)$ for which the transduction has an output
  $(\univ',\struc')$, i.e., those structures $(\univ,\struc)$ such
  that $(\univ,\struc) \models^\store \varphi$, for a store $\store$
  that maps each $X_i$ into a set $\store(X_i) \subseteq \univ$,
\item $\psi_i(x_1, X_1, \ldots, X_n)$ defines the elements from the
  $i$-th layer copied in the output universe:
  \[\univ' \isdef \set{(u,i) \in \univ \times \interv{1}{k} \mid
    (\univ,\struc) \models^{\store[x_1 \leftarrow u]} \psi_i}\]
\item
  $\theta_{(\qrel,i_1,\ldots,i_{\arityof{\qrel}})}(x_1, \ldots,
  x_{\arityof{\qrel}}, X_1, \ldots, X_n)$ define the interpretation of
  $\qrel\in\relations'$ in the output:
  \[\struc'(\qrel) \isdef \{\tuple{(u_1,i_1), \ldots, (u_{\arityof{\qrel}}, i_{\arityof{\qrel}})} \mid
  (\univ,\struc) \Models^{\store[x_1\leftarrow u_1, \ldots, x_{\arityof{\qrel}} \leftarrow u_{\arityof{\qrel}}]}
  \theta_{(\qrel,i_1,\ldots,i_{\arityof{\qrel}})}\}\]
\end{itemize}
Note that the store $\store$ that defines the valuations of $X_1,
\ldots, X_n$ is the same everywhere in the above definition of the
output structure $(\univ',\struc')$. The output of the transduction is
denoted by $\defdof{\scheme}{\store} \isdef (\univ',\struc')$. The set
$\defd{\scheme}(\astruc)$ is the closure under isomorphism of the set
$\set{\defdof{\scheme}{\store} \mid \astruc \models^\store \varphi}$,
i.e., the output structures are the structures isomorphic to some
$\defdof{\scheme}{\store}$, whose elements are not necessarily pairs
of the form $(u,i)\in\univ\times\interv{1}{k}$. A transduction
$\trans$ is \emph{definable} iff $\trans=\defd{\scheme}$, for some
transduction scheme $\scheme$.

\begin{exa}\label{ex:tll}
  Let $\alphabet=\set{\leftedge,\rightedge,\nextedge}$ and
  $\alphabetTwo=\set{\leftedge,\rightedge}$ be alphabets of edge
  labels, which we will use in order to encode edge-labeled binary
  trees (for $\alphabetTwo$) and trees with linked leaves (for
  $\alphabet$). As usual, the left (resp.\  right) child of a node in
  the tree will be linked to its parent via a $\leftedge$-labeled
  (resp.\  $\rightedge$-labeled) edge. The $\nextedge$-labeled edges are
  used to link a leaf to its direct successor, in the lexicographic order.

  We use here the \emph{incidence encoding} of trees (more generally,
  graphs) as relational structures, where edges are elements of the
  universe and the graph is described by the incidence relation
  between edges and vertices (see \autoref{sec:cmso-def}). In
  particular, we use relation symbols $\arel_a$ for $a \in
  \set{\leftedge,\rightedge,\nextedge}$ of arity three, where the
  first argument of $\arel_a$ denotes an edge and the second and third
  element denote the source and destination of this edge,
  respectively.

  Then, there is an \mso\ sentence $\varphi_{\mathsf{tree}}$ that
  defines the set of binary trees, where each non-leaf node has
  exactly one $\leftedge$- and one $\rightedge$-child. Moreover, we
  consider the following \mso{} formul{\ae}: \begin{itemize}[label=$\triangleright$]
  \item $x \lex y$ means that $x$ is the direct predecessor of $y$ in
    the lexicographic order of the labels of the paths from the root
    to $x$ and $y$, induced by the $\leftedge < \rightedge$ order,
  \item $\leaf(x)$ states that $x$ has no outgoing $\leftedge$- or
    $\rightedge$-labeled edges, i.e., is a leaf of the tree.
  \end{itemize}

  The $2$-copying parameterless $(\set{\arel_b}_{b \in
    \alphabetTwo},\set{\arel_a}_{a \in \alphabet})$-transduction
  defined by the scheme
  \[\scheme_{\mathsf{tll}}\isdef\tuple{\varphi_{\mathsf{tree}},
    \set{\psi_1,\psi_2},
    \set{\theta_{(\arel_a,i,j,k)}}_{a\in\alphabet,i,j,k\in\interv{1}{2}}}\]
  where
  \begin{align*}
    \psi_1(x_1) \isdef & ~\mathtt{true}
    \\
    \psi_2(x_1) \isdef & ~\leaf(x_1) \wedge \exists y~.~ \leaf(y) \wedge x_1 \lex y
    \\
    \theta_{(\arel_a,i,j,k)}(x_1,x_2,x_3) \isdef & ~\left\{\begin{array}{ll}
    \arel_a(x_1,x_2,x_3) & \text{if } (i,j,k)=(1,1,1) \\
    \mathtt{false} & \text{otherwise}
    \end{array}\right. \text{for all } a \in \set{\leftedge,\rightedge}
    \\
    \theta_{(\arel_\nextedge,i,j,k)}(x_1,x_2,x_3) \isdef & ~\left\{\begin{array}{ll}
    \mathit{leaf}(x_2) \wedge \mathit{leaf}(x_3) \wedge
    x_1=x_2 \wedge x_2 \lex x_3 & \text{if } (i,j,k)=(2,1,1) \\
    \mathtt{false} & \text{otherwise}
    \end{array}\right.
  \end{align*}
  \noindent adds a $\nextedge$ edge between each pair of successive
  leaves in this lexicographic order. Note that the extra
  $\nextedge$-labeled edges are taken from the $2^{\mathit{nd}}$ layer
  and corresponds to the copy of its source node, i.e., the $\psi_2$
  selects all but the right-most leaf from the tree for the
  $2^{\mathit{nd}}$ layer.

  Dually, the copyless parameterless $(\set{\arel_a}_{a \in
    \alphabet},\set{\arel_b}_{b \in \alphabetTwo})$-transduction
  defined by the scheme \(\scheme_{\mathsf{tll}}^{-1} \isdef
  \tuple{\mathtt{true}, \set{\psi_1},
    \set{\theta_{(\arel_b,1,1)}}_{b\in\alphabetTwo}}\), where
  \begin{align*}
    \psi_1(x_1) \isdef & ~\forall y \forall z ~.~ \neg \arel_\nextedge(x_1,y,z)
    \\
    \theta_{(\arel_b,1,1)}(x_1,x_2,x_3) \isdef & ~ \arel_b(x_1,x_2,x_3), \quad \text{ for all } b \in \set{\leftedge,\rightedge}
  \end{align*}
  removes all $\nextedge$ edges from the input structure.
\end{exa}

The main property of definable transductions is the Backwards
Translation Theorem (see e.g., \cite[Theorem
  1.40]{courcelle_engelfriet_2012}):
\begin{thmC}[\cite{courcelle_engelfriet_2012}]\label{thm:back-trans}
  If $\mathcal{S} \subseteq \strucof{\relations'}$ is a definable set
  and $\trans$ is a definable $(\relations,\relations')$-transduction
  then the set $\trans^{-1}(\mathcal{S})$ is definable.
\end{thmC}
The following properties are direct consequences of the above theorem
and the definition of definable transductions:

\begin{prop}\label{prop:comp-restr-trans}
  \hfill
  \begin{enumerate}
  \item\label{it1:comp-restr-trans} The composition of definable
    transductions is definable.
  \item\label{it2:comp-restr-trans} The domain-restriction of a
    definable transduction by a definable set is definable.
  \item\label{it3:comp-restr-trans} The domain of a definable
    transduction is definable.
  \end{enumerate}
\end{prop}

\subsection{Recognizable Sets}
\label{sec:recognizable}

Let $\sorts$ be a set of \emph{sorts}, ranged over by $\sort$, and let
$\fsignature = \set{f_1, f_2, \ldots}$ be a \emph{functional signature}. which is a set of \emph{function symbols} $f$.
Each function symbol $f$ has an associated tuple of argument sorts and a value sort, denoted $\argsof{f} = \tuple{\sort_1, \ldots, \sort_n}$ and $\typeof{f}$, respectively.
The arity of $f$ is denoted $\arityof{f}\isdef n$.
A variable is a sorted symbol of arity zero, not part of the signature. The sort of a variable $x$ (resp.\  $X$) is
denoted $\typeof{x}$ (resp.\  $\typeof{X}$). Terms are build from
variables and function symbols of matching sorts, as usual. We write
$t(x_1, \ldots, x_n, X_1, \ldots, X_m)$ if $x_1, \ldots, x_n, X_1,
\ldots, X_m$ are the variables from $t$. A \emph{ground term} is a
term without variables. A \emph{first-order term} does not contain
second-order variables.

An $\fsignature$-\emph{algebra} $\algof{A} =
(\set{\universeOf{A}^\sort}_{\sort\in\sorts},\set{f^{\algof{A}}}_{f
  \in \fsignature})$ consists of a \emph{universe}
$\universeOf{A}^\sort$ for each sort $\sort\in\sorts$ and interprets
each function symbol $f\in\fsignature$ by a function $f^{\algof{A}} :
\universeOf{A}^{\sort_1} \times \ldots \times \universeOf{A}^{\sort_n}
\rightarrow \universeOf{A}^{\sort}$, where $\argsof{f} =
\tuple{\sort_1, \ldots, \sort_n}$ and $\typeof{f}=\sort$. The set
$\universeOf{A}\isdef\bigcup_{\sort\in\sorts} \universeOf{A}^\sort$
denotes the union of all universes of $\algof{A}$. The sort of an
element $a \in \universeOf{A}$ is denoted $\typeof{a}$. The algebra
$\algof{A}$ is \emph{locally finite} iff $\universeOf{A}_\sort$ is
finite, for each $\sort\in\sorts$ and finite iff $\universeOf{A}$ is
finite.

An $\fsignature$-term $t(x_1, \ldots, x_n, X_1, \ldots, X_m)$ is
viewed as a function symbol of arity $n+m$, whose interpretation is
obtained by interpreting the function symbols from $t$ in $\algof{A}$
and lifting the functions from elements to sets of elements of the
same sort. A set $\tsignature$ of first-order $\fsignature$-terms defines a \emph{derived
$\tsignature$-algebra} of $\algof{A}$ with the same set of sorts
$\sorts$ and the same universes $\universeOf{A}^\sort$ for each sort
$\sort\in\sorts$, that interprets each function symbol $t$ as the
function $t^\algof{A}$. A \emph{subalgebra} of $\algof{A}$ is any
algebra obtained by restricting the set of sorts, signature and
universes of $\algof{A}$.

An $\fsignature$-algebra $\algof{A}$ is \emph{term-generated} iff its
universe is the set of interpretations of the ground
$\fsignature$-terms in $\algof{A}$ (we call these elements
term-generated). The term-generated subalgebra $\redalgof{A}$ is the
subalgebra defined as the restriction of $\algof{A}$ to its
term-generated elements.

We denote by $\initalg{\fsignature}$ the \emph{initial algebra} over
the functional signature $\fsignature$. The universes
$\inituniv{\fsignature}_\sort$ of $\initalg{\fsignature}$ are the sets
of ground $\fsignature$-terms having the same sort $\sort$ and the
interpretation of each function symbol $f \in \fsignature$ is
$f^\initalg{\fsignature}(t_1,\ldots,t_{\arityof{f}}) \isdef
f(t_1,\ldots,t_{\arityof{f}})$, for all ground terms
$t_1,\ldots,t_{\arityof{f}}$ of matching sorts.

We recall below the standard notion of recognizability:
\begin{defi}\label{def:recognizability}
An equivalence relation $\congrec{}{}$ on $\universeOf{A}$ is a
\emph{congruence} iff $a \congrec{}{} b$ only if
(1) $\typeof{a} = \typeof{b}$ and
(2) for all $f \in \fsignature$, if $a_i \congrec{}{} b_i$ then
$f^{\algof{A}}(a_1,\ldots,a_{\arityof{f}}) \congrec{}{}
f^{\algof{A}}(b_1,\ldots,b_{\arityof{f}})$.
A congruence is \emph{locally finite} iff it has finitely many
equivalence classes of each sort. A congruence $\congrec{}{}$
\emph{saturates} a set $\mathcal{L} \subseteq \universeOf{A}$ iff
$\mathcal{L}$ is a union of equivalence classes of $\congrec{}{}$. A
set is \emph{recognizable} iff there exists a locally finite
congruence that saturates it.
\end{defi}
Any (not necessarily recognizable) set is saturated by a unique
coarsest congruence:

\begin{defi}\label{def:syntactic-congruence}
  The \emph{syntactic congruence} of a set $\mathcal{L} \subseteq
  \universeOf{A}$ in an $\fsignature$-algebra $\algof{A}$ is the
  relation $a \congrec{\algof{A}}{L} b$ defined as
  $\typeof{a}=\typeof{b}$ and $t^\algof{A}(a,c_1,\ldots,c_k) \in
  \mathcal{L} \Leftrightarrow t^\algof{A}(b,c_1,\ldots,c_k)\in
  \mathcal{L}$, for all first-order $\fsignature$-terms
  $t(x,y_1,\ldots,y_k)$ and all $c_1, \ldots, c_k \in \universeOf{A}$.
\end{defi}
The proof that $\congrec{\mathcal{L}}{\algof{A}}$ is the coarsest
congruence that saturates $\mathcal{L}$ is standard, see e.g.,
\cite[Proposition 3.66]{courcelle_engelfriet_2012}. Hence,
$\mathcal{L}$ is recognizable if and only if
$\congrec{\mathcal{L}}{\algof{A}}$ is locally finite.

For the purpose of several proofs in the paper, we introduce an
equivalent definition of recognizability using homomorphisms into
locally finite algebras. A \emph{homomorphism} between
$\fsignature$-algebras $\algof{A}$ and $\algof{B}$ is a function $h:
\universeOf{A} \rightarrow \universeOf{B}$ such that
(1) $f(\universeOf{A}^\sort) \subseteq \universeOf{B}^\sort$, for
all sorts $\sort\in\sorts$, and
(2) \(h(f^\algof{A}(a_1, \ldots, a_{\arityof{f}})) =
f^\algof{B}(h(a_1), \ldots, h(a_{\arityof{f}}))\), for all function
symbols $f \in \fsignature$ and all elements $a_1, \ldots,
a_{\arityof{f}} \in \universeOf{A}$.

\begin{defi}\label{def:rec}
  A set $\mathcal{L} \subseteq \universeOf{A}$ is \emph{recognizable}
  in $\algof{A}$ iff there exists a locally finite algebra $\algof{B}$
  and a homomorphism $h$ between $\algof{A}$ and $\algof{B}$ such that
  $\mathcal{L} = h^{-1}(\universeOf{C})$, for a set $\universeOf{C}
  \subseteq \universeOf{B}$.
\end{defi}
The equivalence with the notion of recognizability introduced above is
that the equivalence relation $\congrec{}{}$ from Definition
\ref{def:recognizability} is the kernel of the homomorphism $h$ from
Definition \ref{def:rec}, see, e.g.,~\cite[Proposition
  3.64]{courcelle_engelfriet_2012} for a proof of equivalence between
the two notions of recognizability.

Recognizability in an algebra implies recognizability in each derived
(sub)algebra, as shown by the following two lemmas:

\begin{lem}\label{lemma:derived-rec}
  Let $\algof{D}$ be a derived algebra of $\algof{A}$. Then,
  $\mathcal{L} \subseteq \universeOf{A}$ is recognizable in
  $\algof{D}$ if it is recognizable in $\algof{A}$.
\end{lem}
\begin{proof}
  Let
  $\algof{D}=(\set{\universeOf{A}^\sort}_{\sort\in\sorts},\set{t^\algof{A}}_{t\in\dsignature})$.
  For any locally finite $\fsignature$-algebra
  $\algof{B}=(\set{\universeOf{B}^\sort}_{\sort\in\sorts},\fsignature^\algof{B})$,
  any homomorphism $h$ between $\algof{A}$ and $\algof{B}$ is also a
  homomorphism between $\algof{D}$ and the derived algebra
  $\algof{D}'=(\set{\universeOf{B}^\sort}_{\sort\in\sorts},\set{t^{\algof{B}}}_{t\in\dsignature})$.
  The homomorphism $h$ and the set
  $\mathcal{C}\subseteq\universeOf{B}$ that witness the
  recognizability of $\mathcal{L}$ in $\algof{A}$ also witness the
  recognizability of $\mathcal{L}$ in $\algof{D}$.
\end{proof}
Note that the converse does not hold, for instance, if we
consider the algebra of words over the alphabet $\set{a,b}$ with
signature consisting of the empty word $\epsilon$ and concatenation. A
derived algebra is obtained by taking the empty word and the derived
operation $x \mapsto axb$. Then $\set{a^n b^n \mid n \in \nat}$ is
recognizable in the derived algebra but not in the original one.

\begin{lem}\label{lemma:sub-rec}
  Let $\algof{B}$ be a subalgebra of $\algof{A}.$ Then, $\mathcal{L}
  \subseteq \universeOf{B}$ is recognizable in $\algof{B}$ if it is
  recognizable in $\algof{A}$.
 \end{lem}
\begin{proof}
   Let $\algof{B} = (\set{\universeOf{B}^\sort}_{\sort\in\sorts'},
   \set{f^\algof{B}}_{f\in\sorts'})$ and $\algof{D} =
   (\set{\universeOf{D}^\sort}_{\sort\in\sorts},
   \fsignature^\algof{D})$ be a locally finite $\fsignature$-algebra
   and $h : \universeOf{A} \rightarrow \universeOf{D}$ be a
   homomorphism, such that $\mathcal{L}=h^{-1}(\mathcal{C})$, for some
   $\mathcal{C} \subseteq \universeOf{D}$. Let $\algof{E}$ be the
   subalgebra of $\algof{D}$ obtained by restricting $\algof{D}$ to
   the sorts $\sorts'$ and the signature $\fsignature'$. Then $h'
   \isdef \proj{h}{\universeOf{B}}$ is a homomorphism between the
   algebras $\algof{B}$ and $\algof{E}$ and $\mathcal{L} =
   h^{-1}(\mathcal{C}\cap\universeOf{E}) \cap \universeOf{B} =
   {h'}^{-1}(\mathcal{C}\cap\universeOf{E})$, which witnesses the
   recognizability of $\mathcal{L}$ in $\algof{B}$.
\end{proof}

\subsection{Context-Free Sets}
\label{sec:context-free}

Let $\fsignature$ be a functional signature and $\nonterminals
\subseteq \Vars$ be a set of second-order variables, called
\emph{nonterminals}, ranged over by $U,V$, etc. A \emph{grammar} is a
finite set of \emph{rules} of the form $U \rightarrow t$, where $U$ is
a nonterminal and $t$ is a $\fsignature$-term with variables from
$\nonterminals$.

A \emph{solution} of $\grammar$ over an $\fsignature$-algebra
$\algof{A}$ is a mapping $\mathcal{S} : \nonterminals \rightarrow
\pow{\universeOf{A}}$ such that $t^{\mathcal{S}} \subseteq
\mathcal{S}(U)$ for each rule $U \rightarrow t \in \grammar$, where
$t^{\mathcal{S}}$ denotes the evaluation of the term with regard to
the sets $\mathcal{S}(U) \subseteq \universeOf{A}$, for each variable
$U \in \nonterminals$. Since the evaluation of terms with set
variables is monotonic with regard to set containment, a least
solution exists and is unique. We denote by $\langof{U}{\grammar}$ the
component corresponding to $U$ within the least solution of
$\grammar$.

\begin{defi}\label{def:context-free}
A set $\mathcal{L}$ is \emph{context-free} iff $\mathcal{L} =
\langof{U}{\grammar}$, for a nonterminal $U\in\nonterminals$ and a
grammar $\grammar$.
\end{defi}

The following theorem generalizes a classical result, namely that the
intersection of a context-free and a regular set of words is
context-free. We restate the result in its algebraic form, also known
as the Filtering Theorem:

\begin{thmC}[{\cite[Theorem 3.88]{courcelle_engelfriet_2012}}]\label{thm:filtering-theorem}
  Let $\mathcal{L}$ be a context-free set and $\mathcal{K}$ be a
  recognizable set. Then, $\mathcal{L} \cap \mathcal{K}$ is
  context-free. Moreover, the grammar for $\mathcal{L} \cap
  \mathcal{K}$ has the same sorts as the one for $\mathcal{L}$.
\end{thmC}
Here, by the sorts of a grammar, we understand the set of sorts that
occur as argument or value sorts in each function symbol occurring in
that grammar.

\section{Graphs}
\label{sec:graphs}

This section introduces hyper-graphs with edges labeled by symbols
from a finite alphabet and distinguished source vertices, the
hyperedge replacement algebra (\autoref{sec:hr}), the subalgebras that
use finitely many source labels (\autoref{sec:sub}) and the notion of
definable sets of graphs, via the encoding of the incidence relation
of a graph by a relational structure (\autoref{sec:cmso-def}).

Let $\sourcelabels$ be a countably infinite set of \emph{source
labels} and $\alphabet$ be an alphabet of \emph{edge labels}, disjoint
from $\sourcelabels$. Each edge label $a\in\alphabet$ has an
associated \emph{arity} $\arityof{a} \geq 1$, i.e., we do not consider
edge labels of arity zero. The sets $\sourcelabels$ and $\alphabet$
are fixed in the rest of the paper.

\begin{defi}\label{def:graphs}
  Let $\slabs \finsubseteq \sourcelabels$ be a finite set of source
  labels. A \emph{concrete graph} of sort $\slabs$ is a tuple $\graph
  = \tuple{\vertof{\graph}, \edgeof{\graph}, \labof{\graph},
    \edgerelof{\graph}, \sourceof{\graph}}$,
  where: \begin{itemize}[label=$\triangleright$]
  \item $\vertof{\graph}$ is a finite set of \emph{vertices},
  \item $\edgeof{\graph}$ is a finite set of \emph{edges}, disjoint
    from $\vertof{\graph}$,
  \item $\labof{\graph} : \edgeof{\graph} \rightarrow \alphabet$ is a
    mapping that defines the labels of the edges,
  \item $\edgerelof{\graph} : \edgeof{\graph} \rightarrow
    \vertof{\graph}^+$ is a mapping that associates each edge a
    nonempty sequence of vertices attached to the edge, such that
    $\arityof{(\labof{\graph}(e))} = \lenof{\edgerelof{\graph}(e)}$,
    for each $e \in \edgeof{\graph}$,
  \item $\sourceof{\graph} : \slabs \rightarrow \vertof{\graph}$ is a
    \emph{one-to-one} mapping that designates the \emph{sources} of
    $\graph$. The vertex $\sourceof{\graph}(s)$ is called the
    $s$-source of $\graph$. Because $\sourceof{\graph}$ is injective,
    a vertex cannot be both an $s$- and $s'$-source, for $s \neq
    s'$. Vertices that are not sources are called \emph{internal}.
  \end{itemize}
  We identify concrete graphs up to isomorphism and define
  \emph{graphs} as isomorphism-equivalence classes of concrete
  graphs. We denote by $\graphs$ the set of graphs.
\end{defi}
\begin{exa}
The leftmost graph in \autoref{fig:graphs} (a) has four vertices of
which three sources labeled $s_1$, $s_2$ and $s_3$ and three edges
labeled $a$, $b$ and $c$. The $a$-labeled edge is attached to three
vertices, whereas the $b$- and $c$-labeled edges are binary. The
arrows on the edges indicate the order of the vertices attached to
them. The middle graph is of sort $\set{s_1,s_2,s_4}$ and the
rightmost one of sort $\set{s_1,s_2,s_3,s_4}$.
\end{exa}

\begin{figure}[t!]
  \centerline{\hspace{5pt}\scalebox{1.1}{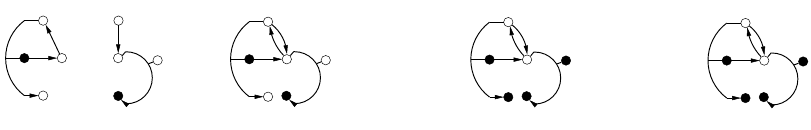}}
  \caption{Composition (a), Restriction (b) and Renaming (c) of
    Graphs. Sources are denoted by hollow and internal vertices by
    solid circles. Arrows indicate the order of vertices attached to the corresponding edges.}
  \label{fig:graphs}
\end{figure}

\subsection{The Hyperedge Replacement Algebra}
\label{sec:hr}
We introduce the \emph{hyperedge replacement} (HR) algebra of
operations on graphs. There are several equivalent definitions of this
algebra in the literature. For instance, the definition from
\cite{CourcelleV} uses operations described by graphs with sources
denoted by $\set{1, \ldots, k}$, having designated edges $e_i$, for $i
\in \interv{1}{n}$, that are deleted and replaced by graphs $\graph_i$
of sorts $\set{1, \ldots, \arityof{\labof{\graph_i}}}$,
respectively. These operations are strongly typed. Instead, we
consider an algebra over a signature of polymorphic operations, having
the same expressivity \cite[Definition
  2.32]{courcelle_engelfriet_2012}. To ensure compliance with the
general definition of multi-sorted algebras from
\autoref{sec:recognizable}, these polymorphic operations can be
understood as families of sorted operations, one for each choice of
the argument sorts.

We fix the set of sorts $\hrsorts$ to be the set of finite subsets of
$\sourcelabels$. The signature $\hrsignature$ consists of the
constants $\emptygraph_\slabs$, for all $\slabs \finsubseteq
\sourcelabels$ and $\sgraph{a}_{(s_1, \ldots, s_{\arityof{a}})}$, for
all $a\in \alphabet$ and $s_1, \ldots, s_{\arityof{a}} \in
\sourcelabels$, the unary function symbols $\restrict{\slabs}$, for
all $\slabs \finsubseteq \sourcelabels$, and $\rename{\alpha}$, for
all finite permutations $\alpha: \sourcelabels \rightarrow
\sourcelabels$ and the binary function symbol $\pop$. The \emph{graph
algebra} $\algebra$ interprets the symbols in $\hrsignature$ as
follows: \begin{enumerate}
\item \textbf{sources only}: the graph $\emptygraph^{\algebra}_\slabs$
  consists of one $s$-source for each $s\in\slabs$ and no edges.
\item \textbf{single edge}: the graph $\sgraph{a}^{\algebra}_{(s_1,
  \ldots, s_{\arityof{a}})}$ consists of an $s_i$-source, for each $i
  \in \interv{1}{\arityof{a}}$, and a single edge labeled with $a$
  attached to the $s_1,\ldots,s_{\arityof{a}}$-sources, in this order.
\item \textbf{restriction}: the unary function
  $\restrict{\slabs}^{\algebra}$ takes as input a graph of sort
  $\slabs'$ and returns the graph of sort $\slabs\cap\slabs'$ obtained
  by removing the source labels in $\slabs' \setminus \slabs$ from
  $\graph$. Note that the vertices from $\img{\slabs}$ are not removed
  in the output graph. Formally, each concrete graph $\graph$ is
  mapped into the concrete graph $\tuple{\vertof{\graph},
    \edgeof{\graph}, \labof{\graph}, \edgerelof{\graph},
    \proj{\sourceof{\graph}}{\slabs}}$ and
  $\restrict{\slabs}^{\algebra}$ is defined as the lifting of this
  operation from concrete graphs to graphs.
\item \textbf{rename:} the unary function $\rename{\alpha}^{\algebra}$
  takes as input a graph of sort $\slabs$ and returns the graph of
  sort $\alpha^{-1}(\slabs)$ obtained by renaming its sources
  according to $\alpha$. Formally, each concrete graph $\graph$ is
  mapped into the concrete graph $\tuple{\vertof{\graph},
    \edgeof{\graph}, \labof{\graph}, \edgerelof{\graph},
    \sourceof{\graph} \circ \alpha}$ and $\rename{\alpha}^{\algebra}$
  is defined as the lifting of this operation from concrete graphs to
  graphs.
\item \textbf{composition:} the binary function $\pop^{\algebra}$
  takes the disjoint union of two graphs of sorts $\slabs_1$ and
  $\slabs_2$ and fuses the vertices labeled by the same source label
  in both. The result is a graph of sort $\slabs_1 \cup
  \slabs_2$. Formally, let $\graph_i$ be concrete graphs of sort
  $\slabs_i$, for $i=1,2$, such that $\vertof{\graph_1} \cap
  \vertof{\graph_2} = \emptyset$ and $\edgeof{\graph_1} \cap
  \edgeof{\graph_2} = \emptyset$. Let $\sim \ \subseteq
  (\vertof{\graph_1} \cup \vertof{\graph_2})^2$ be the least
  equivalence relation such that $u_1 \sim u_2$ if $u_i =
  \sourceof{\graph_i}(s)$, for $i=1,2$ and $s \in \slabs_1 \cap
  \slabs_2$. Then the composition $\graph_{12}$ of $\graph_1$ with
  $\graph_2$ is defined as
  follows: \begin{itemize}[label=$\triangleright$]
    \item $\vertof{\graph_{12}} = \{[u]_\sim \mid u \in
      \vertof{\graph_1} \cup \vertof{\graph_2} \}$,
      $\edgeof{\graph_{12}} = \edgeof{\graph_1} \cup
      \edgeof{\graph_2}$ and
      $\labof{\graph_{12}}\isdef\labof{\graph_1}\cup\labof{\graph_2}$,
    \item $\edgerelof{\graph_{12}}(e) \isdef \tuple{[u_1]_\sim,
      \ldots, [u_k]_\sim}$ for every edge $e \in \edgeof{\graph_i}$,
      such that $\edgerelof{\graph_i}(e) = \tuple{u_1, \ldots, u_k}$,
    \item
      $\sourceof{\graph_{12}}(s) \isdef [\sourceof{\graph_i}(s)]_\sim$
      iff $s \in \slabs_i$, for $i=1,2$.
    \end{itemize}
    where $[u]_\sim$ is the $\sim$-equivalence class of the vertex
    $u$. Then, $\pop^{\algebra}$ is the lifting of this binary
    operation from concrete graphs to graphs.
\end{enumerate}
\begin{exa}
For example, \autoref{fig:graphs} (a) shows the result of the
composition of two graphs, whereas (b) and (c) show the result of
applying restriction and renaming to this composition, respectively.
By $\set{i \leftrightarrow j}$ we denote the finite permutation on
$\nat$ that swaps $i$ with $j$ and maps every $k \in \nat \setminus
\set{i,j}$ to itself.
\end{exa}

\subsection{Hyperedge Replacement Subalgebras}
\label{sec:sub}
For each $\slabs\finsubseteq\sourcelabels$, let $\graphsof{\slabs}
\isdef \{\graph \in \graphs \mid \typeof{\graph}\subseteq\slabs\}$ be
the set of graphs of sort included in $\slabs$ and $\algebra^\slabs$
be the subalgebra of $\algebra$ with the universe $\graphsof{\slabs}$
and the finite signature of operations that use only source labels
from $\slabs$:

\begin{align*}
  \hrsignature^\slabs \isdef &
  \set{\emptygraph_{\slabs'}}_{\slabs' \subseteq \slabs} \cup
  \set{\sgraph{a}_{(s_1, \ldots, s_{\arityof{a}})}}_{\hspace*{-2mm}\begin{array}{l}
      \scriptstyle{a \in \alphabet} \\[-2mm]
      \scriptstyle{s_1, \ldots, s_{\arityof{a}} \in \slabs}
  \end{array}} \hspace*{-8mm} \cup
  \set{\restrict{\slabs'}}_{\slabs' \subseteq \slabs} \cup
  \set{\rename{\alpha}}_{\alpha~ \slabs\text{-permutation}} \cup \set{\pop}
\end{align*}
We denote by $\btwgraphsof{\slabs}$ the set of term-generated elements
of $\algebra^{\slabs}$ and by $\btwalgebraof{\slabs}$ the
term-generated subalgebra of $\algebra^{\slabs}$. Note that, while all
elements of the graph algebra $\algebra$ are term-generated, each
term-generated algebra $\btwalgebraof{\slabs}$ is a strict subalgebra
of $\algebra^{\slabs}$.

\subsection{Definable Sets of Graphs}
\label{sec:cmso-def}

In order to describe sets of graphs using \cmso, we encode graphs as
relational structures over finite relational signatures.  To this end,
we consider the alphabet $\alphabet$ of edge labels to be
finite. Given a sort $\slabs\finsubseteq\sourcelabels$, we define the
relational signature $\relsiggraph{\slabs} \isdef \set{\arel_a \mid a
  \in \alphabet} \cup \set{\arel_s \mid s \in \slabs}$ whose relation
symbols have the arities $\arityof{\arel_a} \isdef \arityof{a}+1$, for
all $a\in\alphabet$, and $\arityof{\arel_s} \isdef 1$, for all
$s\in\slabs$. Note that the signature $\relsiggraph{\slabs}$ is finite
because both $\alphabet$ and $\slabs$ are finite. The \emph{encoding}
of a concrete graph $\graph \in \graphsof{\slabs}$ is the structure
$\gstruc{\graph} = (\vertof{\graph} \cup
\edgeof{\graph},\struc_\graph) \in \strucof{\relsiggraph{\slabs}}$,
where:
\begin{align*}
  \struc_\graph(\arel_a) \isdef & \set{(e,v_1, \ldots, v_{\arityof{a}}) \mid e
  \in \edgeof{\graph},~ \labof{\graph}(e) = a,~ \edgerelof{\graph}(e) = (v_1,
  \ldots, v_{\arityof{a}})} \text{, for all } a \in \alphabet
  \\
  \struc_\graph(\arel_s) \isdef & \set{\sourceof{\graph}(s)} \text{, for all } s \in \slabs
\end{align*}
We note that the encodings of isomorphic concrete graphs are
isomorphic structures. The encoding of a graph denotes the
isomorphism-equivalence class obtained from the encodings of the
concrete graphs in the graph.

This type of encoding, where edges are elements of the universe and
the graph is described by the incidence relation between edges and
vertices, is known as the \emph{incidence encoding}. Another encoding
used in the literature is the \emph{edge encoding}, where the universe
consists of vertices only and edges are tuples from the interpretation
of the relations corresponding to the labels. The expressiveness of
\cmso\ differs in the two encodings, e.g., the existence of a
Hamiltonian cycle can be described in \cmso\ using the incidence but
not the edge encoding~\cite[Proposition
  5.13]{courcelle_engelfriet_2012}. Moreover, key results used in the
paper~\cite[Theorems 1.10 and 2.1]{CourcelleEngelfriet95} hold under
the incidence but not the edge encoding.

\begin{defi}\label{def:definable}
A set of graphs $\mathcal{L} \subseteq \graphsof{\slabs}$ is
\emph{definable} if there exists a \cmso\ formula $\phi$ over
the relational signature $\relsiggraph{\slabs}$ such that
$\gstruc{\mathcal{L}} = \set{\astruc \mid \astruc \Models\phi}$.
\end{defi}
Note that a set of graphs that is not included in $\graphsof{\slabs}$,
for any finite $\slabs\finsubseteq\sourcelabels$, is not definable,
because a \cmso\ formula can only speak of finitely many relation
symbols. We recall the following result~\cite[Theorem
  4.4]{CourcelleI}:

\begin{thmC}[\cite{CourcelleI}]\label{thm:graphs-def-rec}
  Any definable set of graphs is recognizable in $\algebra$.
\end{thmC}

\section{Trees}
\label{sec:trees}

We deviate from the definition of trees used in the classical
literature on tree automata (see, e.g., \cite{comon:hal-03367725} for
a survey) and define trees as the term-generated elements of a suitable
derived subalgebra of graphs. In contrast to the standard definition,
where the number and order of the children of a node in the tree is
determined by the label of that node, our definition encompasses also
trees whose nodes have an unbounded number of children, such that,
moreover, the order of siblings is not important. As explained below,
this definition of trees is strictly more general than the classical
definition of terms over ranked
alphabets~\cite{comon:hal-03367725}. This generality is needed,
because trees are used in the definition of tree decompositions, that
impose no restrictions on the number or the order of siblings.

Let $\alphabetTwo \subseteq \alphabet$ be a set of edge
labels of arities at most two and $\aroot$ be a source label. We use
$c$ and $b$ to denote the symbols of arities one and two of
$\alphabetTwo$, respectively. The $\aroot$-source of a graph is called
its \emph{root}. The signature of trees with $\alphabetTwo$-labeled
edges is the following set of function symbols:
\begin{align*}
  \treesignature(\alphabetTwo) \isdef & \set{\sgraph{c} \mid c \in
    \alphabetTwo, \arityof{c}=1} \cup \set{\extend{b} \mid b \in
    \alphabetTwo, \arityof{b}=2} \cup \set{\pop}
\end{align*}
where $\sgraph{c}$ are constants (we omit specifying the $\aroot$
source label) and $\extend{b}$ are unary. The constants $\sgraph{c}$
are interpreted as trees consisting of a single root vertex, whereas
$\extend{b}$ and $\pop$ are interpreted as in \autoref{fig:trees}.
Formally, $\extend{b}(x) \isdef \rename{\aroot \leftrightarrow
  \temp}(\restrict{\{\temp\}}(\sgraph{b}_{(\temp,\aroot)} \pop x))$,
where $\temp$ is an auxiliary source label used only here and $\aroot
\leftrightarrow \temp$ is the permutation that switches $\aroot$ with
$\temp$. We denote by $\treealgebra(\alphabetTwo)$ the term-generated
algebra whose universe $\trees(\alphabetTwo)$ (of a single sort
$\set{\aroot}$) is the set of interpretations of the ground
$\treesignature(\alphabetTwo)$-terms in $\algebra$.

The standard terminology for trees is immediately retrieved from the
above definition. The vertices of a tree $\tree$ are called
\emph{nodes}. For a binary edge $e \in \edgeof{\tree}$, we say that
$\edgerelof{\tree}(e)_1$ is the \emph{parent} of
$\edgerelof{\tree}(e)_2$ and $\edgerelof{\tree}(e)_2$ is a
\emph{child} of $\edgerelof{\tree}(e)_1$. A node with no children is
called a \emph{leaf}. Since trees are interpretations of ground terms,
each leaf corresponds to the interpretation of at least some constant
$\sgraph{c}$, thus is attached to a unary edge labeled with $c$.  The
\emph{rank} of a tree is the maximum number of children of a node. A
set of trees is \emph{ranked} if the corresponding set of ranks is
finite and \emph{unranked}, otherwise. In particular, the set
$\trees(\alphabetTwo)$ is unranked.

\begin{figure}[t!]
  \centerline{\scalebox{1.1}{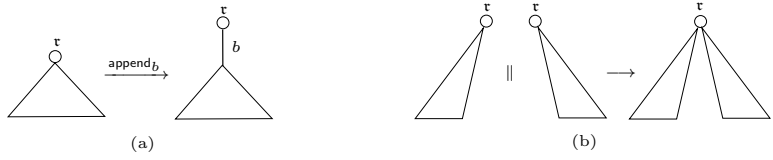}}
  \caption{Append (a) and Composition (b) of Trees. The $\aroot$-sources are denoted by hollow circles. }
  \label{fig:trees}
\end{figure}

\subsection{Definable Sets of Trees}
Because trees are graphs, the encoding of trees is no different from
that of graphs. We consider a finite set $\alphabetTwo$ of edge labels
and define the relational signature $\relsigtree{}(\alphabetTwo)
\isdef \set{\arel_b \mid b \in \alphabetTwo} \cup \set{\arel_\aroot}$,
where $\aroot$ is the singleton source label associated with the
root. Interestingly, the expressiveness of \cmso\ using incidence and
the edge encodings coincide, when trees are considered, instead of
graphs (a consequence of the Sparseness Theorem~\cite[Theorem
  1.44]{courcelle_engelfriet_2012}).

A classical result of Courcelle~\cite{CourcelleI} is the equivalence
between definability and recognizability for (possibly unranked) sets
of trees:

\begin{thmC}[\cite{CourcelleI}]\label{thm:unranked-trees-rec-def}
  For any finite alphabet $\alphabetTwo$, a set of trees $\mathcal{K}
  \subseteq \trees(\alphabetTwo)$ is recognizable in
  $\treealgebra(\alphabetTwo)$ iff $\mathcal{K}$ is definable.
\end{thmC}

As a consequence of Theorem \ref{thm:unranked-trees-rec-def}, we
obtain the equivalence of the recognizability of a set of trees in the
graph and tree algebras:

\begin{cor}\label{cor:trees-graph-rec}
  For each finite alphabet $\alphabetTwo\subseteq\alphabet$, a set of
  trees is recognizable in $\treealgebra(\alphabetTwo)$ iff it is
  recognizable in $\algebra$.
\end{cor}
\begin{proof}
  ``$\Rightarrow$'' By Theorem \ref{thm:unranked-trees-rec-def}, if a
  set of trees is recognizable in $\treealgebra(\alphabetTwo)$ then it
  is definable. By Theorem \ref{thm:graphs-def-rec}, any definable set
  of graphs is recognizable in $\algebra$.  ``$\Leftarrow$'' By Lemmas
  \ref{lemma:derived-rec} and \ref{lemma:sub-rec}, because
  $\treealgebra(\alphabetTwo)$ is a derived subalgebra of $\algebra$.
\end{proof}

As a remark, each definable ranked set of trees can be described using
the \mso{} fragment of \cmso, by~\cite[Proposition
  6.11]{CourcelleI}. The distinction between unranked and ranked sets
of trees is formally established by the fact that there exists sets of
trees definable in \cmso{} that are not \mso-definable~\cite[Corollary
  6.6]{CourcelleI}. For instance, the set of trees of height one
having an even number of leaves is \cmso-definable but not
\mso-definable.
In other words, \cmso{} is strictly more expressive that \mso.

\subsection{Tree Decompositions}
\label{sec:tree-decompositions}

A \emph{tree decomposition} is another way of encoding a graph as a
tree whose edges are all labeled by a binary label $\parent$.
A set of nodes $C \subseteq \vertof{\tree}$ is \emph{connected} in
$\tree$ iff between any two nodes in $C$ there exists an undirected
path of edges from $\edgeof{\tree}$ that traverses only nodes from
$C$. Tree decompositions are used to formalize the notion of
\emph{tree-width}:

\begin{defi}\label{def:tree-decomposition}
    A \emph{tree decomposition} of a concrete graph $\graph$ of sort $\slabs$ is a pair
    $(\tree,\beta)$, where $\tree \in \trees(\parent)$ and $\beta :
    \vertof{\tree} \rightarrow \pow{\vertof{\graph}}$ is a mapping,
    such that:~\begin{enumerate}
  \item\label{it1:tree-decomposition} for each edge $e \in
    \edgeof{\graph}$ there exists a node $n \in \vertof{\tree}$, such
    that $\edgerelof{\graph}(e)_i \in \beta(n)$, for all $1 \le i \le
    \arityof{\labof{\graph}(e)}$,
  \item\label{it2:tree-decomposition} for each vertex $v \in
    \vertof{\graph}$, the set $\bagof{v}{\tree} \isdef \set{n \in
      \vertof{\tree} \mid v \in \beta(n)}$ is nonempty and connected
    in $\tree$.
  \item\label{it3:tree-decomposition} the bag of the root $\aroot$ of $\tree$ contains all sources of $\graph$, i.e., we have $\sourceof{\graph}(s) \in \beta(\sourceof{\tree}(\aroot))$ for all sources $s \in \slabs$.
  \end{enumerate}
  The \emph{width} of the tree decomposition is $\width{\tree,\beta}
  \isdef \max\set{\cardof{\beta(n)} \mid n \in \vertof{\tree}}-1$ and
  the \emph{tree-width} of $\graph$ is $\twd{\graph} \isdef \min
  \set{\width{\tree,\beta} \mid (\tree,\beta) \text { is a tree
      decomposition of } \graph}$. The tree-width of a graph is the
  tree-width of any concrete graph from the isomorphism equivalence
  class (isomorphic concrete graphs have the same tree-width). We
  denote by $\graphsof{\le k}$ the set of graphs $\graph$ such that
  $\twd{\graph} \le k$.
\end{defi}

At this point, we must stress the importance of considering unranked
sets of trees, when reasoning about the tree decompositions of a set
of graphs. In particular, the tree decompositions extracted by the
definable transduction whose existence is stated in~\cite[Theorem
  2.4]{10.1145/2933575.2934508} have no bound on the number of
children of a node. Since our results crucially depend on the
existence of such a transduction, we must consider the unranked sets
of trees, obtained from these tree decompositions, in their full
generality.

We assume basic acquaintance with the notion of grid and the fact that
an $n\times n$ square grid has tree-width $n$
\cite{DBLP:journals/tcs/Bodlaender98}. To see the difference between
the sets $\graphsof{\slabs}$ and $\btwgraphsof{\slabs}$, note that a
$n\times n$ square grid with no sources belongs to $\graphsof{\slabs}$
but not to $\btwgraphsof{\slabs}$, for any
$\slabs\finsubseteq\sourcelabels$ such that $\cardof{\slabs} < n+1$.
The following result shows the fundamental difference between graphs
and \emph{term-generated} graphs with sources from a given finite
set. We restate it here using our notation, for self-containment:

\begin{thmC}[{\cite[Theorem 2.83]{courcelle_engelfriet_2012}}]\label{lemma:cf-twb}
  Let $\slabs \finsubseteq \sourcelabels$ be a sort.
  For each graph $\graph \in \btwgraphsof{\slabs}$, we have $\twd{\graph} \leq
  \cardof{\slabs} - 1$.
\end{thmC}

The encoding of tree decompositions as relational structures uses the
relational signature $\relsigdecomp{\slabs} \isdef
\relsiggraph{\slabs} \cup \set{\node,\parent,\bag}$, where $\node$ is
a unary relation symbol and $\parent$, $\bag$ are binary relation
symbols, respectively.
We encode triples $(\graph,\tree,\beta)$, where $G \in \graphsof{\slabs}$ is a concrete graph and $(\tree,\beta)$ is a tree decomposition of $\graph$, such that $\tree$ is a tree having one
binary edge label $\parent$\footnote{In addition to unary labels of
leafs, that are not important here.}.
We encode $\graph$ by a
structure $\gstruc{\graph} = (\univ_\graph,\struc_\graph)$ over the relational
signature $\relsiggraph{\slabs}$. The tree decomposition is encoded by
an extended structure $\gstruc{\graph,\tree,\beta} \isdef (\univ_\graph \cup
\vertof{\tree},\struc'_\graph)$, where $\struc'_\graph$ agrees with $\struc_\graph$ over
$\relsiggraph{\slabs}$ and the unary and binary relation symbols
$\node$, $\parent$ and $\bag$ are interpreted as follows:
\begin{align*}
  \struc'(\node) \isdef & ~\vertof{\tree} \\
  \struc'(\bag) \isdef & ~\set{(v,n) \in \vertof{\graph} \times \vertof{\tree} \mid v \in \beta(n)} \\
  \struc'(\parent) \isdef & ~\set{(n,m) \in \vertof{\tree} \times \vertof{\tree} \mid
    \exists e \in \edgeof{\tree} ~.~ \edgerelof{\tree}(e)=(n,m)}
\end{align*}
Note that using the edge encoding to represent trees is without loss
of generality, because each edge of a tree with binary edges is
represented by its unique target node.

\section{Context-Free Sets of Graphs}
\label{sec:context-free-graphs}

As one expects, a \emph{context-free set of graphs} is a component of
the least solution of a grammar written using the operations from the
signature $\hrsignature$. A prominent member of this class is the set
of graphs that are values of $\hrsignature$-terms having finitely many
sorts:

\begin{prop}\label{prop:btw-graphs-are-context-free}
  For each sort $\slabs\finsubseteq\sourcelabels$, the set of graphs
  $\btwgraphsof{\slabs}$ is context-free.
\end{prop}
\begin{proof}
  We fix a sort $\slabs\finsubseteq\sourcelabels$.  Let $\grammar$ be
  the grammar having a single non-terminal $X$ and the rules:
  \begin{align*}
    X \rightarrow & ~\emptygraph_{\slabs'} \text{, for each } \slabs' \subseteq \slabs \\
    X \rightarrow & ~\sgraph{a}_{(s_1, \ldots, s_{\arityof{a}})} \text{, for each } a\in \alphabet \text{ and } s_1,\ldots,s_{\arityof{a}} \in \slabs \\
    X \rightarrow & ~\restrict{\slabs'}(X) \text{, for each } \slabs' \subseteq \slabs \\
    X \rightarrow & ~\rename{\alpha}(X) \text{, for each $\slabs$-permutation $\alpha$} \\
    X \rightarrow & ~X \pop X
  \end{align*}
  As all rules of $\grammar$ only use
  $\hrsignature^\slabs$-operations, we clearly have
  $\langof{X}{\grammar} \subseteq \btwgraphsof{\slabs}$. On the other
  hand, every ground term of $\btwgraphsof{\slabs}$ can be constructed
  by the rules of this grammar.  Hence, $\langof{X}{\grammar}
  \supseteq \btwgraphsof{\slabs}$.
\end{proof}

In the rest of this section, we recall results relating a graph
produced by a grammar to the trees that describe the partial order in
which the grammar rules are applied in order to produce that particular
graph. These notions mirror standard concepts used in word grammars,
such as \emph{parse trees} and \emph{yields} (i.e., words obtained by
reading the symbols of the leaves of a parse tree in the lexicographic
order).

\subsection{Parse Trees}
\label{sec:parse-trees}

We define parse trees by fixing an alphabet $\alphabetParse$ of unary
and binary edge labels and the functional signature $\parsesignature$
to be the following sets:
\begin{align*}
  \alphabetParse \isdef & ~\set{\emptygraphSymb_\slabs}_{\slabs \finsubseteq \sourcelabels}
  \cup \set{\sgraphSymb{a}_{(s_1, \ldots, s_{\arityof{a}})}}_{\hspace*{-2mm}\begin{array}{l}
      \scriptstyle{a\in\alphabet} \\[-2mm]
      \scriptstyle{s_1,\ldots,s_{\arityof{a}} \in \sourcelabels}
  \end{array}} \hspace*{-8mm}
  \cup \set{\restrictSymb{\slabs}}_{\slabs \finsubseteq \sourcelabels}
  \cup \set{\renameSymb{\alpha}}_{
      \scriptstyle{\alpha \text{ finite permutation}} 
  }
  \\[-2mm]
  \parsesignature \isdef & ~\treesignature(\alphabetParse)
\end{align*}
The edge labels $\emptygraphSymb_\slabs$ and $\sgraphSymb{a}_{(s_1,
  \ldots, s_{\arityof{a}})}$ are unary, whereas
$\restrictSymb{\slabs}$ and $\renameSymb{\alpha}$ are binary. We
further assume that $\alphabetParse$ is disjoint from the alphabet
$\alphabet$ of edge labels for graphs\footnote{The initial assumption
$\alphabetTwo \subseteq \alphabet$ does not work for $\alphabetParse$
because then $\alphabet$ would need to contain the symbols
$\sgraphSymb{a}_{(s_1, \ldots, s_{\arityof{a}})}$, for all $a \in
\alphabet$, etc. To avoid this confusion, we chose to keep
$\alphabetParse$ and $\alphabet$ disjoint.}.

Note that $\alphabetParse$ and $\parsesignature$ are infinite, because
$\sourcelabels$ is infinite. In cases where finite restrictions of
these sets are necessary, we denote by $\alphabetParseOf{\slabs}$
(resp.\  $\parsesignatureof{\slabs}$) the finite restriction of
$\alphabetParse$ (resp.\  $\parsesignature$) to symbols that use only
source labels from the finite set $\slabs \finsubseteq \sourcelabels$.

The \emph{algebra of parse trees} $\repalgebra$ is a single-sorted
$\hrsignature$-algebra whose universe is the set of trees
$\trees(\alphabetParse)$, i.e., with edge labels from
$\alphabetParse$, and the function symbols from $\hrsignature$ are
interpreted over trees instead of graphs, as follows:
\begin{align*}
  \restrict{\slabs}^\repalgebra(\tree) \isdef & ~\extend{\restrictSymb{\slabs}}^{\treealgebra(\alphabetParse)}(\tree) \\[1mm]
  \rename{\alpha}^\repalgebra(\tree) \isdef & ~\extend{\renameSymb{\alpha}}^{\treealgebra(\alphabetParse)}(\tree)
  \text{, for all } \tree\in\trees(\alphabetParse)
\end{align*}
The rest of the function symbols from $\hrsignature$ are interpreted
in $\repalgebra$ the same as in $\algebra$.

\begin{exa}
\autoref{fig:val} shows a $\hrsignature$-term (a) and its parse tree
(b). Intuitively, a parse tree is obtained from a $\hrsignature$-term
by (i) moving the node labels to the incident edges and (ii)
collapsing all of the adjacent $\pop$-nodes in the term to one
node. 
\end{exa}

\begin{figure}[t!]
  \centerline{\hspace{20pt}\scalebox{1.1}{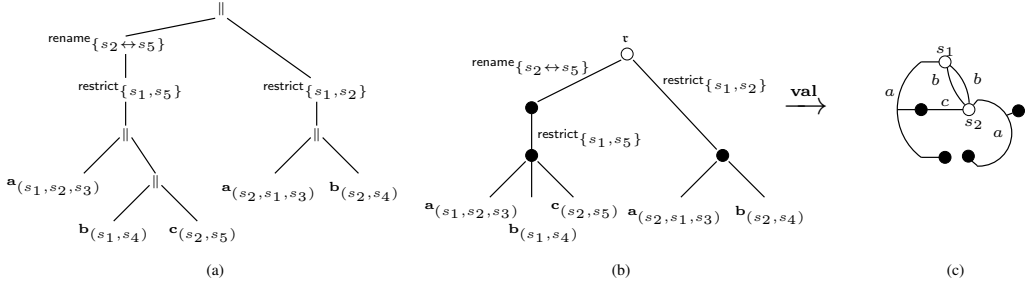}}
  \caption{A $\hrsignature$-term (a), its corresponding parse tree (b)
    and the canonical evaluation (c). Sources are denoted by hollow
    and internal vertices by solid circles.}
  \label{fig:val}
\end{figure}

Since the tree algebra $\treealgebra(\alphabetParse)$ is
term-generated, each tree is the value of a ground $\hrsignature$-term
obtained from a ground $\parsesignature$-term, by replacing each
constant $\underline{\sgraph{c}}$ with $\sgraph{c}$ and each unary
function symbol $\extend{\underline{b}}$ with $b$. Then,
recognizability is equivalent in the tree and the parse tree algebras:

\begin{lem}\label{lemma:red-rec}
  A set $\mathcal{K} \subseteq \trees(\alphabetParse)$ is recognizable
  in $\treealgebra(\alphabetParse)$ iff it is recognizable in
  $\repalgebra$.
\end{lem}
\begin{proof}
  ``$\Rightarrow$'' Assume that $\mathcal{K} \subseteq \trees(\alphabetParse)$ is
  recognizable in $\treealgebra(\alphabetParse)$ and let
  $\algof{B}=(\set{\universeOf{B}},\parsesignature^{\algof{B}})$ be a
  locally finite $\parsesignature$-algebra and $h : \trees(\alphabetParse) \rightarrow
  \universeOf{B}$ be a homomorphism between $\treealgebra(\alphabetParse)$ and
  $\algof{B}$, such that $\mathcal{K}=h^{-1}(\mathcal{C})$, for some
  set $\mathcal{C}\subseteq\universeOf{B}$. We consider the locally
  finite $\hrsignature$-algebra
  $\algof{D}=(\set{\universeOf{B}},\hrsignature^{\algof{D}})$,
  where: \begin{align*}
    \sgraph{c}^\algof{D} \isdef & ~\sgraph{c}^\algof{B} \text{, for all constants } \sgraph{c}\in\hrsignature \\
    b^\algof{D} \isdef & ~\extend{\underline{b}}^\algof{B} \text{, for all unary function symbols } b\in\hrsignature \\
    \pop^\algof{D} ~\isdef & ~\pop^\algof{B}
  \end{align*}
  Then, $h$ is a homomorphism between $\repalgebra$ and $\algof{D}$ and
  $\mathcal{K}=h^{-1}(\mathcal{C})$ witnesses the recognizability of
  $\mathcal{K}$ in $\repalgebra$. ``$\Leftarrow$'' This direction uses
  a symmetric reasoning.
\end{proof}

\subsection{Canonical Evaluation of Parse Trees}
\label{sec:canonical-evaluation}
Parse trees are representations of graphs. The function $\hval :
\trees(\alphabetParse) \rightarrow \graphs$ yields the graph
represented by a parse tree. $\hval$ is defined inductively on the
structure of the parse tree, by interpreting the edge labels as
operation in the graph algebra, in the obvious way. In particular, if
the root of the tree has two or more children, the graph is the
composition of the graphs obtained by the evaluation of the subtrees
rooted in the children. Note that, because the composition is
associative and commutative, the result of $\hval$ does not depend on
the order of the children in the parse tree. Henceforth, we refer to
$\hval$ as the \emph{canonical evaluation} function. For example,
\autoref{fig:val} (c) shows the result of the canonical evaluation of
the parse tree (b).

Formally, the function $\hval$ yielding the graph represented by a
parse tree $t^\repalgebra$, for a ground $\hrsignature$-term $t$, is
defined inductively on the structure of the underlying term as the
unique function that satisfies $\hval(t^\repalgebra) = t^\algebra$.
To show that the function $\hval$ exists and is unique, we note that
the interpretation $\pop^\algebra$ of the composition operation is
associative and commutative, hence its value on a given parse tree
does not depend on the particular order of the subterms composed via
$\pop$ in the $\hrsignature$-term whose value that tree is. The
following proposition formalizes this statement, for trees in general:

\begin{prop}\label{prop:unique-representations}
Every tree $\tree$ of height $n\geq0$ is the value in
$\treealgebra(\alphabetTwo)$ of a $\treesignature(\alphabetTwo)$-term
of the form $t=\big(\pop_{i\in I} \sgraph{c}_i \big) \pop \big(
\pop_{j \in J} \extend{b_j}(t_j) \big)$, such that $I=\emptyset$ implies
$J\neq\emptyset$ and $t_j^{\treealgebra(\alphabetTwo)}$ are trees of
height strictly less than $n$, for all $j \in J$. Moreover, all
$\treesignature(\alphabetTwo)$-terms that represent the same tree are
equal up to the commutativity and associativity of composition.
\end{prop}
\begin{proof}
  By induction on the height $n\geq0$ of $\tree$. For the base case
  $n=0$, the tree consists of a single root node attached to one or
  more edges labeled with unary symbols $c_i \in \alphabetTwo$, for $i
  \in I$. Then there exists a term $\pop_{i\in I} \sgraph{c}_i$ that
  represents $\tree$ and this term is unique modulo the commutativity
  and associativity of composition. For the induction step $n\geq1$,
  the root of $\tree$ has children $\tree_j$ of height strictly less
  than $n$, for $j \in J$, where $J$ is a nonempty set. By the
  inductive hypothesis, there exist terms $t_j$ such that
  $t_j^{\treealgebra(\alphabetTwo)}=\tree_j$, for all $j \in
  J$. Moreover, these terms are unique modulo commutativity and
  associativity of the composition. Let $c_i \in \alphabetTwo$, $i \in
  I$ be the unary labels of the root of $\tree$ and $b_j \in
  \alphabetTwo$, $j \in J$ be the binary labels of the edges to which
  the root of $\tree$ is attached. We consider the term $t \isdef
  \left(\pop_{i\in I} \sgraph{c}_i \right) \pop \big( \pop_{j \in J}
  \extend{b_j}(t_j) \big)$. Then,
  $t^{\treealgebra(\alphabetTwo)}=\tree$ and, moreover, any other term
  $u$ such that ${u}^{\treealgebra(\alphabetTwo)}=\tree$ differs from
  $t$ by a permutation of $c_i$ and $t_j$.
\end{proof}

We prove below that recognizability of a set of graphs in a
finite-sorted subalgebra $\algebra^\slabs$ is preserved under inverse
canonical evaluations. Note that this is not a consequence of the
standard closure of recognizable sets under inverse homomorphisms,
because $\hval$ is not a homomorphism between $\repalgebra$ and
$\algebra$. In fact, no such homomorphism exists because $\repalgebra$
has one sort $\set{\aroot}$, whereas $\algebra$ has infinitely many
sorts, i.e., all the finite subsets of $\sourcelabels$.

\begin{lem}\label{lemma:finite-sort-inverse-homomorphism}
  Let $\slabs \finsubseteq \sourcelabels$ be a sort. For each set
  $\mathcal{L}$ of graphs, $\hval^{-1}(\mathcal{L})$ is recognizable
  in $\repalgebra$ if $\mathcal{L}$ is recognizable in
  $\algebra^\slabs$.
\end{lem}
\begin{proof}
  Let $\algof{B} = (\set{\universeOf{B}^{\slabs'}}_{\slabs' \subseteq
    \slabs}, \set{f^\algof{B}}_{f \in \hrsignature})$ be a locally
  finite algebra and $h : \graphsof{\slabs} \rightarrow
  \universeOf{B}$ be a homomorphism between $\algebra^\slabs$ and
  $\algof{B}$, such that $\mathcal{L}=h^{-1}(\mathcal{C})$, for a set
  $\mathcal{C} \subseteq \universeOf{B}$. Let $\algof{B}' =
  (\set{\universeOf{B}},\set{f^\algof{B}}_{f \in \hrsignature})$ be
  the algebra with a single sort $\set{\aroot}$ and finite
  $\set{\aroot}$-universe consisting of the union of all
  ${\universeOf{B}}^{\slabs'}$, for $\slabs' \subseteq \slabs$.  Then,
  $h\circ\hval$ is a homomorphism between $\repalgebra$ and
  $\algof{B}'$ and, moreover, $\hval^{-1}(\mathcal{L}) = (h \circ
  \hval)^{-1}(\mathcal{C})$.
\end{proof}

\subsection{Definable Transductions Define Context-Free Sets of Graphs}
\label{sec:context-free-transduction}
A classical result is that each context-free word language is obtained
from a recognizable ranked set of trees by reading the word on the
frontier of each tree from left to right. This result has been
generalized to context-free sets of graphs, by noticing that the
\emph{yield} that produces the word corresponding to the frontier of a
tree is a definable transduction. The result of Courcelle and
Engelfriet~\cite[Theorems 1.10 and 2.1]{CourcelleEngelfriet95} is that
\hr\ context-free sets of graphs are the images of recognizable sets
of ground terms via definable transductions. We extend this result
further, by removing the restriction on having an input set of ground
terms. In particular, having unranked sets of trees is crucial to
obtain the final characterization result in
Theorem~\ref{thm:cf-def-v2}. We also restate some of the known results
in our notation, in which parse trees are trees over a finite alphabet
$\alphabetParseOf{\slabs}$ of operations, that use only source labels
from a given set $\slabs \finsubseteq \sourcelabels$.

We start by restating a known result, namely that
$\proj{\hval}{\trees(\alphabetParseOf{\slabs})}$ is a definable
transduction between parse trees and graphs, for every sort
$\slabs$. Note that the restriction to the alphabet
$\alphabetParseOf{\slabs}$ of edge labels is necessary to ensure the
finiteness of the formulae defining the transduction. Given a tree
$\tree\in\trees(\alphabetParseOf{\slabs})$, a source label $s$ is said
to be \emph{present} in $\tree$ if $\hval(\tree)$ has an
$s$-source. The following lemma shows that the presence of a source
label in a tree is a definable property:

\begin{lem}\label{lemma:present}
  For each sort $\slabs \finsubseteq \sourcelabels$ and each
  $s\in\sourcelabels$, one can build an \mso\ sentence $\phi$ such
  that $\gstruc{\tree} \models \phi$ iff $s$ is present in $\tree$,
  for each tree $\tree\in\trees(\alphabetParseOf{\slabs})$.
\end{lem}
\begin{proof}
  Let $\tree \in \trees(\alphabetParseOf{\slabs})$ be a tree. The
  construction of $\phi$ relies on the following equivalent condition,
  that can be easily expressed by a \mso\ sentence:

  \begin{fact}
    $s$ is present in $\tree$ iff there are $n_0, \ldots, n_m \in
    \vertof{\tree}$ and $s_0, \ldots, s_m \in \slabs$, such
    that: \begin{enumerate}
    \item\label{it1:fact:present} $n_0$ is the root of $\tree$ and $s_0
      = s$,
    \item\label{it2:fact:present} $n_{i+1}$ is a child of $n_i$ in
      $\tree$, for all $i \in \interv{0}{m-1}$,
    \item\label{it3:fact:present} $n_m$ is attached to an edge labeled
      by a unary label, which is either $\emptygraphSymb_{\slabs'}$
      and $s_m \in \slabs'$, for some $\slabs' \subseteq \slabs$, or
      $\sgraphSymb{a}_{(s'_1, \ldots, s'_k)}$ and $s_m \in
      \set{s'_1,\ldots,s'_k}$,
    \item\label{it4:fact:present} the edge between $n_i$ and $n_{i+1}$
      is labeled by a binary label, either $\restrictSymb{\slabs'}$
      and $s_i = s_{i+1} \in \slabs'$, or $\renameSymb{\alpha}$ and
      $\alpha(s_i) = s_{i+1}$.
    \end{enumerate}
  \end{fact}
  \proof{ ``$\Rightarrow$'' By Proposition
    \ref{prop:unique-representations}, we have $\tree =
    t^{\treealgebra(\alphabetParse)}$, for some term $t = \big(\pop_{i\in I}
    \sgraph{c}_i \big) \pop \big( \pop_{j \in J} \extend{b_j}(t_j)
    \big)$. Then $s$ is present in $\tree$ because
    either: \begin{itemize}[label=$\triangleright$]
    \item $\sgraph{c}_i = \emptygraphSymb_\slabs'$, for some $i \in I$,
      such that $s \in \slabs'$,
    \item $\sgraph{c}_i = \sgraphSymb{a}_{(s'_1, \ldots, s'_k)}$, for
      some $i \in I$, such that $s \in \set{s'_1, \ldots, s'_k}$,
    \item $b_j = \restrict{\slabs'}$, for some $j \in J$, such that $s
      \in \slabs'$ and $s$ is present in $t_j^\treealgebra$,
    \item $b_j = \rename{\alpha}$, for some $j \in J$, such that $s =
      \alpha(s')$ and $s'$ is present in $t_j^\treealgebra$,
    \end{itemize}
    In the first two cases, we set $m=0$, $n_0$ the root of $\tree$
    and $s_0 \isdef s$. In the last two cases, we set $n_0$ as the
    root of $\tree$, $s_0\isdef s$ and continue building $n_1, \ldots,
    n_m$ and $s_1, \ldots, s_m$ from $t_j$. It is easy to check that
    the conditions (\ref{it1:fact:present}-\ref{it4:fact:present}) are
    satisfied by the sequences $n_0,\ldots,n_m$ and $s_0, \ldots, s_m$
    built as described above.

    \noindent``$\Leftarrow$'' Let $\tree_0, \ldots, \tree_m$ be the
    subtrees of $\tree$ rooted in $n_0, \ldots, n_m$, respectively.
    By condition (\ref{it2:fact:present}) $\tree_{i+1}$ is a subtree
    of $\tree_i$, for each $i \in \interv{0}{m-1}$. By induction on
    $m-i$, one shows that $s_i$ is present in $\tree_i$, for all $i
    \in \interv{0}{m}$. The base case $i=m$ follows from condition
    (\ref{it3:fact:present}). The inductive case $i < m$ follows from
    condition (\ref{it4:fact:present}). Since $s_0 = s$ and $n_0$ is
    the root of $\tree$, by condition (\ref{it1:fact:present}), $s$ is
    present in $\tree$. \qed}

  Back to the construction of $\phi$, the existence of a path starting
  in the root can be described by an \mso\ formula $\psi(X)$ in the
  relational signature
  $\relsigtree{}(\alphabetParseOf{\slabs})$. Further, the local
  conditions (\ref{it3:fact:present}) and (\ref{it4:fact:present}) can
  be encoded by formul{\ae} $\eta_1(x)$ and $\eta_2(x,y)$,
  respectively. Note that $\slabs$ being a given finite set is crucial
  in encoding the conditions such as $s \in \slabs'$ and $s =
  \alpha(s')$, for $\slabs' \subseteq \slabs$ and $\alpha$ a
  $\slabs$-permutation by finite formul{\ae}. Finally, we define:
  \[\phi \isdef \exists X ~.~ \psi(X) \wedge
  \big(\exists x ~.~ X(x) \wedge \eta_1(x)\big) \wedge
  \big(\forall x \forall y ~.~ X(x) \wedge X(y)
  \wedge \bigvee_{
    \scriptscriptstyle{b \in \alphabetParseOf{\slabs}}
  } \exists z ~.~ \arel_b(z,x,y) \rightarrow \eta_2(x,y)\big)\qedhere\]
\end{proof}

A statement similar to the next lemma is proved in~\cite[Proposition
  7.48]{courcelle_engelfriet_2012}. For reasons of self-containment,
we give a proof using our notation, that makes the set of used source
labels explicit:
\begin{lem}\label{lemma:evaluation-homomorphism-transduction}
  For each $\slabs \finsubseteq \sourcelabels$,
  $\proj{\hval}{\trees(\alphabetParseOf{\slabs})}$ is a definable
  $(\relsigtree{}(\alphabetParseOf{\slabs}),~\relsiggraph{\slabs})$-transduction.
\end{lem}
\begin{proof}
We define the desired transduction in two steps.  In the first step,
we expand each node of the input tree into at most $\cardof{\slabs}+1$
many nodes, one for each source that is present in the respective
subtree plus one extra node that represents an edge. In the second
step, we merge the nodes that are fused by the composition
operations. The first step uses an extra binary relation symbol
$\mathord{\equiv}$ that keeps track of the nodes which are to be
merged in the second step. This relation symbol is interpreted over
different layers by formul{\ae} in the relational signature
$\relsiggraph{\slabs}$.

We now describe the first step.  We use a transduction that creates
$\cardof{\slabs}+1$ copies of the input structure.  We will
use the sources in $\slabs$ and an additional source label
$\square\in\sourcelabels\setminus\slabs$ to index the
copies of the input structure. Formally, we define a parameterless transduction
scheme $\scheme = \tuple{\varphi,\set{\psi_s}_{s \in \slabs
    \cup
    \set{\square}},\set{\theta_{(a,s_1,\ldots,s_{\arityof{a}})}}_{a
    \in \alphabet, s_1,\ldots,s_{\arityof{a}} \in \slabs} \cup
  \set{\theta_{(\equiv,s,t)}}_{s,t \in \slabs}}$, as
follows: \begin{itemize}[label=$\triangleright$]
\item $\varphi$ specifies the domain of the transduction, i.e.,
  $\varphi$ expresses that the input structure is the encoding
  $\gstruc{\tree}$ of some tree
  $\tree\in\trees(\alphabetParseOf{\slabs})$. It is easy to verify
  that such an \mso-formula can be built.
\item $\psi_s$ defines the universe of the $s$-th layer of the result,
  for each $s \in \slabs$. Namely, $\psi_s$ holds for an
  element of the universe of the input structure iff this element is a
  vertex and if $s$ is present at the subtree rooted at this vertex
  (these elements will represent the vertices of the output structure).
  Such a formula can be built according to
  Lemma~\ref{lemma:present}. Moreover, $\psi_\square$ holds for all
  elements of the input structure that are edges labeled by unary symbols
  $\sgraphSymb{a}_{(s_1, \ldots, s_n)}$ (these elements will represent
  the edges of the output structure).
\item each $\theta_{(a,s_1,\ldots,s_{\arityof{a}})}$ has free
  variables $x_0,x_1, \ldots,x_{\arityof{a}}$ and defines the
  interpretation of $(\arel_a,s_1,\ldots,s_{\arityof{a}})$ in the
  result, for all $a \in \alphabet$. We define
  $\theta_{(a,s_1,\ldots,s_{\arityof{a}})}$ to hold for tuples
  $((u_0,\square),(u_1,s_1),\ldots,(u_{\arityof{a}},s_{\arityof{a}}))$
  iff $u_0$ represents a graph edge labeled by the unary symbol
  $\underline{\sgraph{a}}(s_1, \ldots, s_{\arityof{a}})$, such that
  $u_1 = \ldots = u_{\arityof{a}}$ is a tree node incident to
  $u_0$. It is easy to build an \mso-formula
  $\theta_{(a,s_1,\ldots,s_{\arityof{a}})}$ defining these properties.
\item each $\theta_{(\equiv,s,t)}$, for $s,t \in \slabs$, has free
  variables $x_1$ and $x_2$, and defines the interpretation of
  $(\equiv,s,t)$ in the result. For a tuple $((u_1,s),(u_2,t))$, we
  define $\theta_{(s,\equiv,t)}$ to hold iff $u_1$ and $u_2$ are tree
  nodes, such that $u_2$ is the child of $u_1$, for some tree edge
  labeled by one of the following symbols: \begin{itemize}[label=$-$]
  \item $\restrictSymb{\slabs'}$, such that $s \in \slabs'$ and $s = t$, or
  \item $\renameSymb{\alpha}$, such that $\alpha(s) = t$.
  \end{itemize}
  It is easy to build an \mso\ formula $\theta_{(\equiv,s,t)}$
  defining these properties.
\end{itemize}
The second step of the construction is a transduction that takes the
least equivalence relation that subsumes the relation defined by the
\mso{} formula $\theta_\equiv \isdef \bigvee_{s,t \in \slabs}
\theta_{(\equiv,s,t)}$ and constructs its quotient structure.  It is
well known that the quotient structure with regard to an equivalence
relation definable in \mso{} can be expressed as a definable
transduction, e.g., see~\cite[Lemma 2.4]{CourcelleV}. It is now
routine to verify that the composition of the two transductions above
has the desired properties. Moreover, by Proposition
\ref{prop:comp-restr-trans} (\ref{it2:comp-restr-trans}), the
composition of definable transductions is definable.
\end{proof}

Following a standard
convention~\cite{comon:hal-03367725,CourcelleEngelfriet95}, we
identify a ground term $t$ over a finite functional signature
$\fsignature$ with a tree, such that each node $n \in \vertof{t}$ is
attached to exactly one unary edge labeled with a function symbol $f
\in \fsignature$ and is, moreover, the source of $\arityof{f}$ binary
edges labeled $1, \ldots, \arityof{f}$. The intuition is that the
destination of the edge labeled $i$ is the $i$-th child of $n$. More
precisely, a ground term $t \in \inituniv{\fsignature}$ is the tree $t
\in \trees(\alphabetTwo_\fsignature)$, where $\alphabetTwo_\fsignature
\isdef \fsignature \cup \interv{1}{\max_{f \in \fsignature}
  \arityof{f}}$. Since each ground term is a graph, we write
$\gstruc{t}$ for the relational structure encoding the term $t$
(\autoref{sec:cmso-def}).

With these conventions in mind, we recall a characterization of
context-free sets of graph as images of recognizable set of terms via
definable transductions, found by Courcelle and
Engelfriet~\cite[Theorems 1.10 and 2.1]{CourcelleEngelfriet95}. The
statement below is given according to our definitions:

\begin{thmC}[\cite{CourcelleEngelfriet95}]\label{thm:cf-ranked}
  A set of graphs $\mathcal{L} \subseteq \graphsof{\emptyset}$ is context-free iff there
  exists: \begin{enumerate}
  \item a finite functional signature $\fsignature$,
  \item a set $\mathcal{K} \subseteq \inituniv{\fsignature}$ of ground
    terms recognizable in $\initalg{\fsignature}$, and
    %
  \item a definable $(\fsignature,~\relsiggraph{\emptyset})$-transduction $F$,
    %
  \end{enumerate}
  such that $\gstruc{\mathcal{L}}=F(\gstruc{\mathcal{K}})$.
\end{thmC}

We extend below the result of Theorem \ref{thm:cf-ranked} from terms
to unranked trees. The main difficulty here is the difference between
the initial algebra $\initalg{\fsignature}$ over a finite functional
signature $\fsignature$ and the tree algebra
$\treealgebra(\alphabetTwo)$ over a finite edge label alphabet
$\alphabetTwo$. Note that a tree is built in $\initalg{\fsignature}$
by taking all the subtrees starting at the children of the root at
once, whereas in $\treealgebra(\alphabetTwo)$, building the same tree
requires a series of $\extend{b}$ and $\pop$ operations.

\begin{cor}\label{cor:cf-unranked}
  A set of graphs $\mathcal{L} \subseteq \graphsof{\emptyset}$ is context-free iff there
  exists:~\begin{enumerate}
  \item a finite alphabet $\alphabetTwo$ of edge labels,
  \item a (possibly unranked) set $\mathcal{K} \subseteq
    \trees(\alphabetTwo)$ of trees recognizable in
    $\treealgebra(\alphabetTwo)$, and
  \item a definable
    $(\relsigtree{}(\alphabetTwo),~\relsiggraph{\emptyset})$-transduction
    $F$,
  \end{enumerate}
  such that $\gstruc{\mathcal{L}}=F(\gstruc{\mathcal{K}})$.
\end{cor}
\begin{proof}
  ``$\Rightarrow$'' By the left to right direction of
  Theorem~\ref{thm:cf-ranked}, there exists a finite functional
  signature $\fsignature$, a recognizable set $\mathcal{K}$ of ground
  terms over $\fsignature$ and a definable
  $(\fsignature,~\relsiggraph{\emptyset})$-transduction $F$, such that
  $\gstruc{\mathcal{L}}=F(\gstruc{\mathcal{K}})$. Note that the
  alphabet $\alphabetTwo_\fsignature$ of unary and binary edge labels
  is finite, because $\fsignature$ is finite. We take $\alphabetTwo$
  to be $\alphabetTwo_\fsignature$, in the following. Since
  $\mathcal{K}$ is recognizable in $\initalg{\fsignature}$, the set
  $\gstruc{\mathcal{K}}$ is \mso-definable, by the classical
  equivalence between recognizability in $\initalg{\fsignature}$ and
  \mso-definability of sets of ground terms over a finite functional
  signature $\fsignature$~\cite{Doner70,MezeiWright67}. Since the set
  of trees $\mathcal{K} \in \trees(\alphabetTwo)$ is
  \mso-definable, it is recognizable in $\treealgebra(\alphabetTwo)$,
  by Theorem \ref{thm:unranked-trees-rec-def}.

  ``$\Leftarrow$'' Let $\mathcal{K}$ be a set of trees recognizable in
  $\treealgebra(\alphabetTwo)$ and let $\mathcal{T} \isdef \set{t \in
    \inituniv{\treesignature(\alphabetTwo)} \!\mid\!
    t^{\treealgebra(\alphabetTwo)}\! \in \mathcal{K}}$ be the set of
  ground $\treesignature(\alphabetTwo)$-terms that evaluate to a tree
  in $\mathcal{K}$. By Theorem \ref{thm:unranked-trees-rec-def}, the
  set $\gstruc{\mathcal{K}}$ is definable.  We denote by
  $\overline{\hval}$ the function that maps $\gstruc{t}$ into
  $\gstruc{t^{\treealgebra(\alphabetTwo)}}$, for each $t \in
  \inituniv{\treesignature(\alphabetTwo)}$. Assume that
  $\overline{\hval}$ is a definable transduction. Then,
  $\gstruc{\mathcal{T}} = \overline{\hval}^{-1}(\gstruc{\mathcal{K}})$
  is definable, by Theorem \ref{thm:back-trans}. By the classical
  result of \cite{Doner70,ThatcherWright68}, the set $\mathcal{T}$ of
  ground $\treesignature(\alphabetTwo)$-terms is recognizable in
  $\inituniv{\treesignature(\alphabetTwo)}$. Since $\overline{\hval}$
  was assumed to be definable, the transduction $\overline{\hval}
  \circ F$ is definable, by Proposition \ref{prop:comp-restr-trans}
  (\ref{it1:comp-restr-trans}), thus $\mathcal{L}$ is context-free, by
  Theorem \ref{thm:cf-ranked}. It remains to show that
  $\overline{\hval}$ is a definable transduction. Let $\aroot$ be the
  (only) source label necessary to build any parse tree (over the
  alphabet $\alphabetParseOf{\set{\aroot}}$). By Lemma
  \ref{lemma:evaluation-homomorphism-transduction},
  $\proj{\hval}{\trees(\alphabetParseOf{\set{\aroot}})}$ is a definable
  transduction. Then, $\overline{\hval} =
  \proj{\hval}{\trees(\alphabetParseOf{\set{\aroot}})} \circ\ \trans$, where
  $\trans$ is the transduction that turns ground terms into parse
  trees. It is easy to see that $\trans$ is definable, hence
  $\overline{\hval}$ is definable, by Proposition
  \ref{prop:comp-restr-trans} (\ref{it1:comp-restr-trans}).
\end{proof}

\section{Parsable Sets of Graphs}
\label{sec:parsable}

This section gives the main result of the paper, i.e., a
characterization of the class of definable context-free sets of
graphs. Essentially, we show that these are exactly the
\emph{parsable} sets, for which the parse trees of the grammar can be
extracted from each graph, by means of a definable transduction.

Parsable sets are closely related to the notion of \emph{strongly
context-free} sets of graphs introduced by Courcelle~\cite[Definition
  4.2]{CourcelleV}. According to the original definition, a strongly
context-free set $\mathcal{L}$ is the image of a set $\mathcal{K}$ of
parse trees over a finite signature of hyperedge-replacement
operations (i.e., graphs $\graph$ having special nonterminal edges
$e\in\edgeof{\graph}$, that can be substituted by any graph $H$ of
sort $\slabs$ with $\cardof{\slabs} = \arityof{\labof{\graph}(e)}$ via
the canonical evaluation function $\hval$), such that, moreover, there
exists a definable transduction $\parsefunc \subseteq
(\proj{\hval}{\mathcal{K}})^{-1}$ such that $\dom{\parsefunc} =
\mathcal{L}$. Here, by \emph{parsable} set, we denote the existence of
the inverse transduction between graphs and parse trees over
$\hrsignature$:

\begin{defi}\label{def:parsable}
  Let $\slabs\finsubseteq\sourcelabels$ be a sort. A set of graphs
  $\mathcal{L} \subseteq \graphs$ is $\slabs$-\emph{parsable} iff
  there exists a definable $(\relsiggraph{\slabs},
  \relsigtree{}(\alphabetParseOf{\slabs}))$-transduction $\parsefunc$
  such that:~\begin{enumerate}
  \item\label{it1:parsable} $\gstruc{\mathcal{L}} = \dom{\parsefunc}$,
    and
  \item\label{it2:parsable} if $(\gstruc{\graph},\gstruc{\tree}) \in
    \parsefunc$ then $\hval(\tree)=\graph$.
  \end{enumerate}
  We call a set of graphs $\mathcal{L} \subseteq \graphs$
  \emph{parsable}, if $\mathcal{L}$ is $\slabs$-\emph{parsable} for
  some $\slabs\finsubseteq\sourcelabels$.
\end{defi}
Choosing $\mathcal{K}$ such that
$\gstruc{\mathcal{K}}=\parsefunc(\gstruc{\mathcal{L}})$ for a parsable
set of graphs $\mathcal{L}$, we obtain that
$\mathcal{L}=\hval(\mathcal{K})$, i.e., $\mathcal{L}$ is strongly
context-free in the sense of~\cite[Definition 4.2]{CourcelleV} (noting
that the $\hrsignature$ operations can be expressed as
hyperedge-replacement operations).  On the other hand, the original
notion of strongly context-free set may appear, at first sight, to be
more general, because it refers to any signature of
hyperedge-replacement operations, whereas our notion of parsable set
fixes the signature to $\hrsignature$. However, every
hyperedge-replacement operation can be expressed by a $\hrsignature$
term and it is straightforward to define transductions that replace
parse trees of hyperedge-replacement operations by $\hrsignature$
parse trees that encode these operations.

A first result of this section
(\autoref{sec:parsing-tree-decompositions}) is that the set of graphs
whose tree-widths are bounded by a constant is parsable (Theorem
\ref{thm:bounded-tree-width-is-strongly-context-free}), thus proving
Conjecture \ref{conj:courcelle2}. As stated in the introduction, this
also establishes Conjecture \ref{conj:courcelle3}, which states that
the definable and context-free sets of graphs are exactly the parsable
ones. The first step in order to prove Theorem
\ref{thm:bounded-tree-width-is-strongly-context-free} is the ability
of extracting an optimal-width tree decomposition from a graph by
means of a definable transduction (for this we use~\cite[Theorem
  2.4]{10.1145/2933575.2934508} and~\cite[Theorem
  2.1]{journals/lmcs/BojanczykP22}). For the second step, we then
prove that each tree decomposition can be further translated into a
parse tree by a definable transduction.

Finally, we give two theorems that characterize the definable
context-free sets of graphs using four equivalent conditions
(\autoref{sec:characterization}). The first theorem (Theorem
\ref{thm:cf-def}) is explicit about the finite set of sources used in
the grammar. The second theorem (Theorem \ref{thm:cf-def-v2})
quantifies this set existentially in each condition and provides, in
addition to Theorem \ref{thm:cf-def}, a purely logical condition, in
terms of a pair of definable transductions that act as encoding
(graphs to trees) and decoding (trees to graphs).

We start by proving that the recognizable subsets of a parsable set
are also parsable:

\begin{lem}\label{lemma:parsable-closure}
  Let $\slabs\subseteq\sourcelabels$ be a sort, $\mathcal{L}$ be a
  $\slabs$-parsable set of graphs and $\mathcal{L}' \subseteq
  \mathcal{L}$ be a set recognizable in $\algebra^\slabs$.  Then,
  $\mathcal{L}'$ is $\slabs$-parsable.
\end{lem}
\begin{proof}
  Let $\parsefunc$ be the definable
  $(\relsiggraph{\slabs},~\relsigtree{}(\alphabetParseOf{\slabs}))$-transduction
  that witnesses the parsability of $\mathcal{L}$ as in Definition
  \ref{def:parsable}.  Then $\mathcal{K} \isdef
  \hval^{-1}(\mathcal{L}')$ is a set of parse trees for
  $\mathcal{L}'$. By Lemma~\ref{lemma:finite-sort-inverse-homomorphism}, $\mathcal{K}$ is
  recognizable in $\repalgebra$. By Lemma~\ref{lemma:red-rec},
  $\mathcal{K}$ is also recognizable in
  $\treealgebra(\alphabetParse)$, hence $\mathcal{K}$ is definable, by
  Theorem \ref{thm:unranked-trees-rec-def}.  By Theorem
  \ref{thm:back-trans}, we obtain that
  $\mathcal{L}'=\parsefunc^{-1}(\mathcal{K})$ is definable, hence the
  domain-restriction of $\parsefunc$ to $\mathcal{L}'$ is definable,
  by Proposition \ref{prop:comp-restr-trans}
  (\ref{it1:comp-restr-trans}).
\end{proof}

The ``only if'' direction of Conjecture \ref{conj:courcelle3} is
proved next. The ``if'' direction will be proved as part of Theorem
\ref{thm:cf-def}.

\begin{prop}\label{prop:parsable}
  Any $\slabs$-parsable set of graphs is both definable and
  context-free, defined by a grammar over the signature
  $\hrsignature^\slabs$.
\end{prop}
\begin{proof}
  Let $\mathcal{L}$ be a $\slabs$-parsable set of graphs, for a sort
  $\slabs\finsubseteq\sourcelabels$. Then, $\mathcal{L}$ is definable
  because $\parsefunc$ is definable, thus $\gstruc{\mathcal{L}} =
  \dom{\parsefunc}$ is definable, where $\parsefunc$ is the
  $(\relsiggraph{\slabs},~\relsigtree{}(\alphabetParseOf{\slabs}))$-transduction
  from \autoref{def:parsable}. By Theorem \ref{thm:graphs-def-rec},
  $\mathcal{L}$ is recognizable in $\algebra$. Since $\mathcal{L}$ is
  $\slabs$-parsable, we obtain that $\mathcal{L}$ is also recognizable
  in the subalgebra $\btwalgebraof{\slabs}$. This is because every
  graph in $\mathcal{L}$ can be built using only operations that are
  the interpretations of the function symbols from
  $\hrsignature^\slabs$, which occur on the edge labels of some parse
  tree from $\trees(\alphabetParseOf{\slabs})$. By
  Proposition~\ref{prop:btw-graphs-are-context-free}, we have that set
  of graphs $\btwgraphsof{\slabs}$ is context-free.  Hence,
  $\mathcal{L} = \mathcal{L} \cap \btwgraphsof{\slabs}$ is
  context-free by Theorem~\ref{thm:filtering-theorem}.
\end{proof}

\subsection{Parsing with Tree Decompositions}
\label{sec:parsing-tree-decompositions}

The definition of parsable sets of graphs requires a definable
transduction from graphs to trees that produces, for each input graph,
a parse tree of that graph relative to some grammar, that does not
depend on the input graph. A candidate for such a parse tree is any
tree decomposition that witnesses the tree-width of a graph. We
recover such an optimal tree decomposition from a seminal result of
Boja\'{n}czyk and Pilipczuk that states the existence of a definable
transduction which computes some optimal tree decomposition of a given
graph. The following theorem combines the results of \cite[Theorem
  2.4]{10.1145/2933575.2934508} and \cite[Theorem
  2.1]{journals/lmcs/BojanczykP22}:

\begin{thmC}[\cite{10.1145/2933575.2934508,journals/lmcs/BojanczykP22}]\label{thm:td-mso}
  For every $k \in \nat$, there exists a definable
  $(\relsiggraph{\emptyset},\relsigdecomp{\emptyset})$-transduction
  $\mathcal{I}$, such that the following holds: \begin{enumerate}
  \item\label{it1:td-mso} $\astruc \in \dom{\mathcal{I}}$ iff $\astruc =
    \gstruc{\graph}$ for some graph $\graph$, such that $\twd{\graph} \le
    k$,
  \item\label{it2:td-mso} if $(\gstruc{\graph},\astruc) \in
    \mathcal{I}$ for some graph $\graph$, then $\astruc =
    \gstruc{\graph,\tree,\beta}$ for some tree decomposition
    $(\tree,\beta)$ of $\graph$ of width at most $k$.
  \end{enumerate}
\end{thmC}

We show next that each encoding of some tree decomposition of a graph
$\graph$ can be mapped to a parse tree that evaluates to $\graph$ via the canonical evaluation:

\begin{lem}\label{lem:mapping-tree-decompositions-back-to-dervation-trees}
For all sorts $\slabs \finsubseteq \sourcelabels$ and $\slabs'
\subseteq \slabs$, there is a definable
$(\relsigdecomp{\slabs'},~\relsigtree{}(\alphabetParseOf{\slabs}))$-transduction
$\mathcal{J}$, such that:\begin{enumerate}
\item\label{it1:lem:mapping-tree-decompositions-back-to-dervation-trees}
  $\astruc \in \dom{\mathcal{J}}$ iff $\astruc = \gstruc{\graph,D,\beta}$ for
  some graph $\graph$ with $\twd{\graph} \le \cardof{\slabs}-1$,
  witnessed by a tree decomposition $(D,\beta)$, and
\item\label{it2:lem:mapping-tree-decompositions-back-to-dervation-trees}
  if $(\gstruc{\graph,D,\beta},\gstruc{\tree}) \in \mathcal{J}$ then
  $\tree\in\trees(\alphabetParseOf{\slabs})$ and $\hval(\tree) =
  \graph$.
\end{enumerate}
\end{lem}
\begin{proof}
Let $k = \cardof{\slabs} - 1$.
The idea of the transduction $\mathcal{J}$
is to use the tree $D$, encoded by the interpretation of the $\node$
and $\parent$ relation symbols from $\relsigdecomp{\slabs}$, as the
skeleton for the output tree $\tree$. In order to label the edges of
$\tree$ with unary and binary edge labels, we guess a coloring of the
vertices in the input graph, using the parameters $\set{X_s}_{s \in
  \slabs}$, such that every vertex is labeled by exactly one color
$X_s$. Given a node $n \in \vertof{D}$, let $\colorof{n} \isdef \{s
\mid \text{ there is a vertex } v \text{ colored by } X_s \text{ and }
\bag(v,n) \text{ holds}\}$ be the colors of the vertices in the bag
$\beta(n)$.  Moreover, for every edge $e \in \edgeof{\graph}$, we let
$\nodefunc(e)$ be the closest node $n$ to the root with
$\edgerelof{\graph}(e)_i \in \beta(n)$, for all $1 \le i \le
\arityof{\labof{\graph}(e)}$.  Note that $\nodefunc(e)$ exists by
Definition \ref{def:tree-decomposition} (\ref{it1:tree-decomposition})
and that $\nodefunc(e)$ is unique, by Definition
\ref{def:tree-decomposition} (\ref{it2:tree-decomposition}) (i.e., if
there exist two distinct nodes containing all vertices attached to
$e$, then all these vertices must also belong to the bag of their
unique common ancestor).

We are going to use a transduction that creates three layers (i.e.,
copies of the input structure) indexed by the names $\vertexlabel$,
$\vertexsetlabel$ and $\edgelabel$, respectively. Then, $\mathcal{J}$ is the
transduction defined by the scheme:
\[\begin{array}{lr}
  \scheme \isdef & \large\langle
  \varphi, \psi_\vertexlabel, \psi_\vertexsetlabel, \psi_\edgelabel,
  \set{\theta_{\restrictSymb{\slabs'}}}_{\slabs' \subseteq \slabs}, \set{\theta_{\renameSymb{\alpha}}}_{\alpha \text{ is $\slabs$-permutation}},
  \\[1mm]  &  
  \set{\theta_{\sgraphSymb{a}(s_1, \ldots, s_{\arityof{a}})}}_{a \in \alphabet, s_1,\ldots,s_{\arityof{a}} \in \slabs},
  \set{\theta_{\emptygraphSymb_{\slabs'}}}_{\slabs'\subseteq \slabs} \large\rangle
\end{array}\]
where:
\begin{itemize}[label=$\triangleright$]
\item $\varphi(\set{X_s}_{s \in \slabs})$ defines the domain of the
  transduction, by checking that the following hold:
  \begin{enumerate}[1.]
  \item the sets $\set{X_s}_{s \in
  \slabs}$ forms a partition of the vertices of $\vertof{\graph}$, i.e., every vertex is labeled of $\vertof{\graph}$ is labelled by exactly one color $X_s$,
  \item the bags of the tree decomposition are all of size at most $k+1$,
  \item for each edge $e \in \edgeof{\graph}$, there is a node $n \in
  \vertof{D}$ whose bag $\beta(n)$ contains all vertices from
  $\vertof{\graph}$ attached to $e$,
  \item the set of nodes $\set{n \in \vertof{D} \mid v \in \beta(n)}$ is
  non-empty and connected in $D$,
  \item the sets $\set{X_s}_{s \in \slabs}$ form a partition of
  $\vertof{\graph}$ that is consistent with the tree decomposition,
  i.e., that in each bag there is at most one vertex labelled by
  $X_s$, for all $s \in \slabs$, and
  \item for each $s$-source $v\in\vertof{\graph}$, the color of $v$
  is indeed $X_s$ and $s$ belongs to the bag associated with the root
  of $D$.
  \end{enumerate}
\item $\psi_\vertexlabel(x_1) \isdef \node(x_1)$ represents the nodes
  of the output tree $\tree$.
\item $\psi_\vertexsetlabel(x_1) \isdef \node(x_1)$ represents the
  unary edges with labels $\emptygraphSymb_{\slabs'}$ of $\tree$.
\item $\psi_\edgelabel(x_1)$ holds for those elements where
  $\node(x_1)$ holds, except for the root of the tree; these elements
  represent the binary $\restrictSymb{\slabs'}$-labeled edges of $\tree$.
  Moreover, $\psi_\edgelabel(x_1)$ holds also for the elements that
  encode the edges of $\graph$; these elements represent the unary
  $\sgraphSymb{a}_{(s_1, \ldots, s_{\arityof{a}})}$-labeled edges of $\tree$.
\item $\theta_{\restrictSymb{\slabs'}}(x_1,x_2,x_3,\set{X_s}_{s \in
  \slabs})$ defines the interpretation of the ternary relation symbol
  $\arel_{\restrictSymb{\slabs'}}$ in $\gstruc{\tree}$, i.e., all
  triples $\tuple{(n_1,\edgelabel), (n_2,\vertexlabel),
    (n_3,\vertexlabel)} \in \vertof{\tree}^3$, such that $n_3$ is the
  parent of $n_2$, $n_1 = n_2$ and $ \slabs'= \colorof{n_2} \cap
  \colorof{n_3}$.
\item $\theta_{\renameSymb{\alpha}}(x_1,x_2,x_3,\set{X_s}_{s \in
  \slabs})$ is set to false for all $\slabs$-permutations $\alpha$ as
  the rename operation is not needed for the construction of graphs
  from tree decompositions,
\item $\theta_{\sgraphSymb{a}_{(s_1, \ldots,
    s_{\arityof{a}})}}(x_1,x_2,\set{X_s}_{s\in\slabs})$ defines the
  interpretation of the binary relation\linebreak symbol
  $\arel_{\sgraphSymb{a}_{(s_1, \ldots, s_{\arityof{a}})}}$ in
  $\gstruc{\tree}$, i.e., all pairs $\tuple{(n_1,\edgelabel),
    (n_2,\vertexlabel)} \in \vertof{\tree}^2$, such that $n_1 \in
  \edgeof{\graph}$ is an edge with label $\labof{\graph}(n_1)=a$ and
  incident vertices $\edgerelof{\graph}(n_1) =
  \tuple{v_1,\ldots,v_{\arityof{a}}}$ colored by $X_{s_1}, \ldots,
  X_{s_{\arityof{a}}}$, respectively, and $n_2 = \nodefunc(n_1)$.
\item
  $\theta_{\emptygraphSymb_{\slabs'}}(x_1,x_2,\set{X_s}_{s\in\slabs})$
  defines the interpretation of the binary relation symbol
  $\arel_{\emptygraphSymb_{\slabs'}}$ in $\gstruc{\tree}$, i.e., all
  pairs $\tuple{(n_1,\vertexsetlabel), (n_2,\vertexlabel)} \in
  \vertof{\tree}^2$, such that $n_1 = n_2$ and $\slabs' =
  \colorof{n_1}$.
\end{itemize}
Note that the formula $\varphi$ ensures that $\astruc \in \dom{\mathcal{J}}$ iff
$\astruc = \gstruc{(\graph,D,\beta)}$ for some graph $\graph$
with $\twd{\graph} \le k$, witnessed by a tree decomposition
$(D,\beta)$.

Let $(\gstruc{(\graph,D,\beta)},\gstruc{\tree}) \in \mathcal{J}$ for
some $\graph$ with $\twd{\graph} \le k$, witnessed by a tree
decomposition $(D,\beta)$. It is easy to verify that
$\tree\in\trees(\alphabetParseOf{\slabs})$.  We note that every node
of $n$ of the tree $\tree$ is also a node of $D$ and vice versa, by
the definition of the transduction scheme $\scheme$. For each $n \in
\vertof{\tree}$, let $\tree_n$ denote the subtree of $\tree$ rooted in
$n$. Let $\set{X_s}_{s \in \slabs}$ be the coloring guessed by the
transduction. For a node $n \in \vertof{D}$, we denote by
$\graphfunc(n)$ the subgraph of $\graph$ consisting of all vertices
that appear in bags of descendants of $n$ and edges $e$, such that
$\nodefunc(e)$ is a descendant of $n$. We mark a vertex $v$ of
$\graphfunc(n)$ as an $s$-source iff $v$ appears in the bag associated
with $n$ and is colored by $X_s$. We show the following:

\begin{fact}
  $\hval(\tree_n) = \graphfunc(n)$, for all $n \in \vertof{\tree}$.
\end{fact}
\proof{ By induction on the structure of $\tree$, let
  $n\in\vertof{\tree}$ be a node with children $n_1, \ldots, n_l$. By
  the inductive hypothesis, we have that
  $\hval(\tree_{n_i})=\graphfunc(n_i)$, for all $i \in \interv{1}{l}$.
  We denote by $\slabs_i \isdef \set{s \in \slabs \mid \exists v \in
    \vertof{\graph} ~.~ v \text{ is colored by } X_s,~ \bag(v,n)
    \text{ and } \bag(v,n_i) \text{ hold}}$ the set of labels that
  occur simultaneously in the bags of $n$ and $n_i$, for all $i \in
  \interv{1}{l}$. We consider $\graph_0$ to be the subgraph of
  $\graph$ consisting of the vertices of $\graph$ from the bag of $n$
  and the edges $e\in\edgeof{\graph}$ such that $n=\nodefunc(e)$.  Let
  $\slabs_0 \isdef \set{s \in \slabs \mid \exists v \in
    \vertof{\graph_0} ~.~ v \text{ is colored by } X_s}$.  Then, we
  have: \[\graph_0 = \left( \emptygraph_{\slabs_0} \pop ( \pop_i
  \sgraph{a_i}(s_1^i, \ldots, s_{n_i}^i)) \right)^\algebra =
  \emptygraph^\algebra_{\slabs_0} \pop^\algebra ( \pop^\algebra_i
  \sgraph{a_i}(s_1^i, \ldots, s_{n_i}^i)^\algebra ) \] for suitably chosen
  edge labels $a_i\in\alphabet$ and source labels $s_j^i \in
  \slabs$. We observe that: \[\graphfunc(n) = \graph_0 \pop^\algebra (
  \pop^\algebra_{i=1..l} \restrict{\slabs_i}^\algebra(\graphfunc(n_i))
  ) ~~~~\text{(*)}\] By the definition of the canonical evaluation, we
  have: \[\hval(\tree_n) = \emptygraph_{\slabs_0}^\algebra
  \pop^\algebra (\pop^\algebra_i \sgraph{a_i}(s_1^i, \ldots,
  s_n^i)^\algebra) \pop^\algebra (\pop^\algebra_{i=1..l}
  \restrict{\slabs_i}^\algebra(\hval(\tree_{n_i})))\] The claim
  follows by the above equation, (*) and the inductive
  hypothesis. \qed}

The proof is concluded by choosing $n$ as the root of $\tree$ in the
above fact, which leads to $\hval(\tree) = \graph$, as required.
\end{proof}

The next corollary of Theorem \ref{thm:td-mso} provides a powerful
result.
By instantiating its statement for $\slabs=\emptyset$, we
obtain that every graph $\graph \in \graphsof{\emptyset}$ of
tree-width $k$ is the image of a parse tree under an \mso-transduction
and, moreover, this parse tree uses exactly $k+1$ source labels:

\begin{cor}
\label{cor:mapping-bounded-tree-width-graphs-back-to-dervation-trees}
  For all sorts $\slabs\finsubseteq\sourcelabels$ and $\slabs' \subseteq \slabs$, there exists a
  definable\linebreak
  $(\relsiggraph{\slabs'},~\relsigtree{}(\alphabetParseOf{\slabs}))$-transduction
  $\mathcal{K}$, such that:
  \begin{enumerate}
  \item\label{it1:mapping-bounded-tree-width-graphs-back-to-dervation-trees}
    $\astruc \in \dom{\mathcal{K}}$ iff $\astruc = \gstruc{\graph}$,
    for some graph $\graph$ that admits a tree decomposition
    $(D,\beta)$ of width at most $\cardof{\slabs}-1$, and
  \item\label{it2:mapping-bounded-tree-width-graphs-back-to-dervation-trees}
    if $(\gstruc{\graph},\gstruc{\tree}) \in \mathcal{K}$ then
    $\tree\in\trees(\alphabetParseOf{\slabs})$ and $\hval(\tree) =
    \graph$.
  \end{enumerate}
\end{cor}
\begin{proof}
  The transduction $\mathcal{K}$ is the result of composing the
  transductions $\mathcal{I}$ (Theorem~\ref{thm:td-mso}) and
  $\mathcal{J}$
  (Lemma~\ref{lem:mapping-tree-decompositions-back-to-dervation-trees}).
  However, as $\mathcal{I}$ only inputs graphs with no sources, we need to add some pre-processing in order to encode and then remove the sources.
  We will define two further definable transductions $A$ and $B$ such
  that the desired transduction $\mathcal{K}$ is the result of
  composing $A$, $\mathcal{I}$, $B$ and $\mathcal{J}$, in this order,
  starting with $A$.
  The pre-processing requires the use of some fresh
  (temporary) edge label $a$ of arity $\cardof{\slabs'}$.
  We now define $A$ as the transduction that outputs the encoding of an input graph $\graph \in \btwgraphsof{\slabs'}$, removes the encoding of the sources $\arel_s$, for $s \in \slabs$, and adds an $a$-labelled edge between all $s$-sources, with $s \in \slabs'$.
  This has the effect that for every tree-decomposition of
  $\graph$ (in particular for the output of the composition of
  transductions $A$ and $\mathcal{I}$) there is a node such that the
  bag associated with this node contains all $s$-sources, for $s \in
  \slabs'$, of $\graph$.
  Next, we would like to apply transduction $\mathcal{J}$.
  However, in order to do so, we need to ensure that the $s$-sources in fact appear in the root of the tree decomposition and that we reintroduce the sources $s \in \slabs'$ at the appropriate places.
  We do so by defining the transduction $B$ that inputs
  the encoding a graph $\graph$ and a tree decomposition $(D,\beta)$
  and outputs an encoding of $\graph$ and a tree decomposition
  $(D',\beta')$, which is obtained from $(D,\beta)$ by rotating the
  tree-decomposition, such that the node of $D$ that contains the
  $s$-sources becomes the root of the tree decomposition. The rotation
  operation is implemented by guessing the root of the rotated tree
  and reversing the order of the pairs from the interpretation of the
  of the $\parent$ relation along the unique path between the root of
  the input tree and the new root.
  Further, $B$ deletes the $a$-labelled edge that has been added by $A$ and reintroduces the predicates $\arel_s$, for $s \in \slabs$.
  It is now easy to verify that the composition of $A$, $\mathcal{I}$, $B$ and $\mathcal{J}$ has the desired properties.
\end{proof}

We are now ready to prove Conjecture \ref{conj:courcelle2}:

\begin{thm}
  \label{thm:bounded-tree-width-is-strongly-context-free}
  For each sort $\slabs \finsubseteq \sourcelabels$, the set
  $\btwgraphsof{\slabs}$ is $\slabs$-parsable.
\end{thm}
\begin{proof}
  By \autoref{lemma:cf-twb}, we have $\twd{\graph} \le
  \cardof{\slabs}-1$ for every graph $\graph \in \btwgraphsof{\slabs}$, witnessed by a tree decomposition $(D,\beta)$ such that every $s$-source of $\graph$, with $s \in \slabs$, appears in the bag associated with the root of $D$.
  Hence, $\gstruc{\btwgraphsof{\slabs}} \subseteq
  \dom{\mathcal{K}}$, where $\mathcal{K}$ is the definable
  transduction given by
  Corollary~\ref{cor:mapping-bounded-tree-width-graphs-back-to-dervation-trees}.
  Moreover, we have that $(\gstruc{\graph},\gstruc{\tree}) \in K$
  implies $\tree\in\trees(\alphabetParseOf{\slabs})$ and $\hval(\tree)
  = \graph$, hence $\gstruc{\btwgraphsof{\slabs}} \supseteq
  \dom{\mathcal{K}}$.
  To see that $\btwgraphsof{\slabs}$ is
  $\slabs$-parsable, we can take $\parsefunc \isdef \mathcal{K}$ as required by Definition~\ref{def:parsable}.
\end{proof}

\subsection{Two Characterizations of Definable Context-Free Sets of Graphs}
\label{sec:characterization}

We combine the previously obtained results in a characterization of
the intersection between the classes of context-free and definable
graph languages. We state these characterization results in two
versions.  In the first version, we explicitly keep track of the set
of sources $\slabs \finsubseteq \sourcelabels$ witnessing that a set
of graphs is context-free:

\begin{thm}\label{thm:cf-def}
  For every set $\mathcal{L} \subseteq \graphs$ of graphs and sort
  $\slabs \finsubseteq \sourcelabels$, the following are equivalent:
  \begin{enumerate}
  \item\label{it1:thm:cf-def} $\mathcal{L}$ is definable and
    context-free, for a grammar over $\hrsignature^\slabs$-operations,
  \item\label{it2:thm:cf-def} $\mathcal{L}$ is recognizable in
    $\algebra$ and $\mathcal{L} \subseteq \btwgraphsof{\slabs}$,
  \item\label{it2a:thm:cf-def} $\mathcal{L}$ is recognizable in
    $\algebra^{\slabs}$ and $\mathcal{L} \subseteq
    \btwgraphsof{\slabs}$,
  \item\label{it3:thm:cf-def} $\mathcal{L}$ is $\slabs$-parsable,
\end{enumerate}
\end{thm}
\begin{proof}
  (\ref{it1:thm:cf-def}) $\Rightarrow$ (\ref{it2:thm:cf-def}) By
  Theorem~\ref{thm:graphs-def-rec} every definable set of graphs
  is recognizable in the algebra $\algebra$. Let $\grammar$ be a
  grammar such that $\mathcal{L}=\langof{X}{\grammar}$, for some
  nonterminal $X$ of $\grammar$.  By assumption, $\grammar$ uses only
  $\hrsignature^\slabs$-operations, hence $\mathcal{L} \subseteq
  \graphsof{\slabs}$.  Moreover, since $\langof{X}{\grammar}$ is the
  least solution of $\grammar$, every graph $\graph \in
  \langof{X}{\grammar}$ is term-generated, thus $\mathcal{L} \subseteq
  \btwgraphsof{\slabs}$.

  \vspace*{\baselineskip}
  \noindent(\ref{it2:thm:cf-def}) $\Rightarrow$ (\ref{it2a:thm:cf-def}) By
  Lemma~\ref{lemma:sub-rec}, since $\algebra^{\slabs}$ is a subalgebra
  of $\algebra$.

  \vspace*{\baselineskip}
  \noindent(\ref{it2a:thm:cf-def}) $\Rightarrow$ (\ref{it3:thm:cf-def}) By
  Theorem~\ref{thm:bounded-tree-width-is-strongly-context-free}, the
  set of graphs $\btwgraphsof{\slabs}$ is $\slabs$-parsable, and by
  Lemma~\ref{lemma:parsable-closure}, the restriction of a
  $\slabs$-parsable set to a recognizable set in $\algebra^{\slabs}$
  is $\slabs$-parsable.

  \vspace*{\baselineskip}
  \noindent(\ref{it3:thm:cf-def}) $\Rightarrow$ (\ref{it1:thm:cf-def}) By
  Proposition~\ref{prop:parsable}.
\end{proof}

The second theorem is more coarse, in that we quantify out
existentially the sort and the bound on the tree-width, in each item:

\begin{thm}\label{thm:cf-def-v2}
  For every set $\mathcal{L} \subseteq \graphsof{\emptyset}$ of graphs
  with no sources, the following are
  equivalent: \begin{enumerate}
  \item\label{it1:thm:cf-def-v2} $\mathcal{L}$ is definable and
    context-free,
  \item\label{it2:thm:cf-def-v2} $\mathcal{L}$ is recognizable and has
    bounded tree-width,
  \item\label{it3:thm:cf-def-v2} $\mathcal{L}$ is parsable,
  \item\label{it4:thm:cf-def-v2} There exists a finite set
    $\alphabetTwo$ of edge labels, a definable
    $(\relsigtree{}(\alphabetTwo),~\relsiggraph{\emptyset})$-transduction
    $F$ and a definable
    $(\relsiggraph{\emptyset},~\relsigtree{}(\alphabetTwo))$-transduction
    $H$, such that~
    (a) $\dom{F \circ H} = \gstruc{\mathcal{L}}$, and
    (b) $F \circ H$ is the identity on $\gstruc{\mathcal{L}}$.
  \end{enumerate}
\end{thm}
\begin{proof}
  (\ref{it1:thm:cf-def-v2}) $\Rightarrow$ (\ref{it2:thm:cf-def-v2}) As
  every grammar uses only $\hrsignature^\slabs$-operations for the set
  of sources $\slabs \finsubseteq \sourcelabels$ that appear in the
  grammar, we obtain that $\mathcal{L}$ is recognizable in $\algebra$
  and $\mathcal{L} \subseteq \btwgraphsof{\slabs}$ from
  Theorem~\ref{thm:cf-def} (\ref{it1:thm:cf-def} $\Rightarrow$
  \ref{it2:thm:cf-def}). Then, $\twd{\graph} \leq \cardof{\slabs} - 1$,
  for every graph $\graph \in \btwgraphsof{\slabs}$, by
  Lemma~\ref{lemma:cf-twb}.

  \vspace*{\baselineskip}
  \noindent(\ref{it2:thm:cf-def-v2}) $\Rightarrow$
  (\ref{it3:thm:cf-def-v2}) Let $k \ge 1$ be the least integer such
  that $\twd{\graph} \le k$, for all graphs $\graph \in \mathcal{L}$,
  and let $\graph$ be such a graph.  Since $\mathcal{L} \subseteq
  \graphsof{\emptyset}$, the graph $\graph$ does not have
  sources. Hence,
  condition~(\ref{it1:mapping-bounded-tree-width-graphs-back-to-dervation-trees})
  of
  Corollary~\ref{cor:mapping-bounded-tree-width-graphs-back-to-dervation-trees}
  is satisfied, for any sort $\slabs \finsubseteq \sourcelabels$ with
  $k \le \cardof{\slabs}-1$, hence $\gstruc{\graph} \in
  \dom{\mathcal{K}}$, where $\mathcal{K}$ is the transduction whose
  existence is stated by
  Corollary~\ref{cor:mapping-bounded-tree-width-graphs-back-to-dervation-trees}.
  Then, there exists a tree $\tree\in\trees(\alphabetParseOf{\slabs})$
  such that $(\gstruc{\graph},\gstruc{\tree}) \in \mathcal{K}$ such
  that $\hval(\tree) = \graph$, by
  condition~(\ref{it2:mapping-bounded-tree-width-graphs-back-to-dervation-trees})
  of
  Corollary~\ref{cor:mapping-bounded-tree-width-graphs-back-to-dervation-trees}.
  Hence, $\mathcal{L} \subseteq \btwgraphsof{\slabs}$ and, by
  Theorem~\ref{thm:cf-def}, $\mathcal{L}$ is $\slabs$-parsable.

  \vspace*{\baselineskip}
  \noindent(\ref{it3:thm:cf-def-v2}) $\Rightarrow$
  (\ref{it4:thm:cf-def-v2}) We have that $\mathcal{L}$ is
  $\slabs$-parsable for some sort $\slabs \finsubseteq
  \sourcelabels$. Let the alphabet be $\alphabetTwo \isdef
  \alphabetParseOf{\slabs}$.  Then, there exists a definable
  $(\relsiggraph{\slabs},~\relsigtree{}(\alphabetTwo))$-transduction
  $H$ that witnesses the $\slabs$-parsability of $\mathcal{L}$.  Since
  $\mathcal{L} \subseteq \graphsof{\emptyset}$, we can assume
  w.l.o.g. that $H$ is a
  $(\relsiggraph{\emptyset},~\relsigtree{}(\alphabetTwo))$-transduction
  (using a transduction that simply removes the predicates that encode the sources).
  Moreover, $F \isdef \proj{\hval}{\trees(\alphabetTwo)}$ is
  definable, by Lemma~\ref{lemma:evaluation-homomorphism-transduction}
  and $F \circ H$ is the identity on $\gstruc{\mathcal{L}}$, as
  required.

  \vspace*{\baselineskip}
  \noindent(\ref{it4:thm:cf-def-v2}) $\Rightarrow$
  (\ref{it1:thm:cf-def-v2}) By Proposition~\ref{prop:comp-restr-trans}
  (\ref{it1:comp-restr-trans}), $F \circ H$ is a definable
  transduction, hence $\dom{F \circ H} = \gstruc{\mathcal{L}}$ is
  definable, by Proposition~\ref{prop:comp-restr-trans}
  (\ref{it3:comp-restr-trans}). Then, $\mathcal{K} \isdef
  F^{-1}(\mathcal{L})$ is definable, by Theorem~\ref{thm:back-trans}.
  By Theorem~\ref{thm:unranked-trees-rec-def}, $\mathcal{K}$ is
  recognizable in $\treealgebra(\alphabetTwo)$ and $\mathcal{L}$ is
  context-free, by Corollary \ref{cor:cf-unranked}.
\end{proof}
We note that item (\ref{it4:thm:cf-def-v2}) is missing from
Theorem~\ref{thm:cf-def} because there is no easy way of computing an
upper bound on the tree-width of $\mathcal{L}$. That is, we identify
the following problem for future work: given \mso-definable
transductions $F$ and $H$ as stated in item~\ref{it4:thm:cf-def-v2} of
Theorem~\ref{thm:cf-def-v2}, compute a bound on the tree-width of
$\mathcal{L}$ based on $F$ and $H$. We remark that the construction
of~\cite{CourcelleEngelfriet95} can be used to derive an upper bound,
but that this bound is likely not the optimal one. On the other hand,
item (\ref{it2a:thm:cf-def}) of Theorem~\ref{thm:cf-def} could be
added to Theorem~\ref{thm:cf-def-v2} (we omit it for conciseness
reasons).

We further note that the problem of whether one of the conditions from
Theorem~\ref{thm:cf-def} (resp.\  Theorem~\ref{thm:cf-def-v2}) holds for
the set of graphs generated by a given graph grammar is
undecidable.
In fact, even the problem of whether a given context-free
word grammar defines a recognizable (and hence definable) word
language is undecidable, according to a result by
Greibach~\cite{Greibach68}.

The impossibility of having an algorithm that decides whether a given
grammar produces a definable set of graphs motivated a number of
definitions that provide sufficient conditions.
Among them proposals for regular graph grammars~\cite{CourcelleV,LPAR2024:Tree_Verifiable_Graph_Grammars,LICS2025:Reg_Grammars_TW_2} over and regular expressions for graphs of tree-width at most~2~\cite{DBLP:conf/icalp/Doumane22} (for more details see the related work section in the beginning of this paper).
We highlight that our characterization provides new tools for the
definition of such sets.
For instance, the definable transductions from Example \ref{ex:tll} fit point (\ref{it4:thm:cf-def-v2}) of Theorem \ref{thm:cf-def-v2}, meaning that adding edges between the lexicographical successors of the leaves in a binary tree taken from a recognizable set produces a context-free and
definable set of graphs (i.e., trees with linked leaves).

\section{Finite versus Locally Finite Recognizability of Graph Sets}
\label{sec:recognizability}

Theorem \ref{thm:cf-def} (\ref{it2:thm:cf-def}-\ref{it2a:thm:cf-def})
proves the equivalence between the recognizability via locally finite
algebras and recognizability via finite algebras, for tree-width
bounded sets of graphs. This is the case because a set $\mathcal{L}
\subseteq \btwgraphsof{\slabs}$ of graphs is recognizable in
$\algebra$ iff it is the inverse image of a homomorphism into a
locally finite recognizer algebra having infinitely many (possibly
non-empty) sorts (\autoref{def:rec}), whereas recognizability of
$\mathcal{L}$ in $\algebra^{\slabs}$ means that the locally finite
recognizer algebra has finitely many non-empty sorts, being thus
finite. Consequently, bounded tree-width sets of graphs can be
recognized using finite algebras, just like
terms\footnote{Recognizability by finite automata, as in the case of
words and terms, is one step further, because automata typically
require a canonical order in which the input is
traversed.}~\cite{Doner70}. This has been initially proved by
Courcelle and Lagergren \cite{courcelle_lagergren_1996}, using a
different argument.

In this section, we prove that locally finite recognizability for
graphs is the limit of recognizability in an infinite increasing
sequence of finite underapproximations (Theorem
\ref{thm:rec-limit}). This means that the equivalence between locally
finite and finite recognizability for bounded tree-width sets of graphs
(points (\ref{it2:thm:cf-def}) and (\ref{it2a:thm:cf-def}) of Theorem
\ref{thm:cf-def}) is actually a cut-off in this infinite increasing
sequence.

The following lemma is an equivalent characterization of the syntactic
congruence (\autoref{def:syntactic-congruence}), that uses only terms
of a restricted form:

\begin{lem}
\label{lem:congr-graphs-characterization}
Let $\mathcal{L} \subseteq \graphsof{}$ be a set of graphs.
Then, $\graph_1 \congrec{\algebra}{\mathcal{L}} \graph_2$
iff $\typeof{\graph_1} = \typeof{\graph_2}$ and
$\rename{\alpha}^\algebra \circ \restrict{\slabs}^\algebra(\graph_1
\pop \graph) \in \mathcal{L} \Leftrightarrow \rename{\alpha}^\algebra
\circ \restrict{\slabs}^\algebra(\graph_2 \pop \graph) \in
\mathcal{L}$, for all graphs $\graph, \graph_1, \graph_2$, finite
permutations $\alpha : \sourcelabels \rightarrow \sourcelabels$ and
sorts $\slabs\finsubseteq\sourcelabels$.
\end{lem}
\begin{proof}
We define the equivalence relation $\requiv{}$ by setting $\graph_1
\requiv{} \graph_2$ iff $\typeof{\graph_1} = \typeof{\graph_2}$ and
$\rename{\alpha}^\algebra \circ \restrict{\slabs}^\algebra(\graph_1
\pop \graph) \in \mathcal{L} \Leftrightarrow \rename{\alpha}^\algebra
\circ \restrict{\slabs}^\algebra(\graph_2 \pop \graph) \in
\mathcal{L}$, for all graphs $\graph, \graph_1, \graph_2$, finite
permuations $\alpha : \sourcelabels \rightarrow \sourcelabels$ and
sets of sources $\slabs\finsubseteq\sourcelabels$.
We now consider some congruence $\equiv$ that saturates $\mathcal{L}$.
Then,
$\graph_1 \equiv \graph_2$ implies that $\graph_1 \requiv{} \graph_2$,
i.e., $\equiv \ \subseteq \ \requiv{}$ (*).
This is because, if $\graph_1 \equiv \graph_2$ then
$\rename{\alpha}^\algebra \circ \restrict{\slabs}^\algebra(\graph_1
\pop \graph) \equiv \rename{\alpha}^\algebra \circ
\restrict{\slabs}^\algebra(\graph_2 \pop \graph)$ (as
$\equiv$ is some congruence), and hence $\rename{\alpha}^\algebra \circ
\restrict{\slabs}^\algebra(\graph_1 \pop \graph) \in \mathcal{L}
\Leftrightarrow \rename{\alpha}^\algebra \circ
\restrict{\slabs}^\algebra(\graph_2 \pop \graph) \in \mathcal{L}$ (as
$\equiv$  saturates
$\mathcal{L}$).

We will establish below that $\requiv{}$ saturates $\mathcal{L}$. By
definition of $\requiv{}$, we have that $\graph_1 \requiv{} \graph_2$
implies that $\graph_1 \in \mathcal{L}\iff\graph_2 \in \mathcal{L}$,
since we can choose $\alpha$ as the identity, $\slabs =
\typeof{\graph_1}$ and $\graph = \emptygraph_{\typeof{\graph_1}}$. By
(*) we obtain that $\requiv{}$ is the coarsest relation that saturates
$\mathcal{L}$. It remains to establish that $\requiv{}$ is a
congruence. Let us consider some graphs $\graph_1 \requiv{} \graph_2$.
We show the closure under the operations of the graph algebra, by a
case distinction: \begin{itemize}[label=$\triangleright$]
\item $\graph_1 \pop^\algebra \graph \requiv{} \graph_2 \pop^\algebra
  \graph$ and $\graph \pop^\algebra \graph_1 \requiv{} \graph
  \pop^\algebra \graph_2$, for all graphs $\graph$: By the
  commutativity of $\pop^\algebra$, it is sufficient to show one of
  the implications. Let us assume $\graph_1 \requiv{} \graph_2$ and
  let $\graph$ be some graph. To establish that $\graph_1
  \pop^\algebra \graph \requiv{} \graph_2 \pop^\algebra \graph$,
  consider some graph $\graph'$, bijective function $\alpha$ and
  $\slabs \finsubseteq\sourcelabels$ such that
  $\rename{\alpha}^\algebra\circ\restrict{\slabs}^\algebra((\graph_1
  \pop^\algebra \graph) \pop^\algebra \graph') \in \mathcal{L}$. Then,
  $\rename{\alpha}^\algebra\circ\restrict{\slabs}^\algebra((\graph_2
  \pop^\algebra \graph) \pop^\algebra \graph') \in \mathcal{L}$
  follows from $\graph_1 \requiv{} \graph_2$, by the definition of
  $\requiv{}$, using the associativity of $\pop^\algebra$ and chosing
  the graph $\graph'' = \graph \pop \graph'$.
\item $\restrict{\slabs^\circ}(\graph_1) \requiv{}
  \restrict{\slabs^\circ}(\graph_2)$, for all $\slabs^\circ
  \finsubseteq \sourcelabels$: We consider some graph $\graph$, finite
  permutation $\alpha$ and $\slabs \finsubseteq\sourcelabels$ such
  that
  $\rename{\alpha}^\algebra\circ\restrict{\slabs}^\algebra(\restrict{\slabs^\circ}(\graph_1)
  \pop^\algebra \graph) \in \mathcal{L}$. We need to show that
  $\rename{\alpha}^\algebra\circ\restrict{\slabs}^\algebra(\restrict{\slabs^\circ}(\graph_2)
  \pop^\algebra \graph) \in \mathcal{L}$. We now verify that we can
  choose $\graph'$, $\alpha'$ and $\slabs'$ such that
  $\rename{\alpha'}\circ\restrict{\slabs'}(\graph_i \pop^\algebra
  \graph') =
  \rename{\alpha}^\algebra\circ\restrict{\slabs}^\algebra(\restrict{\slabs^\circ}(\graph_i)
  \pop^\algebra \graph)$ for $i=1,2$. Indeed, we can choose $\graph' =
  \rename{\beta}(\graph)$, for some permutation $\beta$ that renames
  the sources $\typeof{\graph} \setminus \slabs^\circ$ to some fresh
  sources $\slabs''$, choosing $\slabs' = \slabs \cup \beta(\slabs)$
  and setting $\alpha'$ as the permutation that does all renamings of
  $\alpha$ and $\beta^{-1}$.  The claim then follows from $\graph_1
  \requiv{} \graph_2$.
\item $\rename{\beta}(\graph_1) \requiv{} \rename{\beta}(\graph_2)$,
  for all finite permutations $\beta$: We consider a graph $\graph$,
  finite permutation $\alpha$ and $\slabs \finsubseteq\sourcelabels$,
  such that
  $\rename{\alpha}^\algebra\circ\restrict{\slabs}^\algebra(\rename{\beta}(\graph_1)
  \pop^\algebra \graph) \in \mathcal{L}$. We need to show that
  $\rename{\alpha}^\algebra\circ\restrict{\slabs}^\algebra(\rename{\beta}(\graph_2)
  \pop^\algebra \graph) \in \mathcal{L}$. We now observe that
  $\rename{\alpha \circ
    \beta}\circ\restrict{\beta^{-1}(\slabs)}(\graph_i \pop^\algebra
  \rename{\beta^{-1}}(\graph)) =
  \rename{\alpha}^\algebra\circ\restrict{\slabs}^\algebra(\rename{\beta}(\graph_i)
  \pop^\algebra \graph)$ for $i=1,2$.  The claim follows from
  $\graph_1 \requiv{} \graph_2$. \qedhere
\end{itemize}
\end{proof}

Specializing Lemma~\ref{lem:congr-graphs-characterization} to graphs of empty sort, we obtain the following result that has appeared in previous work:

\begin{corC}[{\cite[Theorem 4.34]{courcelle_engelfriet_2012}}]
\label{cor:congr-graphs-characterization-simp}
Let $\mathcal{L} \subseteq \graphsof{\emptyset}$ be a set of graphs.
Then, $\graph_1 \congrec{\algebra}{\mathcal{L}} \graph_2$
iff $\typeof{\graph_1} = \typeof{\graph_2}$ and
$\restrict{\slabs}^\algebra(\graph_1
\pop \graph) \in \mathcal{L} \Leftrightarrow \restrict{\slabs}^\algebra(\graph_2 \pop \graph) \in
\mathcal{L}$, for all graphs $\graph, \graph_1, \graph_2$ and
sorts $\slabs\finsubseteq\sourcelabels$.
\end{corC}
\begin{proof}
Follows from Lemma~\ref{lem:congr-graphs-characterization} for some $\mathcal{L} \subseteq \graphsof{\emptyset}$ by the observation that for any graph $\graph$ we have $\graph \in \mathcal{L}$ iff $\rename{\alpha}^\algebra(\graph) \in \mathcal{L}$ for any finite
permutation $\alpha$.
\end{proof}
The next step is proving that the syntactic congruences of a language
of graphs of empty sort agree over the algebras $\algebra$ and
$\algebra^\slabs$, for any sort $\slabs \finsubseteq \sourcelabels$:

\begin{lem}\label{lemma:agreement}
  Let $\mathcal{L} \subseteq \graphsof{\emptyset}$ be a language, $\slabs\finsubseteq\sourcelabels$ be a sort, and $\graph_1, \graphs_2 \in \graphsof{\slabs}$ be some graphs.
  Then, $\graph_1 \congrec{\algebra}{\mathcal{L}} \graph_2$ iff $\graph_1
  \congrec{\algebra^{\slabs}}{\mathcal{L}} \graph_2$.
\end{lem}
\begin{proof}
  ``$\Rightarrow$''
  Because $\congrec{\algebra}{\mathcal{L}} ~\cap~ (\graphsof{\slabs} \times \graphsof{\slabs})$ is a congruence that saturates $\mathcal{L}$ w.r.t the algebra $\algebra^{\slabs}$ and $\congrec{\algebra^{\slabs}}{\mathcal{L}}$ is the greatest such
  congruence.
  ``$\Leftarrow$''
  By Corollary~\ref{lem:congr-graphs-characterization}, we need to show that $\restrict{\slabs'}^\algebra(\graph_1
  \pop^\algebra \graph) \in \mathcal{L}$ implies  $\restrict{\slabs'}^\algebra(\graph_2
  \pop^\algebra \graph) \in \mathcal{L}$.
  Because of $\graph_1,
  \graph_2 \in \graphsof{\slabs}$, it suffices to prove that
  $\restrict{\slabs' \cap \slabs}^{\algebra^{\slabs}}(\graph_1 \pop^\algebra  \restrict{\slabs' \setminus \slabs}(\graph))^{\algebra^{\slabs}} \in
  \mathcal{L}$ implies $\restrict{\slabs' \cap    \slabs}^{\algebra^{\slabs}}(\graph_2 \pop^\algebra \restrict{\slabs' \setminus \slabs }^{\algebra^{\slabs}}(\graph)) \in \mathcal{L}$.
  However, this follows from $\graph_1  \requiv{\mathcal{L}}^{\algebra^{\slabs}} \graph_2$.
\end{proof}

Finally, we relate recognizability of a set of graphs of empty sort in
the graph algebra $\algebra$ and any of its subalgebras
$\algebra^\slabs$:

\begin{thm}\label{thm:rec-limit}
Let $\mathcal{L}$ be a set of graphs with no sources. Then,
$\mathcal{L}$ is recognizable in the graph algebra $\algebra$ iff for
each sort $\slabs \finsubseteq \sourcelabels$, the set $\mathcal{L}$
is recognizable in the algebra $\algebra^{\slabs}$.
\end{thm}
\begin{proof}
By Lemma~\ref{lemma:agreement}, we have that $\congrec{\algebra}{\mathcal{L}} ~\cap~ (\graphsof{\slabs} \times \graphsof{\slabs}) = \ \congrec{\algebra^{\slabs}}{\mathcal{L}}$ for every $\slabs \finsubseteq \sourcelabels$.
In particular, $\congrec{\algebra}{\mathcal{L}}$ is locally finite iff $\congrec{\algebra^{\slabs}}{\mathcal{L}}$ is locally finite for every $\slabs \finsubseteq \sourcelabels$.
\end{proof}
Note that the condition of Theorem \ref{thm:rec-limit} is different
from the \emph{weak recognizability} notion introduced by
Boja\'{n}czyk~\cite{journals/corr/abs-2305-18039}, that considers an
infinite sequence of finitely-generated \emph{term-generated}
algebras.  Weak recognizability and recognizability are in fact not
the same, as~\cite[Example 11]{journals/corr/abs-2305-18039}
shows. This points to the crucial difference between term-generated
and non term-generated algebras (Lemma \ref{lemma:cf-twb}).

\section{Conclusions}

We have given a characterization of definable context-free sets of
graphs, by showing their equivalence with bounded tree-width and
recognizable sets (where recognizability is understood either in
locally finite or finite algebras), parsable sets (where the parse
trees can be recovered from the graph by a definable transduction) and
images of recognizable unranked sets of trees under definable
transductions whose inverses are definable as well. We finalize our
study with a discussion on recognizability and a proof that locally
finite recognizer algebras are limits of infinite sequences of finite
recognizer algebras.

\bibliographystyle{alphaurl}
\bibliography{refs}

\end{document}